%% file: main.tex
\documentclass[sigconf]{acmart}
\input{localdefs.tex}

\AtBeginDocument{%
  \providecommand\BibTeX{{%
    \normalfont B\kern-0.5em{\scshape i\kern-0.25em b}\kern-0.8em\TeX}}}

\setcopyright{acmcopyright}
\copyrightyear{2020}
\acmYear{2020}
\acmDOI{10.1145/1122445.1122456}

\begin{document}
\pagestyle{plain}

\title{The Impact of Negation on the Complexity of the \\ Shapley Value in Conjunctive Queries}


\author{Alon Reshef}
\affiliation{
\institution{Technion}
\city{Haifa}
\country{Israel}
\postcode{32000}
}
\email{alonre@cs.technion.ac.il}

\author{Benny Kimelfeld}
\affiliation{
\institution{Technion}
\city{Haifa}
\country{Israel}
\postcode{32000}
}
\email{bennyk@cs.technion.ac.il}

\author{Ester Livshits}
\affiliation{
\institution{Technion}
\city{Haifa}
\country{Israel}
\postcode{32000}
}
\email{esterliv@cs.technion.ac.il}

\renewcommand{\shortauthors}{Reshef et al.}

\begin{abstract}
  The Shapley value is a conventional and well-studied function for
  determining the contribution of a player to the coalition in a
  cooperative game. Among its applications in a plethora of domains,
  it has recently been proposed to use the Shapley value for
  quantifying the contribution of a tuple to the result of a database
  query. In particular, we have a thorough understanding of the
  tractability frontier for the class of Conjunctive Queries (CQs) and
  aggregate functions over CQs. It has also been established that a
  tractable (randomized) multiplicative approximation exists for every
  union of CQs. Nevertheless, all of these results are based on the
  monotonicity of CQs. In this work, we investigate the implication of
  negation on the complexity of Shapley computation, in both the exact
  and approximate senses. We generalize a known dichotomy to account
  for negated atoms. We also show that negation fundamentally changes
  the complexity of approximation. We do so by drawing a connection to
  the problem of deciding whether a tuple is ``relevant'' to a query,
  and by analyzing its complexity.
\end{abstract}



\maketitle

\input{introduction.tex}

\input{preliminaries.tex}

\input{sjfBCQ.tex}


\input{exogenous.tex}

\input{approximation.tex}

\balance

\input{conclusions.tex}

\bibliographystyle{ACM-Reference-Format}
\bibliography{main}

\newpage
\appendix

\input{Appendix.tex}

\end{document}

%% file: localdefs.tex
\usepackage{multirow}
\usepackage{color,soul}
\usepackage{array}
\usepackage{balance}
\usepackage{paralist}
\usepackage{enumitem}
\usepackage{url}
\usepackage{xspace}
\usepackage[ruled,vlined,noend]{algorithm2e}
\usepackage{subcaption}
\usepackage{xcolor}

\def\e#1{\emph{#1}}
\newcommand{\eat}[1]{}

\def\cqrst{\mathsf{q_{\mathsf{RST}}}}

\def\cqrsnt{\mathsf{q_{\mathsf{RS\neg T}}}}
\def\cqrnst{\mathsf{q_{\mathsf{R\neg ST}}}}
\def\cqnrsnt{\mathsf{q_{\mathsf{\neg RS\neg T}}}}
\def\cqtrsnr{\mathsf{q_{\mathsf{RST\neg R}}}}
\def\is{\mathsf{IS}}
\def\ds{\mathsf{NIS}}
\def\s{\mathsf{S}}

\def\dl{\mathrel{{:}{\text{-}}}}

\def\set#1{\mathord{\{#1\}}}

\newenvironment{repeatresult}[2]
{\vskip0.5em\par\textsc{#1} #2.\em}
{\vskip1em}

\newenvironment{repproposition}[1]{\begin{repeatresult}{Proposition}{#1}}{\end{repeatresult}}
\newenvironment{reptheorem}[1]{\begin{repeatresult}{Theorem}{#1}}{\end{repeatresult}}
\newenvironment{replemma}[1]{\begin{repeatresult}{Lemma}{#1}}{\end{repeatresult}}

\def\P{\mathsf{P}}

\def\shp{\mathrm{Shapley}}

\def\FP{\mathrm{FP}^{\#\mathrm{P}}}

\newcommand{\noamdone}[1]{}

\def\val#1{\mathtt{#1}}
\def\is{\mathsf{IS}}
\def\ds{\mathsf{NIS}}

\theoremstyle{plain}
\newtheorem{theorem}{Theorem}[section]
\newtheorem{proposition}[theorem]{Proposition}
\newtheorem{lemma}[theorem]{Lemma}
\newtheorem{corollary}[theorem]{Corollary}

\newcommand{\algname}[1]{{\sf #1}}

\def\signature{\mathcal{S}}
\def\graph{\mathcal{G}}
\def\exograph{g\exo}
\def\arity{\mathord{\mathrm{arity}}}
\def\consts{\mathsf{Const}}
\def\posq{\mathsf{Pos}}
\def\negq{\mathsf{Neg}}
\def\att#1{\mathrm{#1}}

\def\exo{_{\mathsf{x}}}
\def\nexo{_{\setminus\mathsf{x}}}
\def\nexoC{{\setminus\mathsf{x}}}
\def\endo{_{\mathsf{n}}}
\def\pall{\mathord{\mathcal{P}}}
\def\Dom{\mathsf{Dom}}

\usepackage[noend]{algpseudocode}
\usepackage{xspace}
\usepackage{paralist}
\usepackage{color, colortbl}

\def\e#1{\emph{#1}}

\def\pall{\mathord{\mathcal{P}}}

\def\exo{_{\mathsf{x}}}
\def\endo{_{\mathsf{n}}}

\def\set#1{\mathord{\{#1\}}}

\def\dl{\mathrel{{:}{\text{-}}}}

\def\is{\mathsf{IS}}

\def\ds{\mathsf{NIS}}

\def\cqrst{\mathsf{q_{\mathsf{RST}}}}
\def\cqrnotst{\mathsf{q_{\mathsf{R\neg ST}}}}
\def\cqnotrsnott{\mathsf{q_{\mathsf{\neg RS \neg T}}}}
\def\cqrsnott{\mathsf{q_{\mathsf{RS \neg T}}}}
\def\cqsat{\mathsf{q_{\mathsf{SAT}}}}

\def\var{\mathsf{Vars}}
\def\at{\mathsf{Atoms}}
\def\cc{\mathsf{ConnectedComponents}}

\def\Pr#1{\mathord{\mathrm{Pr}\left[#1\right]}}

\def\val#1{\mathtt{#1}}
\def\rel#1{\textsc{#1}}
\def\att#1{\textit{#1}}

\newenvironment{proofsketch}{\begin{proof}{\textit{(Sketch)}\,\,}}
{\end{proof}}

\def\taf{f^{\mathrm{t}}}
\def\regf{f^{\mathrm{r}}}

\def\fpsharpp{\mathrm{FP}^{\mathrm{\#P}}}

\def\stf{f^{\mathrm{s}}}
\def\taf{f^{\mathrm{t}}}
\def\cof{f^{\mathrm{c}}}
\def\regf{f^{\mathrm{r}}}
\def\adf{f^{\mathrm{a}}}

\def\Sat{\mathrm{Sat}}

\def\CQneg{CQ$^\neg$\xspace}
\def\UCQneg{UCQ$^\neg$\xspace}

%% file: introduction.tex
\section{Introduction}
Various formal measures have been proposed for quantifying the
contribution of a fact $f$ to a query answer.  Meliou et
al.~\cite{DBLP:journals/pvldb/MeliouGMS11} adopted the quantity of
\e{responsibility} that is inversely proportional to the minimal
number of endogenous facts that should be removed to make $f$
counterfactual (i.e., removing $f$ transitions the answer from true to
false). Following earlier notions of formal causality by Halpern and
Pearl~\cite{DBLP:conf/uai/HalpernP01}, Salimi et
al.~\cite{DBLP:conf/tapp/SalimiBSB16} proposed the \e{causal effect}:
assuming endogenous facts are randomly removed independently and
uniformly, what is the difference in the expected query answer between
assuming the presence and the absence of $f$?  A recent framework has
proposed to adopt the \e{Shapley value} to the
task~\cite{shapley-icdt2020}.

The Shapley value~\cite{shapley:book1952} is a formula for wealth
distribution in a cooperative game, and it has been applied in a
plethora of domains that require to attribute a share of an outcome
among a group of
entities~\cite{aumann2003endogenous,DBLP:conf/stacs/AzizK14}.  The use
cases include bargaining foundations in
economics~\cite{gul1989bargaining}, takeover corporate rights in
law~\cite{nenova2003value}, pollution responsibility in environmental
management~\cite{petrosjan2003time,liao2015case}, influence
measurement in social network analysis~\cite{narayanam2011shapley},
the utilization of Internet Service Providers (ISPs) in
networks~\cite{ma2010internet}, and advertisement effectiveness on the
Web~\cite{DBLP:conf/www/BesbesDGIS19}.  In data management, the
Shapley value has been used for assigning a level of inconsistency
to facts in inconsistent knowledge
bases~\cite{DBLP:journals/ai/HunterK10,DBLP:conf/ijcai/YunVCB18,DBLP:journals/jiis/GrantH06},
and to determine the relative contribution of features in
machine-learning
predictions~\cite{DBLP:conf/ijcai/LabreucheF18,DBLP:conf/nips/LundbergL17}.

In the framework of Livshits et al.~\cite{shapley-icdt2020}, query
answering is viewed as a cooperative game where the players are the
database facts and the utility function is the query answer, in the
case of aggregate queries, or 0/1 in the case of Boolean
queries. They showed that the (precise or approximate) evaluation of the
Shapley value on common aggregate queries amounts to the evaluation of
the Shapley value for a Boolean query. Hence, in this paper, we focus on Boolean
queries.  To illustrate the Shapley
value on a Boolean Conjunctive Query (which we refer to simply as a
\e{CQ} hereafter), consider the following query asking whether
there is a farmer $m$ who exports a product $p$ to a country $c$ where
$p$ does not grow.
\begin{align}\label{ex:farmer}
  q() \dl \rel{Farmer}(m),\rel{Export}(m,p,c), \neg\rel{Grows}(c,p)
\end{align}
Different $\rel{Farmer}$ facts may have different Shapley values,
depending on how crucial they are to the query---which products they
export, whether they grow in the destination countries, and whether
alternative $\rel{Farmer}$ facts export the same products.
Similarly, each $\rel{Grows}(c,p)$ fact has its Shapley
value. However, while the Shapley value of $\rel{Farmer}$ facts can be
either positive or zero (since they can only help in satisfying the
query), the Shapley value of $\rel{Grows}$ facts can be either
\e{negative} or zero (since they can only help in violating the
query). As explained by Livshits et al.~\cite{shapley-icdt2020},
understanding the complexity of the Shapley value for Boolean queries such as~\eqref{ex:farmer} is also necessary and sufficient for understating the complexity of the Shapley value for aggregate queries such as
\begin{align*}\label{ex:farmer-count}
  \mathsf{Count}
  \{c\mid\rel{Farmer}(m),\rel{Export}(m,p,c), \neg\rel{Grows}(c,p)\}
\end{align*}
that counts the countries that import one or more products that they
do not grow.

As in previous work on quantification of contribution of
facts~\cite{DBLP:journals/pvldb/MeliouGMS11,DBLP:conf/tapp/SalimiBSB16,
  shapley-icdt2020}, we view the database as consisting of two types
of facts: \e{endogenous} facts and \e{exogenous} facts. Exogenous
facts are taken as given (e.g., inherited from external sources)
without questioning, and are beyond experimentation with hypothetical
or counterfactual scenarios. On the other hand, we may have control
over the endogenous facts, and these are the facts for which we reason
about existence and marginal contribution. The exogenous and
endogenous facts are analogous to the \e{observations} and
\e{hypotheses} in the study of abductive
diagnosis~\cite{DBLP:journals/tcs/EiterGL97,DBLP:journals/jacm/EiterG95}
that we refer to later on. In our context, the Shapley value considers
only the endogenous facts as players in the cooperative game.

Livshits et al.~\cite{shapley-icdt2020} have studied the complexity of
computing the Shapley value, and their work is restricted to positive
CQs and UCQs (and aggregates thereof). In this paper, we study the
impact of negation on this complexity. Negation transforms the query
into a non-monotonic query and, as the reader might expect, the impact
is fundamental. As a first step, we generalize their dichotomy in the
complexity for CQs without self-joins into the class of CQs with
negation and without self-joins (Theorem~\ref{thm:sjfBCQ}).

The dichotomy of Livshits et al.~\cite{shapley-icdt2020} classifies
the CQs precisely as in CQ inference in probabilistic
tuple-independent databases~\cite{DBLP:conf/vldb/DalviS04}: if the CQ
is hierarchical, then the problem is solvable in polynomial time, and
otherwise, it is $\fpsharpp$-complete (i.e., complete for the
intractable class of polynomial-time algorithms with an oracle to,
e.g., a counter of the satisfying assignments of a propositional
formula). For illustration, the CQ of~\eqref{ex:farmer} falls on the
hardness side. However, that classification does not take into account
the assumption that some relations may contain \e{only} exogenous
data. For example, in~\eqref{ex:farmer} we might consider the
$\rel{Grows}$ relation as consisting of only exogenous information.
This assumption is very significant, as it makes our example CQ a
tractable one for the Shapley value, \e{in contrast to the
  dichotomy}. In this paper, we establish a
dichotomy that accounts for both negation and exogenous relations
(Theorem~\ref{thm:exo}).

An \e{approximation} of the Shapley value of a database fact $f$ to a
Boolean query can be computed via a straightforward Monte-Carlo
(average-over-samples) estimation of the expectation that Shapley
defines. This estimation guarantees an \e{additive} (or \e{absolute})
approximation. However, our interest is in a \e{multiplicative} (or
\e{relative}) approximation, for two main reasons. First, we seek to
understand the contribution of $f$ \e{relative to other
  facts}, even if it is the case that the Shapley value is small.
Second, in order to get an approximation of the contribution of a
fact to an \e{aggregate query}, a multiplicative approximation is
required~\cite{shapley-icdt2020}.

In the case of a UCQ, a multiplicative approximation of the Shapley
value is tractable, that is, there is a multiplicative Fully
Polynomial-Time Approximation Scheme (FPRAS). This holds true for a
simple reason: an additive FPRAS is also a multiplicative FPRAS, due
to the following \e{gap} property: if the Shapley value is nonzero,
then it must be ``large''---at least the reciprocal of a polynomial.
Nevertheless, once the CQ includes negated atoms, the gap property is
no longer true. In Fact, we show in Theorem~\ref{thm:negation-no-gap}
that \e{every} natural CQ with negation violates the gap
property, since the Shapley value can be exponentially small. This
phenomenon explains why \e{negated atoms} make the Shapley value
fundamentally more challenging to approximate.

In itself, the violation of the gap property shows that the approach
of an additive FPRAS fails to provide a multiplicative FPRAS. Yet, it
does not show that multiplicative FPRAS is computationally hard, since
there might be an alternative way of obtaining a multiplicative FPRAS
in polynomial time. In order to prove hardness of approximation, we
investigate the problem of determining whether a fact $f$ is
\e{relevant} to a query in the following sense: in the presence of all
exogenous facts and \e{some} subset of the endogenous facts, adding
$f$ can change the query answer (from false to true or from true to
false).  In the case of a positive CQ, being relevant to the query
coincides with being an ``actual cause'' in the framework of \e{causal
  responsibility}~\cite{DBLP:journals/pvldb/MeliouGMS11}. It is also
similar to being a \e{relevant hypothesis} in the context of
\e{abductive
  diagnosis}~\cite{DBLP:journals/ci/ConsoleT91,DBLP:journals/tcs/EiterGL97,DBLP:journals/jacm/EiterG95}.
We refer the reader to Bertossi and
Salimi~\cite{DBLP:journals/ijar/BertossiS17} who have established the
connection between causal responsibility and abductive diagnosis.

\eat{
use a connection to the study of program
abduction
and, specifically, that of the \e{relevance}
problem: determine whether a
given fact $f$ belongs to a subset-minimal \e{abductive diagnosis}---a
set of endogenous facts that, together with the exogenous facts,
satisfy the query. Put differently, the fact $f$ is relevant if for
some subset of the endogenous database, the inclusion of $f$ can
transform the query outcome from false to true. To generalize
relevance to non-monotonic queries, we use the following natural
adaptation: a fact $f$ is \e{relevant} if there is a subset of the
endogenous facts such that, together with the exogenous facts, the
addition of $f$ changes the outcome of the query (from true to false
or from false to true).
}

The connection between the relevance to the query and the Shapley
value is direct: if a fact $f$ is \e{polarity consistent} in the sense
that it occurs in a relation of only positive or only negative atoms,
then \e{$f$ is relevant if and only if its Shapley value is nonzero}
(i.e., strictly positive or strictly negative). Therefore, a
multiplicative FPRAS can decide on the relevance with high
probability. In the contrapositive, if we prove that the relevance to
the query is an intractable decision problem, then we also establish
the intractability of an FPRAS approximation. Yet, the relevance is
tractable for positive CQs, and hardness results are known only for
Datalog programs with
recursion~\cite{DBLP:journals/ijar/BertossiS17,DBLP:journals/tcs/EiterGL97}.
We prove here the existence of a CQ and a polarity-consistent fact $f$
such that the decision of relevance to the query (and, hence, the
multiplicative approximation of the Shapley value) is intractable.

Nevertheless, the above approach for proving hardness of the
multiplicative FPRAS of the Shapley value \e{fails} if we assume that
the CQ itself is polarity consistent, that is, every relation symbol
(and not just the one of $f$) occurs either only positively or only
negatively. We prove that the relevance problem is solvable in
polynomial time for polarity-consistent CQs. The question of whether
the Shapley value has a multiplicative FPRAS for polarity-consistent
CQs (and in particular CQs without self-joins) remains an open problem
for future investigation.

We also consider the relevance problem for UCQs with negation.  We
prove that the tractability of the relevance problem generalizes to
polarity-consistent UCQs. Nevertheless, the tractability does not
generalize to unions of polarity-consistent CQs---we show the
existence of such a UCQ where the relevance problem is intractable,
and so is the Shapley zeroness (and multiplicative approximation). In
other words, if every relation symbol occurs either only positively or
only negatively in a UCQ, then the relevance problem is solvable in
polynomial time. Yet, the assumption that this consistency holds just
in every individual disjunct is (provably) not enough.

The rest of the paper is organized as follows. In the next section we
introduce some basic terminology that will be used throughout the
paper. In Section~\ref{sec:bcqn}, we study the complexity of computing
the Shapley value for self-join-free CQs with negation, and in Section~\ref{sec:exo} we
explore the impact of exogenous relations on this complexity.  We
consider the approximate computation of the value in
Section~\ref{sec:approx}.  We summarize our results and discuss
directions for future work in Section~\ref{sec:conclusions}. For lack
of space, some proofs apprear in the Appendix.

%% file: preliminaries.tex
\section{Preliminaries}\label{sec:preliminaries}
{
\definecolor{Gray}{gray}{0.9}
\def\emprow{\multicolumn{2}{l}{}}
\begin{figure*}[t]
\small
\centering
\begin{subfigure}[b]{0.16\linewidth}
\begin{tabular}{r|c|} 
\cline{1-2}
\rowcolor{black}
\multicolumn{2}{l}{\textcolor{white}{\rel{Stud}}}\\\cline{2-2}
 & $\att{name}$ \\\cline{2-2}
$\stf_1$ & $\val{Adam}$\\
$\stf_2$ & $\val{Ben}$\\
$\stf_3$ & $\val{Caroline}$\\
$\stf_4$ & $\val{David}$\\
\cline{2-2}
\emprow
\end{tabular}
\end{subfigure}
\begin{subfigure}[b]{0.13\linewidth}
\begin{tabular}{r|c|} 
\cline{1-2}
\rowcolor{black}
\multicolumn{2}{l}{\textcolor{white}{\rel{TA}}}\\\cline{2-2}
 & $\att{name}$ \\\cline{2-2}
$\taf_1$ & $\val{Adam}$\\
$\taf_2$ & $\val{Ben}$\\
$\taf_3$ & $\val{David}$\\
\cline{2-2}
\emprow\\
\emprow
\end{tabular}
\end{subfigure}
\begin{subfigure}[b]{0.2\linewidth}
\begin{tabular}{r|c|c|} 
\cline{1-3}
\rowcolor{black}
\multicolumn{3}{l}{\textcolor{white}{\rel{Course}}}\\\cline{2-3}
 & $\att{course}$ & $\att{faculty}$ \\\cline{2-3}
$\cof_1$ & $\val{OS}$ & $\val{EE}$\\
$\cof_2$ & $\val{IC}$ & $\val{EE}$\\
$\cof_3$ & $\val{DB}$ & $\val{CS}$\\
$\cof_4$ & $\val{AI}$ & $\val{CS}$\\
\cline{2-3}
\emprow
\end{tabular}
\end{subfigure}
\begin{subfigure}[b]{0.22\linewidth}
\begin{tabular}{r|c|c|} 
\cline{1-3}
\rowcolor{black}
\multicolumn{3}{l}{\textcolor{white}{\rel{Reg}}}\\\cline{2-3}
 & $\att{name}$ & $\att{course}$ \\\cline{2-3}
$\regf_1$ & $\val{Adam}$ & $\val{OS}$\\
$\regf_2$ & $\val{Adam}$ & $\val{AI}$\\
$\regf_3$ & $\val{Ben}$ & $\val{OS}$\\
$\regf_4$ & $\val{Caroline}$ & $\val{DB}$\\
$\regf_5$ & $\val{Caroline}$ & $\val{IC}$\\
\cline{2-3}
\end{tabular}
\end{subfigure}
\begin{subfigure}[b]{0.22\linewidth}
\begin{tabular}{r|c|c|} 
\cline{1-3}
\rowcolor{black}
\multicolumn{3}{l}{\textcolor{white}{\rel{Adv}}}\\\cline{2-3}
 & $\att{advisor}$ & $\att{student}$ \\\cline{2-3}
$\adf_1$ & $\val{Michael}$ & $\val{Adam}$\\
$\adf_2$ & $\val{Michael}$ & $\val{Ben}$\\
$\adf_3$ & $\val{Naomi}$ & $\val{Caroline}$\\
$\adf_4$ & $\val{Michael}$ & $\val{David}$\\
\cline{2-3}
\emprow
\end{tabular}
\end{subfigure}
\caption{\label{fig:DB} The database of the running example }
\end{figure*}
}

We first define the main concepts that we use throughout the paper.

\subsection*{Databases and Queries}
A \e{relational schema} $\signature$ is a finite collection of
relation symbols $R(A_1,\dots,A_k)$, where each $A_i$ is an
\e{attribute} of $R$, and $k$ is the \e{arity} of $R$, denoted by
$\arity(R)$. We assume a countably infinite set $\consts$ of constants
that are used as database values. A database $D$ (over a schema
$\signature$) associates with each relation symbol $R$ in $\signature$
a finite relation $R^D\subseteq \consts^{\arity(R)}$. If
$(c_1,\dots,c_k)$ is a tuple in $R^D$, then we refer to
$R(c_1,\dots,c_k)$ as a \e{fact} of $D$. We then identify a database
$D$ by the set of its facts. We assume that a database $D$ consists of
two disjoint subsets of facts: the set $D\exo$ of \e{exogenous} facts
and the set $D\endo$ \e{endogenous} facts. Hence, we have
$D=D\exo\cup D\endo$.

\begin{example} \label{example:DB} The database of our running example
  is depicted in Figure~\ref{fig:DB}.  The relations $\rel{Stud}$ and
  $\rel{TA}$ store the names of graduate students and teaching
  assistants in the university, respectively. The relation
  $\rel{Course}$ contains information about courses given in different
  faculties of the university. The relation $\rel{Reg}$ associates
  graduate students with the courses they take, and the relation
  $\rel{Adv}$ associates students with their academic advisor. For
  example, Adam is a student and a teaching
  assistant in the university. He is registered to two courses---OS is given in the
  Electrical Engineering faculty and AI in the
  Computer Science faculty. Michael is the academic advisor of Adam.
  \qed\end{example}


A \e{Boolean conjunctive query} over a schema $\signature$ is an
expression of the form: 
$$q()\dl R_1(\vec{t_1}),\dots, R_n(\vec{t_n})$$
where each $R_i$ is a relation symbol of $\signature$ and each
$\vec{t_i}$ is a tuple of variables and constants (where the arity of
$\vec{t_i}$ matches that of $R_i$). We refer to a Boolean conjunctive
query simply as a \e{CQ}. We refer to each $R_i(\vec{t_i})$ as an
\e{atom} of $q$. We denote by $R_\alpha$ the relation corresponding to the atom $\alpha$ of $q$. A self-join in a CQ $q$ is a pair of distinct atoms
of $q$ over the same relation symbol. If $q$ does not contain any
self-joins, then we say that $q$ is \e{self-join-free}.  A
homomorphism from $q$ to $D$ is a mapping of the variables in $q$ to
the constants of $D$ such that every atom in $q$ is mapped to a fact
of $D$. We denote by $D\models q$ the fact that $D$ satisfies $q$
(i.e., there is a homomorphism from $q$ to $D$) and by
$D\not\models q$ the fact that $D$ violates $q$ (i.e., there is no
such homomorphism).

Let $q$ be a CQ. For every variable $x$ of $q$, we denote by $A_x$ the
set of all atoms $R_i(\vec{t_i})$ of $q$ such that $x$ occurs in
$\vec{t_i}$. We say that $q$ is
\e{hierarchical}~\cite{DBLP:journals/cacm/DalviRS09} if at least one
of the following holds for all variables $x$ and $y$ of
$q$: \e{(1)} $A_x\subseteq A_y$, \e{(2)} $A_y\subseteq A_x$, or
\e{(3)} $A_x\cap A_y=\emptyset$. It is
known~\cite{DBLP:conf/vldb/DalviS04} that if $q$ is not hierarchical,
then there exist three atoms $\alpha_x$, $\alpha_y$, and
$\alpha_{x,y}$ in $q$ such that 
the variable $x$ occurs in $\alpha_x$
but not in $\alpha_y$, the variable $y$ occurs in $\alpha_y$ but not
in $\alpha_x$ and both variables occur in $\alpha_{x,y}$. We refer to
each such triplet of atoms as a \e{non-hierarchical triplet} of $q$.

A CQ with \e{safe negation}, or \CQneg for short, has the
form
\[q()\dl R_1(\Vec{t_1}), ... , R_n(\Vec{t_n}), \neg R'_1(\Vec{t'_1}),
  ... , \neg R'_m(\Vec{t'_m})\] where every variable that occurs in a
negated atom also occurs in an atom without negation. We refer to the
atoms of $q$ appearing without negation as the \e{positive atoms} of
$q$ and to the atoms that appear with negation as the \e{negative
  atoms} of $q$. We denote by $\posq(q)$ and $\negq(q)$ the sets of
positive and negative atoms of $q$, respectively. For a \CQneg, we
denote by $D\models q$ the fact that there is a homomorphism mapping
the variables of $q$ to constants of $D$ such that every positive atom
and none of the negative atoms of $q$ is mapped to a fact of $D$. The
extension to the definition of hierarchical CQs to \CQneg{s} is
straightforward.


\begin{example}\label{example:queries}
We use the following queries in our examples:
\begin{align*}
q_1()\dl &\rel{Stud}(x), \neg \rel{TA}(x), \rel{Reg}(x,y)\\
q_2()\dl &\rel{Stud}(x), \neg \rel{TA}(x), \rel{Reg}(x,y) ,\neg \rel{Course}(y, \val{CS})\\
q_3()\dl &\rel{Adv}(x,y), \rel{Adv}(x,z), \neg \rel{TA}(y), \neg \rel{TA}(z), \rel{Reg}(y,\val{IC}), \rel{Reg}(z,\val{DB})\\
q_4()\dl &\rel{Adv}(x,y), \rel{Adv}(x,z), \rel{TA}(y), \neg \rel{TA}(z), \rel{Reg}(z,w), \neg \rel{Reg}(y,w)
\end{align*}

Each of these queries is a \CQneg. The queries $q_1$ and $q_2$ are
self-join-free, while the queries $q_3$ and $q_4$ have self-joins (e.g., the relation
$\rel{Adv}$ occurs twice). The query $q_1$ is hierarchical since
$A_y\subseteq A_x$, but the others are not, since each of them
contains a non-hierarchical triplet (e.g.,
$\rel{Adv}(x,y), \rel{Adv}(x,z), \neg \rel{TA}(z)$). \qed
\end{example}

A \e{Union of Conjunctive Queries} (UCQ) is an expression of the form
$q()\dl q_1()\vee\dots\vee q_n()$ where each $q_i$
is a CQ, and it is satisfied by a database $D$ if $D\models q_i$ for
at least one $i\in\set{1,\dots,n}$. A union of \CQneg{s} is called 
a \UCQneg for short.

\subsection*{The Shapley Value} \label{def:shaply} Given a set $A$ of
players, a \e{cooperative game} is a function
$v:\pall(A)\rightarrow\mathbb{R}$ that maps every subset $B$ of $A$ to
a rational number $v(B)$, such that $v(\emptyset)=0$. The value $v(B)$
represents a value jointly obtained by the players of $B$ when they
cooperate. The Shapley value~\cite{shapley:book1952} measures the
share of each player $a\in A$ in the value $v(A)$ jointly obtained by
all players. Intuitively, the Shapley value is the expected
contribution of $a$ in a random permutation of the players, where the
contribution of $a$ is the change of $v$ due to the addition of $a$. More
formally, the Shapley value is defined as
$$\shp(A,v,a):= \dfrac{1}{|A|!}\sum_{\sigma \in \Pi_A} (v(\sigma_a\cup \{a\})-v(\sigma_a))$$
where $\Pi_A$ is the set of all possible permutations over the players
in $A$, and for each permutation $\sigma$, we denote by $\sigma_a$ the
set of players that appear before $a$ in the permutation.

Let $\signature$ be a schema, $D$ a database over $\signature$, $q$ a
CQ or \CQneg, and $f$ an endogenous fact of $D$.  Following Livshits
et al.~\cite{shapley-icdt2020}, the Shapley value of
$f$ w.r.t.~$q$, denoted $\shp(D,q,f)$, is the value of $\shp(A,v,a)$
where:
\begin{itemize}
    \item $A=D\endo$.
    \item $v(E)=q(E\cup D\exo)-q(D\exo)$ for all $E\subseteq D\endo$.
    \item $a=f$.
\end{itemize}
That is, we consider a cooperative game where the endogenous facts are the players and the wealth function $v(E)$ measures the change to the result of the query due to the addition of the facts of $E$ to the exogenous facts. Here, we view a Boolean CQ $q$ as a numerical query such that $q(D)=1$ if $D\models q$ and $q(D)=0$ otherwise. 


\begin{example} \label{example:q1}
Consider again the database of our running example. We assume that all the facts in \rel{Stud}, \rel{Course} and \rel{Adv} are exogenous, while the facts in \rel{TA} and \rel{Reg} are endogenous.
Consider the query $q_1$ asking if there is a student who is not a TA and is registered to at least one course. Note that facts from \rel{Reg} can only have a positive impact on the query result (i.e., they can only change it from false to true), while the facts of \rel{TA} can only have a negative impact on the result (i.e., they can only change it from true to false). Clearly, it holds that $D\exo\not\models q$, as no fact of $\rel{Reg}$ appears in $D\exo$. The answer to $q_1$ on $D\exo\cup E$ for some $E\subseteq D\endo$ is true if at least one of the following holds: \e{(1)} $\regf_4$ or $\regf_5$ appear in $E$, \e{(2)} $\regf_1$ or $\regf_2$ appear in $E$, but $\taf_1$ does not, or \e{(3)}  $\regf_3$ appears in $E$, but $\taf_2$ does not.

We can immediately see that $\taf_3$ can never affect the query result, since David does not appear in \rel{Reg}; hence, we have that $\shp(D,q_1,\taf_3)=0$.
Adding the fact $\taf_2$ in a permutation would change the query result from true to false if $\regf_3$ has been added before, and none of conditions $(1)$ or $(2)$ holds. Thus, the following subsets of facts may appear before $\taf_2$ in a permutation where it changes the query result: $\{\regf_3\}$, $\{\regf_3, \taf_1\}$, $\{\regf_3, \regf_1, \taf_1\}$, $\{\regf_3, \regf_2, \taf_1\}$, and $\{\regf_3, \regf_2, \regf_1, \taf_1\}$. Note that we can add $\taf_3$ to each of these subsets; thus, we have the following:
\begin{center}
\scalebox{0.9}{$
\shp(D,q_1,\taf_2) =
-\dfrac{1!\cdot 6!+2\cdot2!\cdot5! +3\cdot(3!\cdot4!+4!\cdot3!)+5!\cdot2!}{8!}
$}
\end{center}
and we conclude that $\shp(D,q_1,\taf_2)=-\frac{2}{35}$. Similarly, the fact $\taf_1$ changes the query result from true to false when at least one of $\regf_1$ or $\regf_2$ appears earlier in the permutation, and none of conditions $(1)$ or $(3)$ holds. That is, the fact $\taf_1$ should appear after one of the following subsets of facts: $\{\regf_1\}$,
$\{\regf_2\}$,
$\{\regf_1,\regf_2\}$,
$\{\regf_1\, \taf_2\}$,
$\{\regf_2, \taf_2\}$,
$\{\regf_1,\regf_2\, \taf_2\}$,
$\{\regf_1,\regf_3, \taf_2\}$,
$\{\regf_2, \regf_3, \taf_2\}$, 
$\{\regf_1,\regf_2\,\regf_3, \taf_2\}$. We can again add $\taf_3$ to each of the subsets. 
Thus, the following holds:
\begin{center}
\scalebox{0.9}{$\shp(D,q_1,\taf_1) =-\dfrac{2\cdot 1!\cdot6! +5\cdot2!\cdot5!+6\cdot3!\cdot4!+4\cdot4!\cdot3!+5!\cdot2!}{8!}$}
\end{center}
and $\shp(D,q_1,\taf_1)=-\frac{3}{28}$. As it holds that $|\shp(D,q_1,\taf_1)|>|\shp(D,q_1,\taf_2)|$, we deduce that the fact that Adam is not a TA has a greater negative impact on $q_1$ than the fact that Ben is not a TA. This is expected, since Adam is registered to more courses.
\par
As for the facts in the relation \rel{Reg}, we have that:
\begin{center}
\scalebox{0.9}{
$\shp(D,q_1, \regf_1) = \shp(D,q_1, \regf_2)= \dfrac{37}{210}$}

\scalebox{0.9}{
$\shp(D,q_1,\regf_3) = \dfrac{27}{140}$
}

\scalebox{0.9}{
$\shp(D,q_1,\regf_4)=\shp(D,q_1,\regf_5)=\dfrac{13}{42}$}
\end{center}
We elaborate on these calculations in the Appendix. Note that the sum over the Shapley values of all the endogenous facts is $1$.  
\qed
\end{example}

%% file: sjfBCQ.tex
\section{Exact Evaluation} \label{sec:bcqn}

In this section, we investigate the complexity of computing the
Shapley value for \CQneg{s} without self-joins, and establish the
following dichotomy in the data complexity of the problem.

%

\def\thmsjfBCQ{ Let $q$ be a \CQneg without self-joins. If $q$ is
  hierarchical, then $\shp(D,q,f)$ can be computed in polynomial time,
  given $D$ and $f$. Otherwise, its computation is
  $\FP$-complete.\footnote{Recall that $\FP$ is the class of problems
    that can be
  solved in polynomial time with an oracle to a \#P-complete problem.} }
\begin{theorem}\label{thm:sjfBCQ} \thmsjfBCQ \end{theorem}
For illustration, the theorem states that the Shapley value can
  be computed in polynomial time for the query $q_1$ of
  Example~\ref{example:queries}, but computing it for the
  query $q_2$ is $\FP$-complete.
  Interestingly, the classification criteria is the same as the one
  for self-join-free CQs without
  negation~\cite{shapley-icdt2020}. Hence, the added
  negation does not change the complexity picture for the exact
  computation of the Shapley value. (However, as we will show later,
  the addition of negation has a significant impact on the approximate
  computation of the value.)
  Next, we discuss the proof of Theorem~\ref{thm:sjfBCQ}.
  
\par
Livshits et al.~\cite{shapley-icdt2020} introduced an
algorithm for computing the Shapley value for hierarchical self-join-free CQs. This algorithm relies on a reduction from the problem of
computing the Shapley value to that of computing the number of subsets
of size $k$ of $D\endo$ that, along with $D\exo$, satisfy $q$. We
denote this problem as $|\Sat(D,q,k)|$. As the reduction does not
assume anything about $q$ other than the fact that it is a Boolean
query, the same reduction applies to \CQneg{s}. Hence, it is only left
to show that $|\Sat(D,q,k)|$ can be computed in polynomial time for a
hierarchical \CQneg.

\def\lemmahierarchical{
Let $q$ be a hierarchical \CQneg without self-joins. There is a polynomial-time algorithm for computing the number $|\Sat(D, q, k)|$ of $k$-subsets $E$ of $D\endo$, such that $(D\exo\cup E)\models q$, given $D$ and $k$.
}
\begin{lemma}\label{lemma:hierarchical}
\lemmahierarchical
\end{lemma}

The algorithm \algname{CntSat}~\cite{shapley-icdt2020} for computing
$|\Sat(D,q,k)|$ for self-join-free CQs without negation is a recursive algorithm that reduces the number of variables in the query with each recursive call, based on the hierarchical structure of the query. The treatment of the base case, when no
  variables occur in $q$, is the only part of the algorithm that does not apply to queries with negation, and
  we explain how it should be modified in the Appendix.

In the remainder of this section, we focus on the proof of the
negative side of the theorem. We start by proving hardness for the
four simplest non-hierarchical \CQneg{s}: \begin{align*}
    \cqrst()&\dl R(x), S(x,y), T(y)\\
    \cqnrsnt()&\dl \neg R(x), S(x,y), \neg T(y)\\
    \cqrnst()&\dl R(x), \neg S(x,y), T(y)\\
    \cqrsnt()&\dl R(x), S(x,y), \neg T(y)
\end{align*}
The proof for $\cqrst$ is given in~\cite{shapley-icdt2020}; hence, we show the following.

\eat{
The proof of hardness for $\cqrst$ is given
in~\cite{shapley-icdt2020}; hence, we focus on the
other three queries and prove that computing the Shapley value for
each one of them is $\FP$-complete. Then, for each non-hierarchical
self-join-free \CQneg, we construct a reduction from computing
$\shp(D, q', f)$ for one of the basic queries $q'$ to computing
$\shp(D, q, f)$, similarly to what has been done in the proof of
hardness for non-hierarchical self-join-free CQs without
negation~\cite{shapley-icdt2020}. We start by proving
the following.
}

\def\lemmabasichard{
If $q$ is one of $\cqnrsnt,\cqrnst$, or $\cqrsnt$, then computing $\shp(D, q, f)$ is $\FP{}$-complete.
}
\begin{lemma}\label{lemma:four-hard-queries}
\lemmabasichard
\end{lemma}
\begin{proofsketch}
The proof of the lemma for $\cqnrsnt$ and $\cqrnst$ is by a reduction from the problem of computing $\shp(D,\cqrst,f)$. We show that for every database $D$ and a fact $f\in D$ we have that $\shp(D,\cqrst,f) = -\shp(D,\cqnrsnt,f)$, as $f$ changes the result of $\cqrst$ in a permutation $\sigma$ from false to true if and only if $f$ changes the result of $\cqnrsnt$ in $\sigma^R$ (which is the reverse permutation of $\sigma$) from true to false. As for the query $\cqrnst$, the idea is the following. Given an input database $D$ to the first problem, we construct a database $D'$ to our problem by taking the ``complement'' of the relation $S^D$. That is, we add a fact $f$ over the domain of $D$ to $S^{D'}$ if and only if this fact is not in $S^D$. This transformation does not affect the Shapley value since we can assume that every fact of $S$ is exogenous (as the database constructed in the proof of hardness for $\cqrst$ satisfies this property~\cite{shapley-icdt2020}). 
We will use this idea of the ``complement'' of a relation in our proofs in the next sections.

Most intricate is the proof of hardness for the query $\cqrsnt$. This
is due to its non-symmetrical structure that prevents us from
constructing a direct reduction from the problem of computing
$\shp(D,\cqrst,f)$. Similarly to the proof of hardness for
$\cqrst$~\cite{shapley-icdt2020}, we construct a
reduction from the problem of computing the number of independent sets
$|\is(g)|$ in a bipartite graph $g$, which is known to be
\#P-complete. Given an input bipartite
graph $g=(A\cup B,E)$,
we construct $n+1$ input instances $(D_i,f)$ for our problem (where
$n=|A|+|B|$), that provide us with an independent system of $n+1$  
linear equations over the numbers $|\s(g,k)|$, defined as follows.

For each $k=0,\dots,n$, the set $\s(g,k)$ contains every subset
$E\subseteq(A\cup B)$ of size $k$, such that for every
$a\in (E\cap A)$, and for every $(a,b)\in E$, we have that $b\in E$;
that is, for every vertex $a$ on the left-hand side of $g$ added to
$E$, we also add to $E$ every neighbor of $a$ in the graph. More
formally, \[\s(g):= \{A' \cup B'| A'\subseteq A, B'\subseteq B,
  \forall{(a,b)\in E} [a\in A'\Rightarrow b\in B']\}\] Note that a
subset $A'\cup B'$ in $\s(g)$ may contain vertices in $B'$ that are
not connected to any vertex in $A'$. Then, we denote by $\s(g,k)$ the
collection of subsets of size $k$ in $\s(g)$. We claim that
$|\s(g)| = |\is(g)|$. This holds since a subset $A'\cup B'$ of
vertices of $g$ is an independent set if and only if the subset
$A'\cup (B\setminus B')$ belongs to $\s(g)$. 

Each input instance $(D_i,f)$ to our problem is obtained from the bipartite graph $g$ by adding one vertex to the right-hand side of $g$ and $i$ vertices to its left-hand side. We connect every new vertex on the right-hand side to the new vertex on the left-hand side. Then, we add to $D_i$ an endogenous fact $R(a)$ for every vertex $a$ on the left-hand side of $g$, an endogenous fact $T(b)$ for every vertex $b$ on the right-hand side of $g$, and an exogenous fact $S(a,b)$ for every edge $(a,b)$ in $g$. We then compute, for each one of the instances, the Shapley value of the fact corresponding to the new vertex on the right-hand side of $g$, and obtain an equation over the numbers $|\s(g,k)|$. We show that the equations are independent; hence, we can
compute $|\is(g)|=\sum_{k=0}^n |\s(g,k)|$.
\end{proofsketch}

\def\lemmasjfbcqhard{If $q$ is a non-hierarchical
  \CQneg without self-joins, then computing $\shp(D, q, f)$ is $\FP{}$-complete. }

Using Lemma~\ref{lemma:four-hard-queries}, we can prove the whole
hardness side of Theorem~\ref{thm:sjfBCQ}. We adapt the reduction to
the one used for the case of non-hierarchical self-join-free CQs
without negation~\cite{shapley-icdt2020}. Recall that
every non-hierarchical self-join-free \CQneg contains three atoms
$\alpha_x, \alpha_y, \alpha_{x,y}$, such that $x$ and $y$ are two
variables of $q$, the variable $x$ occurs in $\alpha_x$ while $y$ does
not, the variable $y$ occurs in $\alpha_y$ while $x$ does not, and
both variables occur in $\alpha_{x,y}$. Furthermore, since $q$ is
safe, we can always choose $\alpha_x, \alpha_y, \alpha_{x,y}$ such
that if two of the atoms are negative, the negative ones are
$\alpha_x$ and $\alpha_y$. Hence, for every such $q$, we can construct
a reduction from computing the Shapley value for one of the queries
$\cqrst, q_{\neg R S \neg T}, \cqrnst$ or $\cqrsnt$
(depending on the polarity of $\alpha_x$, $\alpha_y$, and
$\alpha_{x,y}$) to computing $\shp(D,q,f)$, where the atoms over the
relations $R$, $S$ and $T$ are represented by the atoms
$\alpha_x,\alpha_{x,y}$ and $\alpha_y$, respectively.

\subsubsection*{Remarks}
We conclude the section with two comments. First, Livshits et
al.~\cite{shapley-icdt2020} have shown how their
dichotomy for CQs can be extended to arbitrary
summations over CQs, using the linearity of expectation. Our dichotomy
here can be extended to aggregate functions over \CQneg{s} in a similar way.
For example, Theorem~\ref{thm:sjfBCQ} implies that the Shapley value
of a fact can be efficiently computed for the following aggregate query that sums up all the profits $r$ of exports of products $p$ to countries $c$ where $p$ does not grow:
\[\mathsf{Sum}
 \{\!\{r\mid\rel{Export}(p,c), \neg\rel{Grows}(c,p),\rel{Profit}(c,p,r)\}\!\}\] 
Second, the proof of Theorem~\ref{thm:sjfBCQ} heavily
relies on the assumption that the query is self-join-free. However,
our hardness results for the basic non-hierarchical queries
$\cqrst,\cqnrsnt,\cqrnst$ and $\cqrsnt$ can be generalized to certain
\CQneg{s} with self-joins, by replacing the atom over the relation $T$
with another atom over the relation $R$ (e.g., we can prove hardness
for the query $\neg R(x),S(x,y),\neg R(y)$). This can be proved using
a reduction from the corresponding self-join-free query (e.g., the
query $\neg R(x),S(x,y),\neg T(y)$) by assuming, without loss of
generality, that the values in the domain of $R^D$ and the values in
the domain of $T^D$ are disjoint. In fact, this result can be
generalized to a larger class of \CQneg{s} with self-joins, and we
give this result in the Appendix (Theorem~\ref{thm:bcqn_sj}).


%% file: exogenous.tex
\def\cqcomma{\,,\,}

\section{Accounting for Exogenous Relations}\label{sec:exo}

In the previous section, we showed that computing the Shapley value is
$\FP$-complete for every non-hierarchical self-join-free \CQneg. Yet,
this hardness result does not take into account the reasonable
assumption that some of the relations in the database contain \e{only
  exogenous facts}. For example, Meliou et
al.~\cite{DBLP:journals/pvldb/MeliouGMS11} discussed the case where
all the relations in the database are exogenous, except for one (e.g.,
``Director'' or ``Movie''); this one relation may be a suspect of
containing erroneous data, or the one that holds the single type of
entities of whom contribution we wish to quantify. In this section, we
show that accounting for such relations significantly changes the
complexity picture and, in particular, it makes some of the
intractable queries according to Theorem~\ref{thm:sjfBCQ} tractable.
In fact, we generalize Theorem~\ref{thm:sjfBCQ} to account for
exogenous relations and therefore establish the precise class of
\CQneg{s} that become tractable. Throughout this section, we underline
the relations containing only exogenous facts and their associated
query atoms.


\begin{example}~\label{example1:exo} Livshits et
  al.~\cite{shapley-icdt2020} demonstrated their work on a database from the domain of
  academic publications. They reasoned about the contribution of
  researchers to the total number of citations and assumed that the
  information about the publications is exogenous.  In particular, they considered
  the query:
$$q()\dl \rel{Author}(x,y)\cqcomma \underline{\rel{Pub}(x,z)} \cqcomma
  \underline{\rel{Citations}(z,w)}$$ 
  Since $q$ is not
hierarchical, their result classifies it as intractable. However, in this section, we show that there is a polynomial-time algorithm for computing the Shapley
value for $q$, under the assumption that $\rel{Pub}$ and
$\rel{Citations}$ contain only exogenous facts. Furthermore, we show
that even if we had that prior knowledge about the relation
$\rel{Citations}$ alone, we would still able to compute the Shapley
value efficiently. This is due to the fact that we can
reduce the problem of computing $\shp(D,q,f)$ to that of
computing $\shp(D,q',f)$ for the hierarchical query
$q'()\dl \rel{Author}(x,y)\cqcomma\rel{Pub}(x,z)$, by removing from the relation $\rel{Pub}$ in $D$ every fact
$\rel{Pub}(a,b)$ such that there is no fact $\rel{Citations}(b,c)$ in $D$ and then removing
the relation $\rel{Citations}$ from the query.

Next, consider the database of our running example (Figure~\ref{fig:DB}). We have
  assumed that the information about the students and courses in the
  faculty is exogenous, and our goal was to understand how much the
  fact that a student takes or teaches a course affects the result of
  different queries. For example, consider again the query $q_2$ from
  Example~\ref{example:queries}.
$$q_2()\dl \underline{\rel{Stud}(x)} \cqcomma \neg \rel{TA}(x),
\rel{Reg}(x,y) \cqcomma \neg \underline{\rel{Course}(y, \val{CS})}$$
 Theorem~\ref{thm:sjfBCQ} classifies this query as intractable for computing the Shapley value, as it is not hierarchical.  Yet,
again, the Shapley value can be computed in polynomial time, using an algorithm that takes into consideration the assumption that every fact in $\rel{stud}$ and $\rel{course}$ is exogenous. Note that when negation is added to the picture, we cannot simply remove exogenous atoms, as removing an exogenous atom may turn a query with safe negation into a query with negation that is not safe (e.g., $q'()=\underline{R(x)}\cqcomma \neg S(x,y), T(y)$).
\qed \end{example}



\subsection{Generalized Dichotomy}

We start by formally defining the problem that we study in this section. We define an \e{exogenous relation} $R$ to be a relation that consists only of exogenous facts.
We fix a schema $\signature$, a set $X$ of exogenous relations in
$\signature$, and a self-join-free \CQneg $q$. We denote by
$\signature_X$ a schema with the set $X$ of exogenous relations. Note
that we do not assume anything about the facts in the relations
outside $X$ and they may contain both endogenous and exogenous facts.
Then, our goal is to compute $\shp(D,q,f)$, given a database $D$ over
$\signature_X$ and a fact $f\in D$.

Clearly, the assumption that some of the relations of $\signature$ are exogenous does not change the fact that we can compute the Shapley value in polynomial time for any hierarchical \CQneg.
To understand the impact of this assumption on the complexity of non-hierarchical CQs,
consider the query $\cqrnst$ defined in Section~\ref{sec:bcqn}. If we assume that only $S$ is exogenous, then the query remains hard,
as $S$ already contains only exogenous facts in the proof of hardness
for $\cqrnst$
(Lemma~\ref{lemma:four-hard-queries}).
We can generalize this example and show that having a
non-hierarchical triplet $(\alpha_x,\alpha_{x,y},\alpha_y)$ where
$R_{\alpha_x}\not\in X$ and $R_{\alpha_y}\not\in X$ is a sufficient
condition for $\FP$-hardness, as the hardness proofs of the
previous section can be easily generalized to this case. Is having
such a triplet a necessary condition for hardness? Next, we answer
this question negatively.
Consider the following queries:
$$q()\dl \neg R(x,w), \underline{S(z,x)}, \neg\underline{ P(z,w)}, T(y,w)$$
$$q'()\dl \neg R(x,w), \underline{S(z,x)}, \neg\underline{ P(z,y)}, T(y,w)$$
In both queries, the exogenous relation are $S$ and $P$, and they
differ only in one variable that occurs in the atom of $P$. While the
two queries are very similar and are both classified as intractable by
Theorem~\ref{thm:sjfBCQ}, we will show that in the model considered in
this section, $\shp(D,q,f)$ can be computed in polynomial time for
every endogenous fact $f$, while computing $\shp(D,q',f)$ is
$\FP$-complete.
This holds true as while in both cases the non-exogenous
atoms are connected via the exogenous atoms, they are connected in
different ways. While the connection in $q$ between the variable $x$
in $\neg R(x,w)$ and the variable $y$ in $T(y,w)$ goes through the variable $w$, in
$q'$ the connection between $x$ and $y$ is possible through the
variable $z$ as well, and we need to be able to distinguish between
these two cases. In the terminology we set next, we say that $x$ and
$y$ are connected via a \e{non-hierarchical path} in $q'$ (but not
in $q$).


Let $\signature_X$ be a schema. The \e{Gaifman graph} $\graph(q)$ of a
\CQneg $q$ is the graph that contains a vertex for every variable in
$q$ and an edge between two vertices if the corresponding variables
occur together in an atom of $q$. We say that a \CQneg $q$ has a
\e{non-hierarchical path} if there are two atoms $\alpha_x,\alpha_y$
and two variables $x,y$ in $q$ such that: \e{(1)}
$R_{\alpha_x}\not\in X$ and $R_{\alpha_y}\not\in X$, \e{(2)} the variable $x$ occurs in $\alpha_x$ but not in $\alpha_y$, while the variable $y$ occurs in $\alpha_y$ and not in $\alpha_x$,
and
\e{(3)} the graph obtained from $\graph(q)$ by removing every vertex
corresponding to a variable occurring in $\alpha_x$ or in $\alpha_y$,
contains a path between $x$ and $y$. In this case, we say that the
non-hierarchical path of $q$ is \e{induced} by the atoms $\alpha_x$
and $\alpha_y$.

\begin{figure}
    \begin{subfigure}{0.2\textwidth}
  \input{gaifman-nhp.pspdftex}
  \caption{\label{fig:gaifman1}}
\end{subfigure}
\begin{subfigure}{0.2\textwidth}\centering
    \input{gaifman-nonhp.pspdftex}
  \caption{\label{fig:gaifman2}}
\end{subfigure}
  \caption{The Gaifman graphs of the queries $q$ (left) and $q'$ (right) of Example~\ref{ex:gaifman}.\label{fig:gaifman}}
\end{figure}
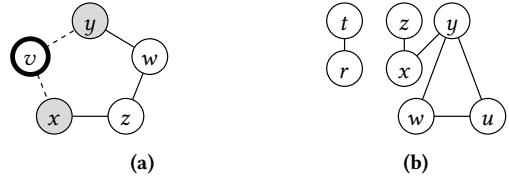

\begin{example}\label{ex:gaifman}
Consider the query:
\[ q()\dl \neg R(x), \underline{Q(x,v)}, \underline{S(x,z)}, U(z,w),
  \neg\underline{P(w,y)}, T(y,v) \] The Gaifman graph $\graph(q)$ is
illustrated in Figure~\ref{fig:gaifman1}. We claim that $q$ has a
non-hierarchical path induced by the atoms $\neg R(x)$ and $T(y,v)$. Note
that there is a path $x-v-y$ in $\graph(q)$ between $x$ and $y$;
however, this is not enough to determine that $q$ has a
non-hierarchical path, as we need to find a path that does not pass
through the variables of $\neg R(x)$ and $T(y,v)$. And indeed, if we remove
from $\graph(q)$ the variable $v$ occurring in the atom $T(y,v)$ and
every edge connected to it (i.e., every dotted line in the graph of
Figure~\ref{fig:gaifman1}), there is a path $x-z-w-y$ between the
variables $x$ and $y$ in the resulting graph, and we conclude that $q$
has a non-hierarchical path.

Next, consider the query:
\begin{align*} 
q'()\dl & U(t,r), \neg T(y), Q(y,w), \\
&\neg\underline{V(t)}, \underline{R(x,y)},\neg\underline{S(x,z)}, \underline{O(z)}, \underline{P(u,y,w)} 
\end{align*}
The
reader can easily verify, using the graph of
Figure~\ref{fig:gaifman2}, that $q'$ does not have a non-hierarchical
path. This is because the variables of $U(t,r)$ and the variables of
$\neg T(y)$ or $Q(y,w)$ are not connected in $\graph(q')$.
Moreover, every variable in $\neg T(y)$ also appears in $Q(y,w)$; hence, no
non-hierarchical path can be induced by these two atoms.
\qed
\end{example}

We prove the following generalization of Theorem~\ref{thm:sjfBCQ} that
account for exogenous relations.
\begin{theorem}\label{thm:exo} Let $\signature_X$ be a schema and let
  $q$ be a \CQneg without self-joins. If $q$ has a non-hierarchical
  path,
  then computing $\shp(D,q,f)$ is $\FP$-complete. Otherwise,
  $\shp(D,q,f)$ can be computed in polynomial time, given $D$ and $f$.
\end{theorem}

The proof of the hardness side of Theorem~\ref{thm:exo} is very
similar to the proof of hardness for Theorem~\ref{thm:sjfBCQ}. Given a
self-join-free \CQneg that has a non-hierarchical path, we construct a
reduction from the problem of computing $\shp(D,q',f)$ where $q'$ is
one of $\cqrst$, $\cqnrsnt$, or $\cqrsnt$ to that of computing
$\shp(D,q,f)$. The main difference between the proofs is that while in
the proof of Theorem~\ref{thm:sjfBCQ} we used the atom
$\alpha_{x,y}$ to represent the atom $S(x,y)$ in $q'$, here we use the
entire non-hierarchical path (or, more precisely, the atoms along the
edges of the non-hierarchical path) to represent this atom. The atoms
inducing the non-hierarchical path are used to represent the atoms
over the relations $R$ and $T$ in $q'$, and their polarity determines
the specific $q'$ we reduce from. (The full proof is in the Appendix.)
In the remainder of this section, we discuss the proof of the positive
side of Theorem~\ref{thm:exo}.

\subsection{Algorithm for the Tractable Cases}
We will show that computing the Shapley value for a self-join-free \CQneg that does not have a non-hierarchical path can be reduced to computing the Shapley value for a hierarchical query $q'$ without self-joins. Our reduction consists of three steps that will form the basis to our algorithm. Since the Shapley value can be computed in polynomial time for hierarchical \CQneg{s} (Theorem~\ref{thm:sjfBCQ}), and the same algorithm works for the model that we consider in this section, we will conclude that the Shapley value can be computed in polynomial time for such queries.

For the remainder of this section, we fix a schema  $\signature_X$ and a self-join-free \CQneg $q$ that does not have a non-hierarchical path. We first introduce some definitions and notations that we will use throughout the proof. We denote by $\at(q)$ and $\var(q)$ the sets of atoms and variables of $q$, respectively. We say that an atom $\alpha$ of $q$ is an \e{exogenous atom} if $R_\alpha\in X$. We say that a variable $x$ of $q$ is an \e{exogenous variable} if it occurs only in exogenous atoms of $q$. We denote the sets of all exogenous atoms and variables of $q$ by $\at\exo(q)$ and $\var\exo(q)$, respectively. We denote by $\at\nexo(q)$ the set $(\at(q)\setminus\at\exo(q))$ of non-exogenous atoms in $q$ and by $\var\nexo(q)$ the set $(\var(q)\setminus\var\exo(q))$ of non-exogenous variables.

Next, we define the \e{exogenous atom graph} $\exograph(q)$ of $q$ to be the graph that contains a vertex for every exogenous atom in $q$ and an edge between two vertices if the corresponding two atoms share an \e{exogenous} variable. The following lemma draws a connection between the properties of $\exograph(q)$ and the existence of a non-hierarchical path in $\graph(q)$. In particular, we prove that if a query $q$ does not have a non-hierarchical path, then for every connected component $C$ of $\exograph(q)$ there is a non-exogenous atom $\alpha$ of $q$ such that $\var\nexo(C)\subseteq\var(\alpha)$. This property is of high significance, as our reduction strongly relies on it.

\begin{lemma}\label{lemma:ccsubset}
For every connected component $C$ of $g\exo(q)$ there is an atom $\alpha\in\at\nexo(q)$ such that $\var\nexo(C)\subseteq\var(\alpha)$.
\end{lemma}
\begin{proof}
Let $C$ be a connected component of $\exograph(q)$. Assume, by way of contradiction, that there is no $\alpha\in\at\nexo(q)$ such that $\var\nexo(C)\subseteq\var(\alpha)$, and let $\alpha\in\at\nexo(q)$ be an atom of $q$ such that $\var\nexo(C)\cap\var(\alpha)$ is maximal among all atoms in $\at\nexo(q)$. Since $\var\nexo(C)\neq\emptyset$ and every non-exogenous variable occurs in a non-exogenous atom, there exists $x\in (\var\nexo(C)\cap\var(\alpha))$.
Moreover, since $\var\nexo(C)\not\subseteq\var(\alpha)$, there exists $y\in \var\nexo(C)$ that does not occur in $\alpha$. Since $y$ is not an exogenous variable, there is another $\alpha'\in\at\nexo(q)$ such that $y\in\var(\alpha')$. It cannot be the case that $x\in\var(\alpha')$ (as otherwise, we get a contradiction to the maximally of $\var\nexo(C)\cap\var(\alpha)$); hence, we conclude that $x\in(\var(\alpha)\setminus\var(\alpha'))$, $y\in(\var(\alpha')\setminus\var(\alpha))$, and $x,y\in\var(C)$. 

We claim that $\alpha$ and $\alpha'$ induce a non-hierarchical path in $\graph(q)$.
Since $x,y\in\var(C)$, there exist two atoms $\beta_1,\beta_2\in C$ such that $x\in\beta_1$ and $y\in\beta_1$. Since $\beta_1$ and $\beta_2$ belong to the same connected component, there exists a path in $\exograph(q)$ between $\beta_1$ and $\beta_2$, such that the edges along the path correspond to exogenous variables of $q$. Therefore, there is a path $x-v_1-\dots-v_n-y$ in $\graph(q)$, such that each $v_i$ is an exogenous variable (hence, $v_i\not\in\var(\alpha)$ and $v_i\not\in\var(\alpha')$). This path is a non-hierarchical path by definition.
\end{proof}

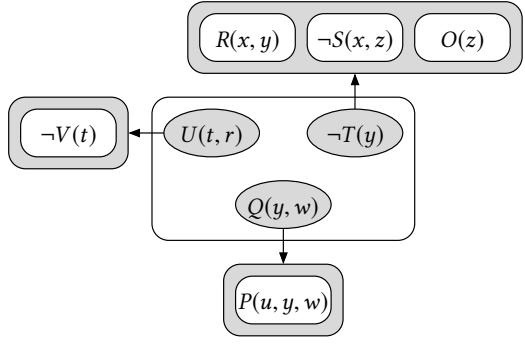
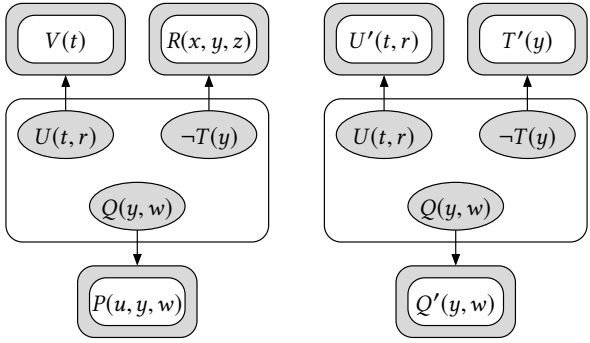
\begin{figure}
  \begin{subfigure}{0.4\textwidth}\centering
  \input{alg-orig.pspdftex}
  \caption{Original query.\label{fig:ex-step1}}
  \vspace{1em}
\end{subfigure}
 \begin{subfigure}[t]{0.22\textwidth}\centering
  \input{alg-pos.pspdftex}
  \caption{Joining exogenous relations.\label{fig:ex-step2}}
\end{subfigure}\quad
 \begin{subfigure}[t]{0.22\textwidth}\centering
  \input{alg-final.pspdftex}
  \caption{Removing exogenous variables and adding variables from the containing atom.\label{fig:ex-step3}}
\end{subfigure}
\caption{Illustration of the execution of the algorithm $\algname{ExoShap}$ on the query $q'$ from Example~\ref{ex:gaifman}.\label{fig:ex}}
\end{figure}

\begin{example}
Consider the query $q'$ of Example~\ref{ex:gaifman}. We have already established that $q'$ does not have a non-hierarchical path. Figure~\ref{fig:ex-step1} illustrates both the exogenous atom graph of $q'$ and the result of Lemma~\ref{lemma:ccsubset}. The atoms in the white rectangles are the exogenous atoms of $q'$, and the atoms in the gray circles are the non-exogenous atoms. Every gray rectangle containing a set of exogenous atoms represents a connected component in $\exograph(q')$. For example, the atoms $R(x,y)$ and $\neg S(x,z)$ share the exogenous variable $x$ and the atoms $\neg S(x,z)$ and $O(z)$ share the exogenous variable $z$. Hence, all three atoms form a connected component $C$ in the graph. The only non-exogenous variable in $C$ is $y$ and, indeed, there is a non-exogenous atom $\neg T(y)$ that uses $y$. In fact, there are two such atoms, and in the next step we can arbitrary select one of them. The exogenous atom $P(u,y,w)$ is a connected component on its own, as its only exogenous variable $u$ does not occur in any other atom. And, again, there is a non-exogenous atom $Q(y,w)$ that uses both non-exogenous variables $y$ and $u$ of $P(u,y,w)$.
\qed\end{example}


Next, we discuss the first step of our reduction. We prove that we can replace every connected component $C$ of $\exograph(q)$ with a single exogenous atom in $q$, obtained by ``joining'' all the atoms of $C$ (and the corresponding relations of $D$), without affecting the Shapley value.
Since some of the atoms in a connected component $C$ may be negated, and it is not clear how to combine positive and negative atoms into a single atom, we first replace them with positive atoms and compute the complement of the corresponding relations. Formally, given a negated atom $\alpha$, we denote by $\overline{\alpha}$ the atom obtained from $\alpha$ by removing the negation. Then, we denote by $\overline{R_{\alpha}^D}$ the relation obtained from $R_{\alpha}^D$ by adding every fact over the domain of $D$ if and only if it does not appear in $R_{\alpha}^D$. That is, if the arity of $R_\alpha$ is $k$, then we add to $\overline{R_{\alpha}^D}$ a fact $R_{\alpha}(c_1,\dots,c_k)$, where each $c_i$ is a constant from the domain of $D$, if and only if $R_\alpha(c_1,\dots,c_k)\not\in R_\alpha^D$. Hence, we obtain a query $q'$ by replacing every negated exogenous atom $\alpha$ of $q$ with the atom $\overline{\alpha}$, and we construct a database $D'$ by replacing every exogenous relation $R_\alpha^D$ corresponding to a negated atom of $q$ with the complement relation $\overline{R_\alpha^D}$. The same idea has been used in the proof of hardness for the query $\cqrnst$ in the previous section (Lemma~\ref{lemma:four-hard-queries}), and we prove that this transformation of the database and the query does not affect the Shapley value (i.e., $\shp(D,q,f)=\shp(D',q',f)$ for every $f$) in the Appendix.

From now on, we assume that every exogenous atom of $q$ is positive. We use that assumption to prove the following.




\def\joinlemma{ Computing $\shp(D,q,f)$, given $D$ and $f$, can be
  efficiently reduced to computing $\shp(D',q',f)$ for a \CQneg $q'$
  without self-joins such that: \e{(1)} every exogenous variable of
  $q'$ occurs is a single atom, and \e{(2)} $q'$ does not have any
  non-hierarchical path.}
\begin{lemma} \label{lemma:joinexo}
\joinlemma
\end{lemma}

In the proof of Lemma~\ref{lemma:joinexo}, given in the Appendix, we show that we can combine all the atoms $\alpha_1,\dots,\alpha_k$ of a connected component $C$ in $\exograph(q)$ into a single atom $\alpha_C$ such that $\var(\alpha_C)=\cup_{i\in\set{1,\dots,k}}\var(\alpha_i)$, while simultaneously replacing all the relations $R_{\alpha_1},\dots, R_{\alpha_k}$ in $D$ with a single relation $R_{\alpha_C}$ obtained by joining the $k$ relations according to the variables of the corresponding atoms, without affecting the Shapley value. 
We repeat this process with every connected component of $\exograph(q)$ and obtain a query $q'$ satisfying the first property of Lemma~\ref{lemma:joinexo}. As for the second property, we show that the existence of a non-hierarchical path $x-v_1-\dots-v_n-y$ in $q'$ induced by the atoms $\alpha_x$ and $\alpha_y$ implies the existence of a non-hierarchical path in $q$ induced by the same atoms, since every two consecutive variables $v_i,v_{i+1}$ in the path either occur together in a non-exogenous atom of $q$ or in a connected component of $\exograph(q)$.

One may suggest that it is possible to avoid replacing the negated exogenous atoms of $q$ with positive atoms before combining exogenous atoms, by simply constructing the relation $R_{\alpha_C}$ in $D$ using the query $q_C(\vec{x})\dl \alpha_1,\dots,\alpha_k$, where $\alpha_1,\dots,\alpha_k$ are the original (possibly negated) atoms in the connected component $C$, and $\vec{x}$ contains every variable of $C$. The problem with this approach is that the resulting $q_C$ may have non-safe negation, as a non-exogenous variable of $C$ may appear only in negated atoms of $C$ (and in a positive atom outside $C$). Thus, it is essential to replace the relations corresponding to negated exogenous atoms of $q$ with their complement relations, before combining the atoms of $C$ into a single one.

\begin{example}
Consider again the query $q'$ illustrated in Figure~\ref{fig:ex}. Since the atom $\neg S(x,z)$ in the topmost connected component $\set{R(x,y),\neg S(x,z),O(z)}$ is negated, we first replace it with a positive atom $S(x,z)$. Then, we combine all three atoms into a single atom $R(x,y,z)$, as illustrated in Figure~\ref{fig:ex-step2}, and replace these atoms in the query with the new atom. The new relation in the database will contain every answer to the query $q_C(x,y,z)\dl R(x,y),S(x,z),O(z)$ on $D$. Note that $\neg V(t)$ is a connected component on its own, but we still replace it with a positive atom $V(t)$.
\qed\end{example}



Next, we use the results of Lemmas~\ref{lemma:ccsubset} and~\ref{lemma:joinexo} to reduce the computation of $\shp(D,q,f)$ to the computation of $\shp(D',q',f)$ for a query $q'$ where every exogenous atom corresponds to a non-exogenous atom such that the two have the exact same variables.

\def\matchingAtomLemma{ Computing $\shp(D,q,f)$ can be
  efficiently reduced to computing $\shp(D',q',f)$ for a \CQneg
  $q'$ without self-joins such that: \e{(1)} for every $\alpha\in\at\exo(q')$ there
  exists $\alpha'\in\at\nexo(q')$ for which
  $\var(\alpha)=\var(\alpha')$, and \e{(2)} $q'$ does not have any
  non-hierarchical path.  }

\begin{lemma} \label{lemma:removeexo}
\matchingAtomLemma
\end{lemma}

Recall that $\at\exo(q)$ and $\at\nexo(q)$ are the sets of exogenous and non-exogenous atom in $q$, respectively. To establish Lemma~\ref{lemma:ccsubset}, we first observe that we can remove every exogenous variable from $q$ without affecting the Shapley value, as Lemma~\ref{lemma:joinexo} implies that we can assume that every exogenous variable occurs in a single exogenous atom. Then, Lemma~\ref{lemma:ccsubset} implies that for every exogenous atom $\alpha$ there exists a non-exogenous atom $\beta$ such that $\var(\alpha)\subseteq\var(\beta)$. Hence, for each such $\alpha$, we select such an atom $\beta$ and replace the atom $\alpha$ in $q$ with the atom $R_{\alpha'}(x_1,\dots,x_n)$, where $\{x_1,\dots,x_n\}$ is the set of variables in $\beta$. In the database, we replace the relation corresponding to the atom $\alpha$ with a new relation $R_{\alpha'}$ that contains every fact obtained from the Cartesian product of: \e{(1)} the projection of $R_{\alpha}^D$ to the attributes corresponding to non-exogenous variables in $\alpha$, and \e{(2)} every possible combination of $|\var(\beta)|-|\var\nexo(\alpha)|$ values from the domain of $D$.
The fact that $q'$ does not have a non-hierarchical path is rather straightforward based on the fact that every pair $\set{u_1,u_2}$ of variables of $q'$ occurring together in an atom of $q'$ necessarily occur together in a non-exogenous atom of $q'$ that is also an atom of $q$.

\begin{example}
Figure~\ref{fig:ex-step3} illustrates the implications of Lemma~\ref{lemma:removeexo} on the query $q'$ of Example~\ref{ex:gaifman}. We replace the relation $R(x,v,z)$ with the relation $T'(v)$ obtained from it by removing the exogenous variables $x$ and $z$. As for the atom $V(t)$, it does not contain any exogenous variables, but the non-exogenous atom $U(t,r)$ containing all the variables of $V(t)$ also uses the variable $r$; hence, we add this variable and obtain a new atom $U'(t,r)$.
\qed\end{example}

Our final observation is that a query $q'$ satisfying the properties of Lemma~\ref{lemma:removeexo} is hierarchical. This holds true since $\at\nexo(q')$ does not contain a non-hierarchical triplet (otherwise, the original $q$ would contain a non-hierarchical path). Adding an atom $\alpha$ to a hierarchical query $q$ such that $\var(\alpha)=\var(\alpha')$ for some atom $\alpha'$ in $q$ cannot change the non-hierarchical structure of the query.

\begin{algorithm}[t] \label{alg:exo}
\SetAlgoLined
\ForAll {negated $\alpha\in\at\exo(q)$}
        {
        Replace $\alpha$ in $q$ with $\overline{\alpha}$
        
        Replace $R_{\alpha}^D$ with $\overline{R_{\alpha}^{D}}$
}
\ForAll{$\set{\alpha_1,\dots\alpha_k}\in\cc(\exograph(q))$}
    {
    
    $\set{x_1,\dots,x_n}\leftarrow$ variables occurring in ${\alpha_1,\dots\alpha_k}$
    
    $q'(x_1,\dots,x_n)\dl \alpha_1,\dots,\alpha_k$
    
    Replace $\alpha_1,...,\alpha_k$ in $q$ with $R_{\alpha}(x_1,\dots,x_n)$
    
    Replace $R_{\alpha_1}^D,...,R_{\alpha_k}^D$ with $R_{\alpha}^{D}=q'(D)$ 
    }
\ForAll {$\alpha\in\at\exo(q)$}
    {Let $\beta\in\at\nexo(q)$ s.t.~$\var\nexo(\alpha)\subseteq \var(\beta)$
    
    $\set{x_1,\dots,x_n}\leftarrow$ variables occurring in $\beta$
    
    $\set{y_1,\dots,y_m}\leftarrow$ non-exogenous variables occurring in $\alpha$
    
    $q'(y_1,\dots,y_m)\dl \alpha$
    
    
    
    Replace $\alpha$ in $q$ with $R_{\alpha'}(x_1,\dots,x_n)$ 
    
    Replace $R_{\alpha}^D$ with $R_{\alpha'}^D=q'(D)\times \{(c_1,...,c_{n-m})\mid c_i\in \Dom(D)\}$ 
    }
 \Return $\shp{}(D,q,f)$
\caption{$\algname{ExoShap}(D,q,f)$\label{alg:exogenous}}
\end{algorithm}

We summarize the section with the algorithm \algname{ExoShap$(D,q,f)$}
(Algorithm~\ref{alg:exogenous}) for computing the Shapley value for a
self-join-free \CQneg that does not have a non-hierarchical path.  The
algorithm starts by modifying $q$ and $D$ according to the steps
described throughout this section. First, it replaces the negated
exogenous atoms of $q$ with positive atoms, and the corresponding
relations in $D$ with their complement relations. Then, it combines
the exogenous atoms in every connected component of $\exograph(q)$
into a single atom while joining the corresponding relations of
$D$. Finally, it removes the exogenous variables of $q$, and adds to
every exogenous atom the missing variables from the non-exogenous atom
containing it. The final database is constructed from the Cartesian
product of the projection of every exogenous relation $R_\alpha^D$ to
the attributes corresponding to the non-exogenous variables of
$\alpha$, and the set $\{c_1,\dots,c_r\mid c_i\in\Dom(D)\}$, where $r$
is the number of non-exogenous variables we have added to
$\alpha$. Then, the algorithm invokes an algorithm for computing the
Shapley value for hierarchical queries.


\subsection{Application to Probabilistic Databases}
We conclude by observing that our results in this section also apply
to the problem of query evaluation over tuple-independent
probabilistic databases~\cite{DBLP:conf/vldb/DalviS04}. Fink and
Olteanu~\cite{DBLP:journals/tods/FinkO16} have studied this problem
for queries with negation.  They considered the class RA$^{-}$ of
queries that includes the \CQneg{s}. When restricting to \CQneg{s},
they established that query evaluation is possible in polynomial time
for hierarchical \CQneg{s}, and it is $\FP$-complete otherwise. The proofs
of this section immediately provide a generalization of their result
to account for \e{deterministic relations}, where the probability of
every fact is $1$. The only difference is that instead of using the
algorithm for computing the Shapley value for hierarchical \CQneg{s}, we
will use the algorithm for query evaluation over tuple-independent
probabilistic databases for hierarchical \CQneg{s}. Hence, we obtain the
following result (where $X$ is the set of deterministic relations).

\begin{theorem}\label{thm:prob} Let $\signature_X$ be a schema and let
  $q$ be a \CQneg without self-joins. If $q$ has a non-hierarchical path,
  then its evaluation over tuple-independent probabilistic databases
  is $\FP$-complete. Otherwise, the query can be evaluated in
  polynomial time. \end{theorem}

%% file: gaifman-nhp.pspdftex
\begin{picture}(0,0)%
\includegraphics{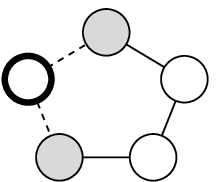}%
\end{picture}%
\setlength{\unitlength}{3947sp}%
\begingroup\makeatletter\ifx\SetFigFont\undefined%
\gdef\SetFigFont#1#2#3#4#5{%
  \reset@font\fontsize{#1}{#2pt}%
  \fontfamily{#3}\fontseries{#4}\fontshape{#5}%
  \selectfont}%
\fi\endgroup%
\begin{picture}(1005,838)(616,-330)
\put(1501,108){\makebox(0,0)[b]{\smash{{\SetFigFont{9}{10.8}{\familydefault}{\mddefault}{\updefault}{\color[rgb]{0,0,0}$w$}%
}}}}
\put(1126,333){\makebox(0,0)[b]{\smash{{\SetFigFont{9}{10.8}{\familydefault}{\mddefault}{\updefault}{\color[rgb]{0,0,0}$y$}%
}}}}
\put(1351,-267){\makebox(0,0)[b]{\smash{{\SetFigFont{9}{10.8}{\familydefault}{\mddefault}{\updefault}{\color[rgb]{0,0,0}$z$}%
}}}}
\put(901,-267){\makebox(0,0)[b]{\smash{{\SetFigFont{9}{10.8}{\familydefault}{\mddefault}{\updefault}{\color[rgb]{0,0,0}$x$}%
}}}}
\put(751,108){\makebox(0,0)[b]{\smash{{\SetFigFont{9}{10.8}{\familydefault}{\mddefault}{\updefault}{\color[rgb]{0,0,0}$v$}%
}}}}
\end{picture}%

%% file: gaifman-nonhp.pspdftex
\begin{picture}(0,0)%
\includegraphics{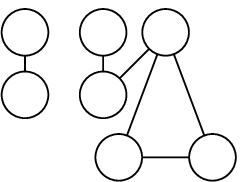}%
\end{picture}%
\setlength{\unitlength}{3947sp}%
\begingroup\makeatletter\ifx\SetFigFont\undefined%
\gdef\SetFigFont#1#2#3#4#5{%
  \reset@font\fontsize{#1}{#2pt}%
  \fontfamily{#3}\fontseries{#4}\fontshape{#5}%
  \selectfont}%
\fi\endgroup%
\begin{picture}(1140,838)(331,-330)
\put(826,333){\makebox(0,0)[b]{\smash{{\SetFigFont{9}{10.8}{\familydefault}{\mddefault}{\updefault}{\color[rgb]{0,0,0}$z$}%
}}}}
\put(451, 33){\makebox(0,0)[b]{\smash{{\SetFigFont{9}{10.8}{\familydefault}{\mddefault}{\updefault}{\color[rgb]{0,0,0}$r$}%
}}}}
\put(1351,-267){\makebox(0,0)[b]{\smash{{\SetFigFont{9}{10.8}{\familydefault}{\mddefault}{\updefault}{\color[rgb]{0,0,0}$u$}%
}}}}
\put(901,-267){\makebox(0,0)[b]{\smash{{\SetFigFont{9}{10.8}{\familydefault}{\mddefault}{\updefault}{\color[rgb]{0,0,0}$w$}%
}}}}
\put(1126,333){\makebox(0,0)[b]{\smash{{\SetFigFont{9}{10.8}{\familydefault}{\mddefault}{\updefault}{\color[rgb]{0,0,0}$y$}%
}}}}
\put(826, 33){\makebox(0,0)[b]{\smash{{\SetFigFont{9}{10.8}{\familydefault}{\mddefault}{\updefault}{\color[rgb]{0,0,0}$x$}%
}}}}
\put(451,333){\makebox(0,0)[b]{\smash{{\SetFigFont{9}{10.8}{\familydefault}{\mddefault}{\updefault}{\color[rgb]{0,0,0}$t$}%
}}}}
\end{picture}%

%% file: alg-orig.pspdftex
\begin{picture}(0,0)%
\includegraphics{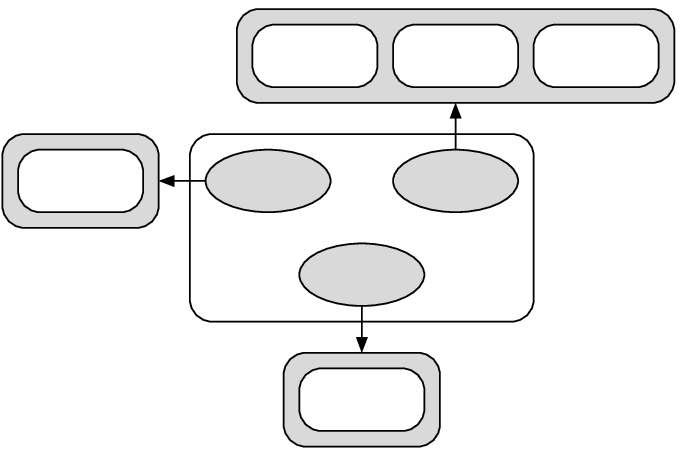}%
\end{picture}%
\setlength{\unitlength}{3947sp}%
\begingroup\makeatletter\ifx\SetFigFont\undefined%
\gdef\SetFigFont#1#2#3#4#5{%
  \reset@font\fontsize{#1}{#2pt}%
  \fontfamily{#3}\fontseries{#4}\fontshape{#5}%
  \selectfont}%
\fi\endgroup%
\begin{picture}(3249,2124)(-236,-898)
\put(151,333){\makebox(0,0)[b]{\smash{{\SetFigFont{9}{10.8}{\familydefault}{\mddefault}{\updefault}{\color[rgb]{0,0,0}$\neg V(t)$}%
}}}}
\put(1951,933){\makebox(0,0)[b]{\smash{{\SetFigFont{9}{10.8}{\familydefault}{\mddefault}{\updefault}{\color[rgb]{0,0,0}$\neg S(x,z)$}%
}}}}
\put(2626,933){\makebox(0,0)[b]{\smash{{\SetFigFont{9}{10.8}{\familydefault}{\mddefault}{\updefault}{\color[rgb]{0,0,0}$O(z)$}%
}}}}
\put(1276,933){\makebox(0,0)[b]{\smash{{\SetFigFont{9}{10.8}{\familydefault}{\mddefault}{\updefault}{\color[rgb]{0,0,0}$R(x,y)$}%
}}}}
\put(1501,-717){\makebox(0,0)[b]{\smash{{\SetFigFont{9}{10.8}{\familydefault}{\mddefault}{\updefault}{\color[rgb]{0,0,0}$P(u,y,w)$}%
}}}}
\put(1951,333){\makebox(0,0)[b]{\smash{{\SetFigFont{9}{10.8}{\familydefault}{\mddefault}{\updefault}{\color[rgb]{0,0,0}$\neg T(y)$}%
}}}}
\put(1501,-117){\makebox(0,0)[b]{\smash{{\SetFigFont{9}{10.8}{\familydefault}{\mddefault}{\updefault}{\color[rgb]{0,0,0}$Q(y,w)$}%
}}}}
\put(1051,333){\makebox(0,0)[b]{\smash{{\SetFigFont{9}{10.8}{\familydefault}{\mddefault}{\updefault}{\color[rgb]{0,0,0}$U(t,r)$}%
}}}}
\end{picture}%

%% file: alg-pos.pspdftex
\begin{picture}(0,0)%
\includegraphics{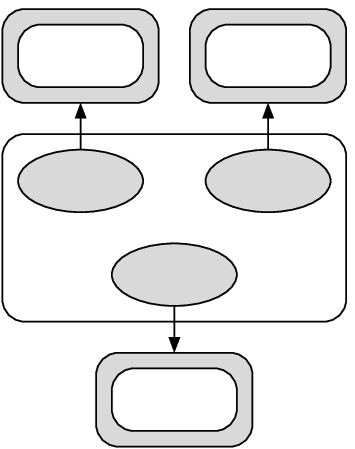}%
\end{picture}%
\setlength{\unitlength}{3947sp}%
\begingroup\makeatletter\ifx\SetFigFont\undefined%
\gdef\SetFigFont#1#2#3#4#5{%
  \reset@font\fontsize{#1}{#2pt}%
  \fontfamily{#3}\fontseries{#4}\fontshape{#5}%
  \selectfont}%
\fi\endgroup%
\begin{picture}(1674,2124)(664,-898)
\put(1501,-117){\makebox(0,0)[b]{\smash{{\SetFigFont{9}{10.8}{\familydefault}{\mddefault}{\updefault}{\color[rgb]{0,0,0}$Q(y,w)$}%
}}}}
\put(1501,-717){\makebox(0,0)[b]{\smash{{\SetFigFont{9}{10.8}{\familydefault}{\mddefault}{\updefault}{\color[rgb]{0,0,0}$P(u,y,w)$}%
}}}}
\put(1951,933){\makebox(0,0)[b]{\smash{{\SetFigFont{9}{10.8}{\familydefault}{\mddefault}{\updefault}{\color[rgb]{0,0,0}$R(x,y,z)$}%
}}}}
\put(1051,933){\makebox(0,0)[b]{\smash{{\SetFigFont{9}{10.8}{\familydefault}{\mddefault}{\updefault}{\color[rgb]{0,0,0}$V(t)$}%
}}}}
\put(1951,333){\makebox(0,0)[b]{\smash{{\SetFigFont{9}{10.8}{\familydefault}{\mddefault}{\updefault}{\color[rgb]{0,0,0}$\neg T(y)$}%
}}}}
\put(1051,333){\makebox(0,0)[b]{\smash{{\SetFigFont{9}{10.8}{\familydefault}{\mddefault}{\updefault}{\color[rgb]{0,0,0}$U(t,r)$}%
}}}}
\end{picture}%

%% file: alg-final.pspdftex
\begin{picture}(0,0)%
\includegraphics{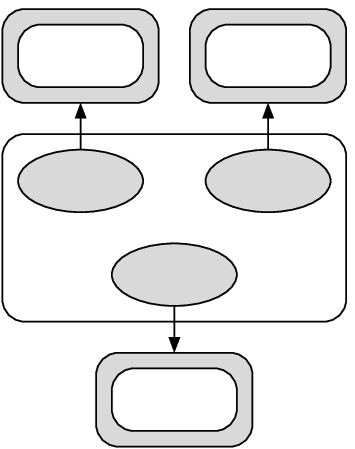}%
\end{picture}%
\setlength{\unitlength}{3947sp}%
\begingroup\makeatletter\ifx\SetFigFont\undefined%
\gdef\SetFigFont#1#2#3#4#5{%
  \reset@font\fontsize{#1}{#2pt}%
  \fontfamily{#3}\fontseries{#4}\fontshape{#5}%
  \selectfont}%
\fi\endgroup%
\begin{picture}(1674,2124)(664,-898)
\put(1951,933){\makebox(0,0)[b]{\smash{{\SetFigFont{9}{10.8}{\familydefault}{\mddefault}{\updefault}{\color[rgb]{0,0,0}$T'(y)$}%
}}}}
\put(1501,-717){\makebox(0,0)[b]{\smash{{\SetFigFont{9}{10.8}{\familydefault}{\mddefault}{\updefault}{\color[rgb]{0,0,0}$Q'(y,w)$}%
}}}}
\put(1051,933){\makebox(0,0)[b]{\smash{{\SetFigFont{9}{10.8}{\familydefault}{\mddefault}{\updefault}{\color[rgb]{0,0,0}$U'(t,r)$}%
}}}}
\put(1951,333){\makebox(0,0)[b]{\smash{{\SetFigFont{9}{10.8}{\familydefault}{\mddefault}{\updefault}{\color[rgb]{0,0,0}$\neg T(y)$}%
}}}}
\put(1501,-117){\makebox(0,0)[b]{\smash{{\SetFigFont{9}{10.8}{\familydefault}{\mddefault}{\updefault}{\color[rgb]{0,0,0}$Q(y,w)$}%
}}}}
\put(1051,333){\makebox(0,0)[b]{\smash{{\SetFigFont{9}{10.8}{\familydefault}{\mddefault}{\updefault}{\color[rgb]{0,0,0}$U(t,r)$}%
}}}}
\end{picture}%

%% file: approximation.tex
\section{Approximation}\label{sec:approx}
As seen in the previous sections, computing the exact Shapley value is
often hard. Hence, in this section, we consider its approximate
computation. There exists a Multiplicative
Fully-Polynomial Randomized Approximation Scheme (FPRAS) for computing
the Shapley value for any CQ and, in fact, for any union of
CQs~\cite{shapley-icdt2020}. Here, we show that the
addition of negation changes the complexity picture completely.
Recall that an FPRAS for a numeric function $f$ is a randomized
algorithm $A(x,\epsilon,\delta)$, where $x$ is an input for $f$ and
$\epsilon,\delta\in(0,1)$. The algorithm returns an
$\epsilon$-approximation of $f(x)$ with probability at least
$1-\delta$ in time polynomial in $x$, $1/\epsilon$ and
$\log(1/\delta)$. More formally, for an \e{additive} (or \e{absolute})
FPRAS we have that:
\[\Pr{f(x)-\epsilon\leq A(x,\epsilon,\delta)\leq f(x)+\epsilon)}\geq
  1-\delta\,,\] and for a \e{multiplicative} (or \e{relative}) FPRAS we have
that: \[\Pr{{f(x)}/{(1+\epsilon)}\leq A(x,\epsilon,\delta)\leq
    (1+\epsilon)f(x)}\geq 1-\delta\,.\]


\subsection{Additive vs.~Multiplicative Approximation} We start by
showing that there exists an additive FPRAS for computing the Shapley
value for \CQneg{s}. The additive FPRAS for \CQneg{s} is a
generalization of the additive FPRAS for CQs. We observe that when
negated atoms are allowed along with self-joins, a fact $f$ may change
the query result from false to true in one permutation, while changing
the query result from true to false in another permutation. For a
random permutation $\sigma$ of the facts in $D\endo$, the result of
$q(D\exo\cup\sigma_f\cup\{f\})-q(D\exo\cup\sigma_f)$ is a random
variable $x\in\{-1,0,1\}$. By using the Hoeffding bound for sums of
independent random variables in bounded
intervals~\cite{hoeffding1994probability}, we get an additive FPRAS
for computing $\shp(D,q,f)$ by taking the average value of $x$ over
$O\left(\log(1/\delta)/\epsilon^2\right)$ samples of random
permutations.


For CQs, an additive FPRAS is also a
multiplicative FPRAS~\cite{shapley-icdt2020}. This is
due to the \e{gap property}: there exists a polynomial $p$ such that
for all databases $D$ and facts $f\in D$ it holds that $\shp(D,q,f)$
is either zero or at least $1/p(|D|)$. We will now show that this property does not hold when negation is added to the picture; hence, this approach for
obtaining a multiplicative approximation of the Shapley value is no
longer valid. As an example, consider the query $$q()\dl R(x),S(x,y),\neg R(y)$$
and the database $D$ constructed as follows. For every $i\in\set{0,\dots,2n}$ we add to $D$ an exogenous fact $S(c_x^i,c_y^i)$. Moreover, for every $i\in\set{1,\dots,n}$ we add to $D$ an exogenous fact $R(c_x^i)$ and an endogenous fact $R(c_y^i)$, and for every $i\in\set{0,n+1,\dots,2n}$ we add an endogenous fact $R(c_x^i)$. We will show that the fact $f=R(c_x^0)$ does not satisfy the gap property. 

First, note that $D\exo\models q$ since for every $i\in\set{1,\dots,n}$, there is
a homomorphism $h$ from $q$ to $D$, where $h(x)=c_x^i$ and $h(y)=c_y^i$. For the fact
$f$ to change the query result from false to true, we first need to
add \e{all} the endogenous facts of the form $R(c_y^i)$ in $D$ to a permutation. Moreover, the first endogenous
fact of the form $R(c_x^i)$ that will
be added to the permutation will change the query result from false to
true, and no fact could change it back to false; hence, the fact $f$
has to appear before all these facts in a permutation. Overall, there
is exactly one subset $E\subseteq D\endo$, such that
$(D\exo \cup E)\not\models q$ and
$(D\exo \cup E \cup \{f\})\models q$. We have that $|E|=n$ and
$|D\endo|=2n+1$; thus, we conclude the following.
 \[|\shp(D,q,f)|=\dfrac{n!n!}{(2n+1)!}\leq\dfrac{1}{2^n}=2^{-\Theta(|D|)}\]
 
 We can generalize this result and show that the gap property does not
 hold for any ``natural'' CQ with negation and without constants.

 \def\thmgap{ Let $q$ be a satisfiable \CQneg with at least one
   negated atom.  Assume that $q$ has no constants, and that $q$ is
   positively connected. There is a sequence $\set{D_n}_{n=1}^\infty$
   of databases and a fact $f$ such that $|D_n|= \Theta(n)$ and
   $0<|\shp(D_n,q,f)|\leq 2^{-\Theta(n)}$.  }
\begin{theorem}\label{thm:negation-no-gap}
\thmgap
\end{theorem}

Note that by ``positively connected'' we mean that the positive atoms of $q$ are connected (i.e., every two variables of $q$ are connected in the Gaifman graph through positive atoms). The proof of Theorem~\ref{thm:negation-no-gap} is nontrivial and, as usual, we give it in the Appendix.

 Theorem~\ref{thm:negation-no-gap} implies that we need at least
 $2^{\Theta(|D|)}$ sample permutations (and, in particular,
 exponential time) to obtain a multiplicative approximation from the
 additive one. This does not mean that there is no multiplicative
 approximation for \CQneg{s}; however, we will show that there
 are \CQneg{s} for which a multiplicative approximation does not exist
 at all (under conventional complexity assumptions).

\subsection{Hardness of Multiplicative Approximation}
We now explore the complexity of computing a multiplicative
approximation for the Shapley value through a connection to the
problem of relevance to the query. 

\begin{definition}[Relevance]
  Let $q$ be a Boolean query and $D$ a database. A fact $f\in D\endo$
  is \e{relevant to $q$} if
  $q(D\exo\cup E) \neq q(D\exo\cup E \cup \{f\})$ for some
  $E\subseteq D\endo$; we then say that $f$ is \e{positively} (resp.,
  \e{negatively}) relevant to $q$ if $q(D\exo\cup E \cup \{f\})$ is
  true (resp., false).
\end{definition}


This problem of determining whether a fact is relevant to a query is
strongly related to the approximation problem, as we cannot obtain a
multiplicative approximation in cases where we cannot decide if the
Shapley value is zero or not. In turn, deciding on zeroness is related
to the relevance problem.  Clearly, if $\shp(D,q,f)\neq 0$, then $f$
is relevant to $q$. However, it may be the case that $f$ is relevant
to $q$ but $\shp(D,q,f)=0$, as the following example shows.

\begin{example}
  Consider the query $q()\dl R(x,y),\neg R(y,x)$ and the database
  $\set{R(\val{1},\val{2}),R(\val{2},\val{1})}$ where both facts are
  endogenous. The fact $R(\val{1},\val{2})$ is positively relevant for
  $E=\emptyset$, and it is negatively relevant for
  $E=\set{R(\val{2},\val{1})}$. Therefore, the number of permutations
  where $f$ changes the query result from false to true is equal to
  the number of permutations where $f$ changes the result from true to
  false and we have that $\shp(D,q,f)=0$.
\qed  \end{example}

Nevertheless, there are cases where the relevance problem coincides
with the problem of deciding whether the Shapley value is zero.
The Shapley value of a relevant fact $f$ can be zero if and only if
$f$ is both positively and negatively relevant. This may be the case
if and only if $f$ belongs to a relation that appears both as a
positive and a negative atom in the query. We call a relation symbol
\e{polarity consistent} if it appears in $q$ only in positive atoms or
only in negative atoms. A fact over a polarity-consistent relation
symbol is relevant to $q$ if and only if
$\shp(D,q,f)\neq 0$.

\begin{example}
  Consider again the queries of our running example
  (Example~\ref{example:queries}). Clearly, in the queries $q_1$ and
  $q_2$, every relation is polarity consistent, as the queries are
  self-join-free. The same holds for the query $q_3$, as \rel{Adv} and
  \rel{Reg} occur only in positive atoms while \rel{TA} occurs only in
  negative atoms. The query $q_4$, on the other hand, contains both
  polarity-consistent relations (i.e., \rel{Adv}) and relations that
  occur in both positive and negative atoms (i.e., \rel{TA} and
  \rel{Reg}). In this case, a fact $f$ in the relation \rel{Adv} is
  relevant to $q_4$ if and only if $\shp(D,q_4,f)>0$. However, for
  a fact $f$ in \rel{TA} it may be the case that $f$ is relevant
  to $q_4$ while $\shp(D,q_4,f)=0$.\qed
\end{example}

It is straightforward to show that the relevance to a CQ without
negation can be decided in polynomial time. The problem is known to be
NP-complete for Datalog programs with
recursion~\cite{DBLP:journals/ijar/BertossiS17}. We now show that
there exists a \CQneg $q$ containing a polarity-consistent relation $T$, such that the relevance of a $T$-fact to $q$ is
NP-complete. (Hence, so is the problem of deciding if the Shapley
value is zero.)


Consider the following \CQneg:
\[\cqtrsnr()\dl T(z),\neg R(x),\neg R(y), R(z), R(w), S(x,y,z,w)\]
We prove the following.
{
\definecolor{Gray}{gray}{0.9}
\def\emprow{\multicolumn{1}{l}{}}
\begin{figure}[t]
\small
\centering
\begin{subfigure}[b]{0.2\linewidth}
\begin{tabular}{|c|} 
\cline{1-1}
\rowcolor{black}
\multicolumn{1}{l}{\textcolor{white}{$R$}}\\\cline{1-1}
\rowcolor{Gray}
$\val{c}$\\
\rowcolor{Gray}
$\val{a}$\\
$\val{1}$\\
$\val{2}$\\
$\val{3}$\\
$\val{4}$\\
\cline{1-1}
\end{tabular}
\end{subfigure}
\begin{subfigure}[b]{0.4\linewidth}
\begin{tabular}{|c|c|c|c|} 
\cline{1-4}
\rowcolor{black}
\multicolumn{4}{l}{\textcolor{white}{$S$}} \\\cline{1-4}
\rowcolor{Gray}
$\val{1}$ & $\val{2}$ & $\val{a}$ & $\val{a}$\\
\rowcolor{Gray}
$\val{b}$ & $\val{b}$ & $\val{1}$ & $\val{3}$\\
\rowcolor{Gray}
$\val{3}$ & $\val{4}$ & $\val{1}$ & $\val{2}$\\
\rowcolor{Gray}
$\val{d}$ & $\val{d}$ & $\val{c}$ & $\val{c}$\\
\cline{1-4}
\emprow\\
\emprow\\
\end{tabular}
\end{subfigure}
\begin{subfigure}[b]{0.1\linewidth}
\begin{tabular}{|c|c|c|} 
\cline{1-1}
\rowcolor{black}
\multicolumn{1}{l}{\textcolor{white}{$T$}}\\\cline{1-1}
$\val{c}$\\
\rowcolor{Gray}
$\val{a}$\\
\rowcolor{Gray}
$\val{1}$\\
\rowcolor{Gray}
$\val{2}$\\
\rowcolor{Gray}
$\val{3}$\\
\rowcolor{Gray}
$\val{4}$\\
\cline{1-1}
\end{tabular}
\end{subfigure}
\caption{\label{fig:relevance} The database constructed in the proof of Proposition~\ref{prop:nsc1} for $(x_1\vee x_2)\wedge(\neg x_1\vee \neg x_3)\wedge(x_3\vee x_4\vee \neg x_1\vee \neg x_2)$. }
\end{figure}
}

\def\propnscfirst{ Deciding whether $f\in T^D$ is relevant to $\cqtrsnr$,
  given $D$ and $f$, is NP-complete.
}
\begin{proposition}\label{prop:nsc1}
\propnscfirst
\end{proposition}
\begin{proofsketch}
The proof, given in the Appendix, is by a reduction from the
satisfiability problem for $(2^+,2^-,4^{+-})$-CNF formulas, which are
formulas of the form $c_1\wedge\dots\wedge c_m$ where each clause
$c_i$ is either of the form $(x_j\vee x_k)$ or
$(\neg x_j\vee \neg x_k)$ or
$(x_j\vee x_k\vee \neg x_r\vee \neg x_p)$.
We prove that this problem is NP-complete in the Appendix.
We reduce this problem to the relevance problem for a
fact $f$ in the relation $T$.  Figure~\ref{fig:relevance} illustrates
the database constructed in the proof of Proposition~\ref{prop:nsc1}
for the formula
$(x_1\vee x_2)\wedge(\neg x_1\vee \neg x_3)\wedge(x_3\vee x_4\vee \neg
x_1\vee \neg x_2)$. The gray facts are exogenous. We now explain the
general idea of the proof using this example. Given a formula $\varphi$, we first add to $D$ an endogenous
fact $R(i)$ and an exogenous fact $T(i)$ for every
$i\in\set{\val{1},\dots,\val{n}}$ (where $n$ is the number of
variables used in $\varphi$). Then, for every clause of the form
$(x_i\vee x_j)$ we add an exogenous fact $S(i,j,\val{a},\val{a})$ to
$D$ (e.g., the fact $S(\val{1},\val{2},\val{a},\val{a})$ in the
database of Figure~\ref{fig:relevance} represents the clause
$(x_1\vee x_2)$). For every clause of the form
$(\neg x_i\vee \neg x_j)$ we add an exogenous fact
$S(\val{b},\val{b},i,j)$ to $D$ (e.g., the fact
$S(\val{b},\val{b},\val{1},\val{3})$ in the database of
Figure~\ref{fig:relevance} represents the clause
$(\neg x_1\vee \neg x_3)$). For every clause of the form
$(x_k\vee x_r\vee\neg x_i\vee \neg x_j)$ we add an exogenous fact
$S(k,r,i,j)$ to $D$ (e.g., the fact
$S(\val{3},\val{4},\val{1},\val{2})$ in the database of
Figure~\ref{fig:relevance} represents the clause
$(x_3\vee x_4\vee \neg x_1\vee \neg x_2)$). We also add to $D$ the
exogenous facts $R(\val{a})$ and $T(\val{a})$.

Next, we add an endogenous fact $f=T(\val{c})$ to $D$ and our goal
is to decide whether $f$ is relevant . For $f$ to change the
query result in any permutation, we also add the exogenous facts
$R(\val{c})$ and $S(\val{d},\val{d},\val{c},\val{c})$; thus, for every $E\subseteq D\endo$, it holds that
$(D\exo\cup E\cup \{f\})\models \cqtrsnr$.  We show that $f$ is
relevant to $\cqtrsnr$ if and only the formula $\varphi$ is
satisfiable. First, note that $D\exo\models \cqtrsnr$. In the example of Figure~\ref{fig:relevance}, this is due to
the existence of the facts $S(\val{1},\val{2},\val{a},\val{a})$,
$R(\val{a})$ and $T(\val{a})$ and the absence of the fact $T(\val{c})$
in $D\exo$. 
We
can assume that every $(2^+,2^-,4^{+-})$-CNF formula has a clause of
the form $(x_i\vee x_j)$ (hence, an exogenous fact of the form
$S(i,j,\val{a},\val{a})$ always exists in $D$), since the
satisfiablity problem is trivial for formulas that do not contain such
a clause---the zero assignment satisfies all of them. 
Hence, for the
fact $f$ to change the query result in a permutation, we first need to
add a subset $E$ of endogenous facts of the form $R(i)$ such that
$(D\exo\cup E)\not\models \cqtrsnr$. We show that the existence of a
satisfying assignment $z$ implies that such a subset $E$ exists
(formally, $E=\{R(i)\mid z(x_i)=1\}$). On the other hand, if $\varphi$
is not satisfiable, then such $E$ does not exist. The formula of our
example is satisfiable (e.g., by the assignment $z$ such that
$z(x_1)=z(x_4)=0$ and $z(x_2)=z(x_3)=1$), and the reader can verify
that indeed for $E=\set{R(\val{2}),R(\val{3})}$ we have that
$(D\exo\cup E)\not\models \cqtrsnr$ while
$(D\exo\cup E\cup\{f\})\models \cqtrsnr$. 
\end{proofsketch}

\begin{corollary}\label{cor:relcq}
  Given a database $D$ and a fact $f\in T^{D}$, deciding whether
  $\shp(D,\cqtrsnr,f)=0$ is NP-complete. 
\end{corollary}

The existence of a multiplicative FPRAS for $\shp(D,\cqtrsnr,f)$ would
imply the existence of a randomized algorithm that, for every
$\delta\in(0,1)$, returns zero if $\shp(D,\cqtrsnr,f)=0$ and a value
$v\neq 0$ otherwise, with probability at least $1-\delta$. Hence, we
could obtain a randomized algorithm for deciding if
$\shp(D,\cqtrsnr,f)=0$ from a multiplicative FPRAS for
$\shp(D,\cqtrsnr,f)$, in contradiction to the result of
Corollary~\ref{cor:relcq}.

Note that in Proposition~\ref{prop:nsc1} (and
Corollary~\ref{cor:relcq}) we consider a fact that belongs to a
polarity-consistent relation; however, the query is not
polarity-consistent as it contains a relation that appears both in a
positive and a negative atom of $q$ (i.e., the relation $R$). What
about the cases where \e{every} relation of $q$ is
polarity-consistent? We show that the problem of deciding whether the
Shapley value is zero can always be solved in polynomial time for
polarity-consistent queries. Hence, we conclude that having a non
polarity-consistent relation is a necessary condition for hardness of
this problem.





\begin{proposition}
Let $q$ be a polarity-consistent \CQneg. Given a database $D$ and a fact $f$, the following decision problems are solvable in polynomial time:
\begin{itemize}
    \item Is $f$ relevant to $q$?
    \item Is $\shp(D,q,f)=0$?
\end{itemize}
\end{proposition}

\begin{algorithm}[t]

\For{$h:\var(q)\rightarrow\Dom(D)$}{
\If{$h$ maps an atom $\alpha\in\negq(q)$ to some $f'\in D\exo$} {
\textbf{continue}
}
\If{$h$ maps an atom $\alpha\in\posq(q)$ to some $f'\not\in D$} {
\textbf{continue}
}
$P=\{f'\in D\endo\mid h\mbox{ maps an atom }\alpha\in\posq(q)\mbox{ to }f'\}$

$N=\{f'\in D\endo\mid h\mbox{ maps an atom }\alpha\in\negq(q)\mbox{ to }f'\}$

\If{$f\not\in P$}{
\textbf{continue}
}

\If{$(D\exo\cup (P\setminus\{f\})\cup (\negq_q(D\endo)\setminus N))\not\models q$} {
\Return true
}

}
\Return false
\caption{\algname{IsPosRelevant}$(D,q,f)$\label{alg:isposrel}}
\end{algorithm}


Since $q$ is polarity-consistent, the relevance to $q$ is the same as
the Shapley value being nonzero. Hence, to prove the proposition, we
introduce the algorithm $\algname{IsPosRelevant}$ (depicted as
Algorithm~\ref{alg:isposrel}) for deciding whether a fact $f$ is
positively relevant to $q$. The algorithm $\algname{IsNegRelevant}$ for deciding
whether a fact is negatively relevant is very similar, and we give it
in the Appendix. In the algorithms, we denote by $\Dom(D)$ the set of
constants used in the facts of $D$. Moreover, we denote by
$\negq_q(D\endo)$ the set of facts in $D\endo$ that appear in relations associated with
negative atoms of $q$.

In $\algname{IsPosRelevant}$, our goal is to decide if there is a subset $E\subseteq D\endo$ such that $(D\exo\cup E)\not\models q$ while $(D\exo\cup E\cup\{f\}) \models q$. Hence, we go over all possible mappings $h$ from the variables of $q$ to the constants of $D$, that map at least one positive atom of $q$ to $f$. Each such mapping defines a set $P$ of facts $f'\in D\endo$ such that $h$ maps a positive atom of $q$ to $f'$, and a set $N$ of facts $f'\in D\endo$ such that $h$ maps a negative atom of $q$ to $f'$. For a mapping $h$ to be an evidence for the relevance of $f$, it has to map every positive atom of $q$ to a fact of $D$ (and at least one such atom to $f$ itself) and none of the negative atoms of $q$ to a fact of $D\exo$ (otherwise, $h$ is not a homomorphism from $q$ to $D$). Moreover, it should be the case that every fact in $P\setminus\{f\}$ and none of the facts in $N$ appears in the set $E$ (the set of facts added before $f$ in a permutation). This ensures that $(D\exo\cup E\cup\{f\}) \models q$. 

However, this is not enough, as it may be the case that $(D\exo\cup E)\models q$ as well. To make sure that this is not the case, there must exist a set $F$ of facts corresponding to negative atoms of $q$ such that $(D\exo \cup (P\setminus f)\cup F)\not\models q$ and $F$ does not contain any fact of $N$. Then, for $E'=(E\cup F)$ we have that $(D\exo\cup E')\not\models q$ while $(D\exo\cup E'\cup\{f\}) \models q$. The main observation here is that since $q$ is polarity consistent, every fact corresponds to either positive or negative atoms of $q$, but not both; hence, we can add all the facts in $\negq_q(D\endo)\setminus N$ to $F$, and check whether the resulting set satisfies the conditions. If this is not the case, then no subset of $F$ satisfies this condition, and $h$ cannot be the evidence to the relevance of $f$. The proof of correctness of the algorithm is in the Appendix. The algorithm terminates in polynomial time since the number of mappings from the variables of $q$ to the constants of $D$ is polynomial in the size of $D$ when considering data complexity.

Interestingly, while the relevance problem can be solved in polynomial
time for any polarity-consistent CQ, this is no longer the case when
considering a union of polarity-consistent CQs. Specifically, we show
that the relevance to the \UCQneg
$\cqsat()\dl q_1() \vee q_2() \vee q_3() \vee q_4()$ is NP-complete, where:
\begin{align*}
    q_1()&\dl C(x_1,x_2,x_3,v_1,v_2,v_3), T(x_1,v_1), T(x_2,v_2), T(x_3,v_3)\\
    q_2()&\dl V(x), \neg T(x,\val{1}), \neg T(x,\val{0})\\
    q_3()&\dl T(x,\val{1}), T(x,\val{0})\\
    q_4()&\dl R(\val{0})
\end{align*}
In particular, we show that it is hard to decide whether the fact $R(\val{0})$ is relevant to $\cqsat$.

\def\propucqrelhard{ Given a database $D$ and the fact $f=R(\val{0})$,
  deciding whether $f$ is relevant to $\cqsat$ is NP-complete.  }

\begin{proposition}\label{prop:ucq_rel}
\propucqrelhard
\end{proposition}

The proof of the proposition is by a reduction from the satisfiability
problem for 3CNF formulas. Given an input formula $\varphi$, we construct an input database $D$ to our problem by
adding a fact $V(i)$ and two facts $T(i,\val{1})$ and $T(i,\val{0})$
for every variable $x_i$, and a fact $C(i,j,k,v_i, v_j, v_k)$ for every clause $(l_i\vee l_j\vee l_k)$ in $\varphi$, where $l_r$ is either $x_r$ or $\neg x_r$ for every $r\in\set{i,j,k}$. If $l_r=x_r$ then $v_r=\val{0}$ and if $l_r=\neg x_r$ then $v_r=\val{1}$. Intuitively, the purpose of the first query is to ensure that the assignment satisfies
every clause, the purpose of the second query is to ensure that the
assignment assigns at least one value to each variable, and the
purpose of the third query is to ensure that the assignment assigns at
most one value to each variable. Hence, we show that there exists a
satisfying assignment if and only if the fact $R(\val{0})$ is
(positively) relevant to $\cqsat$.
 
Since the relation $R$ is polarity-consistent and only occurs as a
positive atom in $\cqsat$, we again conclude the following.

\begin{corollary}
  Given a database $D$ and a fact $f\in R^D$, deciding whether
  $\shp(D,\cqsat,f)=0$ is NP-complete.
\end{corollary}

Note that while every individual \CQneg in the query $\cqsat$ is
polarity-consistent, the whole query is not, as
the relation $T$ appears as a positive atom in $q_1$ and $q_3$ and as
a negative atom in $q_2$. If a \UCQneg $q$ is such that the whole query is
polarity-consistent, then the relevance problem is solvable in
polynomial time. This is due to the fact that a fact $f$ is relevant
to such a UCQ $q$ if and only if it is relevant to at least one of the
CQs in $q$. Hence, we can use our algorithms $\algname{IsPosRelevant}$
and $\algname{IsNegRelevant}$ for every individual \CQneg in $q$ to decide
whether a fact $f$ is relevant to $q$. In this case, we cannot
preclude the existence of a multiplicative approximation.

%% file: conclusions.tex
\section{Concluding Remarks}\label{sec:conclusions}
We have investigated the complexity of computing the Shapley value for CQs and UCQs with negation. In particular, we have generalized a dichotomy by Livshits et al.~\cite{shapley-icdt2020} to classify the class of all CQs with negation and without self-joins. We further generalized this dichotomy to account for exogenous relations that are allowed to contain only exogenous facts. We have also studied the complexity of approximating the Shapley value in a multiplicative manner. The presence of negation makes this approximation fundamentally harder than the monotonic case, since the gap property (that unifies the additive and multiplicative FPRAS task) no longer holds. We have shown the hardness of approximation by making the connection to the problem of deciding relevance to a query, and by establishing hardness results for that problem.

This work leaves open several immediate directions for future research. In particular, we do not yet have a dichotomy for the class of CQs with self-joins (with or without negation). We know from past research that self-joins may cast dichotomies considerably more challenging to prove~\cite{DBLP:journals/jacm/DalviS12}. In addition, we have not yet studied the implication of the constraint of \e{endogenous relations} as an analogue of the exogenous relations; we believe that this problem tightly relates to the problem of model counting for conjunctive queries that has only recently been resolved~\cite{DBLP:journals/corr/abs-1908-07093}. Finally, we leave open some fundamental questions about the algorithmic and proof techniques for Shapley approximation. Is there a multiplicative FPRAS in the absence of the gap property? Are there cases where the relevance problem is tractable but a multiplicative approximation is computationally hard (beyond some ratio)?

%% file: Appendix.tex
\section{Details for Section \ref{sec:preliminaries}}
We now provide the missing computations for Example~\ref{example:q1}.
If the fact $\regf_3$ appears in a permutation before $\taf_2$ and conditions $(1)$ and $(2)$ do not hold, then $\regf_3$ changes the query result from false to true. Thus, there are five possible subsets of the endogenous facts that can appear before $\regf_3$ in a permutation where $\regf_3$ affects the query result: $\emptyset,\{\taf_1\}, \{\regf_1, \taf_1\},\{\regf_2, \taf_1\}, \{\regf_1, \regf_2, \taf_1\}$. To each one of those subsets we can add the fact $\taf_3$; hence, overall, we have ten possible subsets and we conclude that:
\\
\\
\scalebox{0.9}{
$\shp(D,q_1,\regf_3) = \dfrac{7!+2\cdot1!\cdot6! +3\cdot2!\cdot5!+3\cdot3!\cdot4!+4!\cdot3!}{8!}=\dfrac{27}{140}$}

As for the fact $\regf_1$, it changes the query result from false to true if both $\regf_2$ and $\taf_1$ appear later in the permutation, and none of the conditions $(1)$ or $(3)$ holds. Hence, the subsets of the endogenous facts that can appear before $\regf_1$ in a permutation are:
$\emptyset,\{\taf_3\}, \{\taf_2\}, \{\regf_3, \taf_2\},  \{\regf_3, \taf_2, \taf_3\}$, and we have that: 
\[\shp(D,q_1, \regf{}_1) = \dfrac{7!+2!\cdot6!+2!\cdot5!+3!\cdot4!}{8!}=\dfrac{27}{140}\] The same calculations hold for the fact $\regf{}_2$.

Finally, adding $\regf_4$ to a permutation before $\regf_5$ would change the query result from false to true, unless conditions $(2)$ or $(3)$ hold. In this case, there is a much larger number of subsets of facts that can appear before $\regf_4$ in a permutation where it changes the query result. We divide these subsets to four groups:
\begin{itemize}
    \item Subsets without any fact from \rel{Reg}, that is, all subsets of $\{\taf_1, \taf_2, \taf_3\}$.
    \item Subsets where the only possible facts from \rel{Reg} are $\regf_1, \regf_2$. In this group we have: $\{\regf_1, \taf_1\}, \{\regf_2, \taf_1\}, \{\regf_1, \regf_2, \taf_1\}$. To each of these subsets we can add a subset of the facts $\{\taf_2, \taf_3\}$.
    \item Subsets that include only the fact $\regf_3$ from \rel{Reg}. Here we have the subset $\{\regf_3, \taf_2\}$, and we can add to it every subset of $\{\taf_1, \taf_3\}$.
    \item Subsets that contain the fact $\regf_3$ and at least one of $\regf_1, \regf_2$, that is, $\{\regf_3, \taf_2, \regf_1, \taf_1\},\{\regf_3, \taf_2, \regf_2, \taf_1\}$, and $\{\regf_3, \taf_2, \regf_1,  \regf_2, \taf_1\}$.
\end{itemize}
To each one of these subsets we can add the fact $\taf_3$. Overall we have thirty possible subsets, and we conclude that: \\
\scalebox{0.9}{
$\shp(D,q_1,\regf_4)=$}\\\\
\scalebox{0.9}{
$\dfrac{7! +3\cdot1!\cdot6!+6\cdot2!\cdot5!+8\cdot3!\cdot4!+7\cdot4!\cdot3!+4\cdot5!\cdot2!+6!\cdot1!}{8!}=\dfrac{13}{42}$}\\
The same calculations hold for $\shp(D,q_1,\regf_5)$.

\section{Details for Section \ref{sec:bcqn}}
In this section, we provide the proofs of the lemmas used in the proof of hardness of Theorem~\ref{thm:sjfBCQ}. For convenience, we give the theorem here again.

\begin{reptheorem}{\ref{thm:sjfBCQ}}
\thmsjfBCQ
\end{reptheorem}

We start by proving the positive side of the theorem.

\begin{replemma}{\ref{lemma:hierarchical}}
\lemmahierarchical
\end{replemma}
\begin{proof}
  The algorithm \algname{CntSat} of Livshits et
  al.~\cite{shapley-icdt2020} for computing
  $|\Sat(D,q,k)|$ is a recursive algorithm that reduces the number of
  variables in the query with each recursive call. If there is a
  variable $x$ that occurs in every atom of $q$ (i.e., a root
  variable), then the problem is solved using dynamic programming, by
  considering every possible value of $x$. If no variable occurs in
  all atoms, then the query can be split into two disjoint
  sub-queries, in which case the problem is solved separately for each
  one of them. The treatment of these two cases applies to any
  hierarchical \CQneg as it only relies on the hierarchical structure
  of the query; however, the treatment of the base case, when no
  variables occur in $q$, does not apply to queries with negation, and
  we now explain how it should be modified.

  If at least one atom of $q$ does not correspond to any fact of $D$,
  then \algname{CntSat} will return $0$, as $D\not\models q$. This
  will also be the case if $k<|A|$ or $k>|D\endo|$, where
  $A=\at(q)\cap D\endo$. In any other case, the algorithm will return
  ${D\endo-|A|}\choose {k-|A|}$ which is the number of possibilities
  to select $k-|A|$ facts among those in $D\endo\setminus A$ (as every
  fact of $A$ should be selected to satisfy $q$). By modifying the
  base case in the following way, we ensure that the algorithm returns
  $|\Sat(D,q,k)|$ for a \CQneg. The algorithm will return $0$ in one
  of the following cases: \e{(a)} at least one of the positive atoms
  of $q$ does not appear as a fact of $D$, \e{(b)} at least one of the
  negative atoms of $q$ appears as a fact of $D$, or \e{(c)} $k<|A^+|$
  or $k>|D\endo|$ where $A^+=\posq(q)\cap D\endo$. In any other case,
  the result will be ${D\endo-|A^+|}\choose {k-|A^+|}$. It is rather
  straightforward that the modified algorithm will indeed return
  $|\Sat(D,q,k)|$ in polynomial time, based on the correctness and
  efficiency of \algname{CntSat}.
  \end{proof}

Next, we prove the hardness side of the theorem. First, we prove hardness for the basic non-hierarchical self-join-free queries, and then we reduce these problems to the problem of computing $\shp(D,q,f)$ for any non-hierarchical self-join-free \CQneg $q$.
We start by proving the following.

\begin{replemma}{\ref{lemma:four-hard-queries}}
\lemmabasichard
\end{replemma}

We prove the lemma separately for each one of the queries. We start with the query $\cqnotrsnott$.

\def\lemmanrsnt{
Computing $\shp(D, \cqnrsnt, f)$ is $\FP{}$-complete.
}
\begin{lemma}
\lemmanrsnt
\end{lemma}
\begin{proof}
We construct a reduction from the problem of computing $\shp(D, \cqrst, f)$ to that of computing $\shp(D, \cqnotrsnott, f)$. We make the following assumptions on the input database $D$ to the first problem: \e{(a)} every fact in $S$ is exogenous, and \e{(b)} for every $S(a,b)$ in $D$, it holds that both $R(a)$ and $T(b)$ are in $D$ as well. The database used in the proof of hardness for $\cqrst$~\cite{shapley-icdt2020} satisfies these properties; hence, computing $\shp(D, \cqrst, f)$ for such an input is $\FP$-complete.

Let $D$ be such database, and let $f\in D\endo$. Assume, without loss of generality, that $f=R(\val{0})$ (the proof for a fact $f$ in $T$ is symmetric). Let:
\[P_1 = \{\sigma\mid \sigma \in \Pi_{D\endo}, (\sigma_f \cup D\exo)\not\models\cqrst,(\sigma_f \cup D\exo\cup\{f\})\models\cqrst\}\]
\[P_2 = \{\sigma\mid \sigma \in \Pi_{D\endo}, (\sigma_f \cup D\exo)\models\cqnrsnt,(\sigma_f \cup D\exo\cup\{f\})\not\models\cqnrsnt\}\]
Recall that $\Pi_{D\endo}$ is the set of all possible permutations of the endogenous facts, and $\sigma_f$ is the set of facts that appear before $f$ in the permutation $\sigma$. That is, $P_1$ and $P_2$ are the sets of all permutations where $f$ changes the query result from false to true w.r.t.~$\cqrst$ and $\cqnrsnt$, respectively.
We claim that $|P_1|=|P_2|$ due to a bijection that exist between the two sets.

Let $g$ be a function defined as follows:
\[g: P_1\xrightarrow{}P_2,\hspace{2mm} g(\sigma)=\sigma^R\]
where $\sigma^R$ is the permutation $\sigma$ in reversed order (i.e., $\sigma_i = \sigma_{n-i+1}$ for all $i=\{1,2,...,n\}$, where $n=|D\endo|$).
First, we prove that if $\sigma\in P_1$ then $g(\sigma)\in P_2$.
If $f$ changes the query result from false to true in $\sigma$, then for every $S(a,b)\in D\exo$, at least one of the facts $R(a), T(b)$ is not in $D\exo \cup \sigma_f$, and there is at least one fact $T(c)$ in $D\exo \cup \sigma_f$ such that $S(\val{0},c)\in D\exo$. According to our assumption, for every $S(a,b)\in D\exo$ both $R(a)$ and $T(b)$ exist in $D$; hence, at least one of those is in $\sigma^R_{f}$. Moreover, there is at least one fact $T(c)$, which is not in $D\exo \cup \sigma^R_f$,  such that $S(\val{0},c)\in D\exo$. We conclude that $(D\exo\cup \sigma_f)\models \cqnrsnt$ because $S(\val{0},c)$ satisfies $\cqnotrsnott$ and $T(c)\not\in(D\exo\cup\sigma_f)$. Moreover, since for every fact $S(a,b)$ for $a\neq \val{0}$ at least one of $R(a)$ or $T(b)$ is in $D\exo\cup\sigma_f$, we have that $(D\exo\cup\sigma_f\cup\{f\})\not\models\cqnrsnt$, and $g(\sigma)\in P_2$.

Next, we prove that the function is injective and surjective.
\begin{itemize}
    \item \underline{Injectivity:} Let $\sigma_1, \sigma_2 \in P_1$ such that $\sigma_1\neq \sigma_2$. It follows directly that $\sigma_1^R\neq \sigma_2^R$.
    \item \underline{Surjectivity:} Let $\sigma\in P_2$. Since $(D\exo\cup\sigma_f)\models \cqnrsnt$, it holds that for every $S(a,b)\in D\exo$ such that $a\neq \val{0}$, at least one of $R(a), T(b)$ is in $D\exo \cup \sigma_f$. In addition, since $(D\exo\cup\sigma_f\cup\{f\})\not\models \cqnrsnt$, there is a fact $S(\val{0},c)\in D\exo$ such that $T(c)$ is not in $D\exo \cup \sigma_f$. Observe that $\sigma_R\in P_1$, as for every $S(a,b)\in D\exo$ such that $a\neq \val{0}$, at least one of $R(a), T(b)$ is not in $D\exo \cup \sigma^R_f$, and there is a fact $S(\val{0},c)\in D\exo$ such that $T(c)$ is in $D\exo \cup \sigma^R_f$. Thus $f$ changes the query result from false to true w.r.t.~$\cqrst$ in $\sigma^R$. Since ${\sigma^R}^R= \sigma$, we get that $g(\sigma^R) = \sigma$.
\end{itemize}

Thus, by the definition of the Shapley value we obtain that:
\[\shp(D,\cqrst, f) = \dfrac{|P_1|}{n!}=\dfrac{|P_2|}{n!}= -\shp(D,\cqnotrsnott, f)\]
and that concludes our proof.
\end{proof}

Next, we give the proof of hardness for $\cqrnst$.

\def\lemmarnst{
Computing $\shp(D, \cqrnst, f)$ is $\FP{}$-complete.
}
\begin{lemma}
\lemmarnst
\end{lemma}
\begin{proof}
We show a reduction from the problem of computing $\shp(D, \cqrst, f)$ to that of computing $\shp(D, \cqrnotst, f)$. As in the proof of the previous lemma, we make the assumption that every fact in $S$ is \e{exogenous}, while preserving the hardness of the original problem. Let $D$ be a database, and let $f\in D\endo$. Assume, without loss of generality, that $f=R(\val{0
})$. Let $D'$ be a database over the same schema as $D$, where the relations $R$ and $T$ of $D'$ consist of the exact same set of facts as the relations $R$ and $T$ in $D$. As for the relation $S$ in $D'$, it will contain the following set of facts:
\[S^{D'} = \{S(a,b) \mid R(a), T(b)\in D, S(a,b)\notin D\}\]
Observe that $D'\endo=D\endo$. we define two sets of permutations:
\[P_1 := \{\sigma\in\Pi_{D\endo}\mid(D\exo\cup\sigma_f)\not\models\cqrst,(D\exo\cup\sigma_f\cup\{f\})\models\cqrst\}\]
\[P_2 := \{\sigma\in\Pi_{D\endo}\mid(D'\exo\cup\sigma_f)\not\models\cqrnotst,(D'\exo\cup\sigma_f\cup\{f\})\models\cqrnotst\}\]
We prove that $P_1=P_2$ by showing a mutual inclusion between the sets:

\underline{$P_1\subseteq P_2$:} Let $\sigma\in P_1$. Since $(D\exo\cup\sigma_f)\not\models\cqrst$, for every pair of facts $R(a),T(b)$ in $\sigma_f$, we have that $S(a,b)\not\in\ D\exo$. Moreover, since $(D\exo\cup\sigma_f\cup\{f\})\models\cqrst$, there is a fact $S(\val{0},c)\in D\exo$ such that $T(c)\in\sigma_f$. By the definition of $S^{D'}$, we have that $S(a,b)\in D\exo'$ for every pair of facts $R(a),T(b)$ in $\sigma_f$; hence, $(D'\exo\cup\sigma_f)\not\models\cqrnotst$. We also have that $S(\val{0},c)\not\in D\exo'$ (while $T(c)\in\sigma_f$), and we conclude that $(D\exo\cup\sigma_f\cup\{f\})\models\cqrnst$.

\underline{$P_2\subseteq P_1$:} Let $\sigma\in P_2$. For every $R(a), T(b) \in\sigma_f$, there is a fact $S(a,b)\in D'\exo$, and there is also a fact $T(c)\in \sigma_f$ such that $S(\val{0},c)\notin D'\exo$. Thus, it holds that for every pair of facts $R(a),T(b)\in\sigma_f$, the fact $S(a,b)$ does not belong to $D\exo$ by the definition of $S^{D'}$, so $R(a), T(b)$ cannot be a part of any answer to $\cqrst$, and overall $(D\exo\cup\sigma_f)\not\models\cqrst$. Moreover, $S(\val{0},c)\in D\exo$, so the set of facts $\set{R(\val{0}), S(\val{0},c), T(c)}$ satisfies $\cqrst$, and we conclude that $(D'\exo\cup\sigma_f\cup\{f\})\models\cqrst$.

Finally, we deduce that:
\[\shp(D,\cqrst,f) = \dfrac{|P_1|}{n!} = \dfrac{|P_2|}{n!}= \shp(D', \cqrnotst, f)\]
where $n=|D\endo|=|D\endo'|$.
\end{proof}

Next, we prove hardness for the query $\cqrsnott$. This proof of hardness is the most intricate one due to the non-symmetrical structure of the query.

\def\lemmarsnt{
Computing $\shp(D, \cqrsnt, f)$ is $\FP{}$-complete.
}
\begin{lemma}
\lemmarsnt
\end{lemma} 

\begin{proof}
We construct a reduction from the known $\#\P$-complete problem of computing $|\is(g)|$---the number of independent sets in a bipartite graph $g$. Given an input graph $g=(A\cup B, E)$, where $A$ and $B$ are the disjoint sets of vertices in $g$, we define the following set:
\[\s(g):= \{A' \cup B'| A'\subseteq A, B'\subseteq B, \forall_{(a,b)\in E} (a\in A'\Rightarrow b\in B')\}\]
That is, $S(g)$ contains all subsets of the vertices in $g$, such that if a vertex from $A$ is in the subset, all of its neighbours from $B$ are in the subset as well. Note that a subset in $\s(g)$ may include additional vertices from $B$ that not connected to any vertex from $A$ in the subset. We denote by $\s(g,k)$ the set of all $E\in\s(g)$ such that $|E|=k$.

Given a bipartite graph $g=(A\cup B, E)$ where $|A|=m, |B|=n$, and $N=m+n$ such that none of the vertices in $g$ is isolated, we build a database $D^0$ which consists of the following facts: an endogenous fact $R(a)$ for every vertex $a\in A$, an endogenous fact $T(b)$ for every vertex $b\in B$, an exogenous fact $S(a,b)$ for every edge $(a,b)\in E$, and another endogenous fact $T(\val{0})$ for a fresh constant $\val{0}$. In addition, for every $a\in A$, $D^0$ will contain the exogenous fact $S(a,\val{0})$. We will compute the Shapley value for the fact $f=T(\val{0})$. (In fact, we will compute $1-\shp(D_0,\cqrsnt,f)$.)

There are two types of permutations $\sigma\in \Pi_{D^0\endo}$ for which it holds that $\cqrsnott(D^0\exo \cup \sigma_f) =  \cqrsnott(D^0\exo \cup \sigma_f \cup\{f\})$ (i.e., permutations where $f$ does not change the result of $\cqrsnott$):
\begin{enumerate}
    \item $(D^0\exo \cup \sigma_f)\not\models\cqrsnt$ and $(D^0\exo \cup \sigma_f \cup\{f\})\not\models\cqrsnt$. In this case, no fact from $R$ is in $\sigma_f$. Otherwise, there is $R(a)\in\sigma_f$ such that $R(a), S(a,\val{0})$ is an answer to $\cqrsnott$, based on the construction of $D^0$, in which case $(D^0\exo \cup \sigma_f)\models\cqrsnt$. Adding $f$ cannot form a new answer to the query. The number of permutations that satisfy this property is $P_{0\rightarrow0}= \dfrac{(N+1)!}{m+1}$, since each of the $m+1$ facts in $R^D\cup\{f\}$ has an equal chance to be the first (among these facts) to appear in a permutation, and we are interested in the permutations where $f$ appears before any fact in $R^D$.
    \item $(D^0\exo \cup \sigma_f)\models\cqrsnott$ and $(D^0\exo \cup \sigma_f \cup\{f\})\models\cqrsnt$. Here, we observe that $F=\{a\mid R(a)\in \sigma_f\}\cup\{b\mid T(b)\in \sigma_f\}$ is a subset of vertices in $g$ such that $F\notin\s(g)$. Otherwise, we get that for every $S(a,b)\in D^0\exo$ if $R(a)\in\sigma_f$ then $T(b)\in\sigma_f$, by the definition of $\s(g)$. Therefore, none of the pairs $R(a), S(a,b) \in (D\exo\cup \sigma_f)$ is an answer to $\cqrsnott$, since $T(b)\in \sigma_f$ as well. Hence, every pair of facts from $R$ and $S$ that satisfies $\cqrsnt$ is of the form $R(a),S(a,\val{0})$, but when we add $f$ to the permutation, none of these pairs satisfies the query anymore, in contradiction to the fact that $(D^0\exo \cup \sigma_f \cup\{f\})\models\cqrsnt$.
    We denote the number of permutations that satisfy this property as $P_{1\rightarrow1}$.   
\end{enumerate}
Next, we denote by $P_{1\rightarrow0}$ the number of permutations in $\Pi_{D^0\endo}$ where $f$ changes the result of $\cqrsnt$ from true to false. We have $N+1$ endogenous facts in $D^0\endo$, so we conclude that: $$P_{0\rightarrow0}+P_{1\rightarrow1}+P_{1\rightarrow0}=(N+1)!$$
Therefore, it holds that
$$P_{1\rightarrow0} = (N+1)!-(P_{0\rightarrow0}+P_{1\rightarrow1})$$   
By the definition of the Shapley value for database facts (Definition~\ref{def:shaply}) we obtain that:
\[\shp(D^0, \cqrsnott, f)=\dfrac{P_{0\rightarrow0}\cdot0+P_{1\rightarrow1}\cdot0+P_{1\rightarrow0}\cdot(-1)}{(N+1)!}\]
\[=\dfrac{(N+1)!-(P_{0\rightarrow0}+P_{1\rightarrow1})}{(N+1)!}=1-\dfrac{P_{0\rightarrow0}+P_{1\rightarrow1}}{(N+1)!}\]
Hence, we are able to compute $P_{1\rightarrow1}$ from the Shapley value of $f$ in the following way:
\[P_{1\rightarrow1} = (1-\shp(D^0, \cqrsnott, f))\cdot(N+1)!-\dfrac{(N+1)!}{m+1}\] 
\par
In the next step, we build $N+1$ instances $(D^r, f)$ for every $r\in\{1,...,N+1\}$ as follows: the database $D^r$ will contain an endogenous fact $R(a)$ for every vertex $a\in A$, an endogenous fact $T(b)$ for every vertex $b\in B$, an exogenous fact $S(a,b)$ for every edge $(a,b)\in E$, and the endogenous fact $T(\val{0})$. In addition, for every $i\in\{1,...,r\}$, the database $D^r$ will contain an endogenous fact $R(0_i)$ and the exogenous fact $S(0_i,\val{0})$. Once again, we consider the two types of permutations $\sigma\in \Pi_{D^r\endo}$, where $f$ does not change the result of $\cqrsnt$ from true to false:
\begin{enumerate}
    \item  $(D^r\exo \cup \sigma_f)\models\cqrsnott$ and $(D^r\exo \cup \sigma_f \cup\{f\})\models\cqrsnott$. In this case, the set $F=\{a\mid R(a)\in \sigma_f, a\neq 0_i\}\cup\{b\mid T(b)\in \sigma_f\}$ of vertices in $g$ is such that $F\notin\s(g)$. Otherwise, we have that the only pairs of facts from $R$ and $S$ satisfying $\cqrsnt$ are of the form $R(0_i), S(0_i,\val{0})$, which is a contradiction to the fact that $(D^r\exo \cup \sigma_f \cup\{f\})\models\cqrsnt$, since $f$ is ``cancelling'' each of these pairs (i.e., $\set{R(0_i),S(0_i,\val{0}),T(\val{0})}\not\models\cqrsnt$). The number of permutations satisfying this property is: $P^r_{1\rightarrow1} = P_{1\rightarrow1}\cdot m_r$, where $m_r=\binom{N+r+1}{r}\cdot r!$, as the $r$ endogenous facts of the form $R(0_i)$ can be added to every such permutation in every possible position without affecting the result of the query.
    \item $(D^r\exo \cup \sigma_f)\not\models \cqrsnott$ and $(D^r\exo \cup \sigma_f \cup\{f\})\not\models \cqrsnott$. To ensure that $(D^r\exo \cup \sigma_f)\not\models \cqrsnott$, none of the facts $R(0_i)$ can appear in $\sigma_f$ (or, otherwise, the pair $R(0_i), S(0_i,\val{0})$ would satisfy the query). Furthermore, The set $F=\{a\mid R(a)\in \sigma_f, a\neq 0_i\}\cup\{b\mid T(b)\in \sigma_f\}$ must be such that $F\in \s(g)$ or, otherwise, there is $S(a,b)\in D^r\exo$ such that $R(a)\in \sigma_f$ and $T(b)\notin \sigma_f$, which implies that $(D\exo\cup\sigma_f)\models\cqrsnt$, in contradiction to our assumption. The number of such permutations is: $P^r_{0\rightarrow0}=\sum_{k=0}^N |\s(g,k)|\cdot k! \cdot (N-k+r)!$.
\end{enumerate}
Hence, we have that: 
\[\shp(D^r, \cqrsnott, f) = 1-\dfrac{P^r_{0\rightarrow0}+P^r_{1\rightarrow1}}{(N+r+1)!}=1-\dfrac{P_{1\rightarrow1}\cdot m_r+P^r_{0\rightarrow0}}{(N+r+1)!}\]
We get that:
\[P^r_{0\rightarrow0}=(1-\shp(D^r, \cqrsnott, f))\cdot(N+r+1)!-P_{1\rightarrow1}\cdot m_r=\]
\[\sum_{k=0}^N |\s(g,k)|\cdot k! \cdot (N-k+r)!\]
(Recall that $P_{1\rightarrow 1}$ can be computed from the Shapley value of the fact $T(\val{0})$ in the instance $D^0$.)
As a consequence, we get a system of $N+1$ equations:

\begin{gather*}
    \left( {\begin{array}{cccc}
   0!(N+1)! & 1!N! & ... & N!1! \\
   0!(N+2)! & 1!(+1)! & ... & N!2! \\
   : & : & : & : \\
   0!(2N+1)! & 1!(2N)! & ... & N!(N+1)!
  \end{array} } \right)
  \left( {\begin{array}{c}
   |\s(g,0)| \\
   |\s(g,1)| \\
   : \\
   |\s(g,N)|
  \end{array} } \right)
  \\ \\
  = \left( {\begin{array}{c}
   (1-\shp(D^1, \cqrsnott, f))\cdot(N+2)!-P_{1\rightarrow1}\cdot m_1 \\
   (1-\shp(D^2, \cqrsnott, f))\cdot(N+3)!-P_{1\rightarrow1}\cdot m_2 \\
   : \\
   (1-\shp(D^{N+1}, \cqrsnott, f))\cdot(2N+2)!-P_{1\rightarrow1}\cdot m_{N+1}
  \end{array} } \right)
\end{gather*}
This is the same system of equations that Livshits et al. obtained in the hardness proof for $\cqrst$~\cite{shapley-icdt2020}. There, they prove that the determinant of the coefficient matrix is not zero; hence, this system is solvable in polynomial time, providing us the with the value of $|\s(g)|=\sum_{k=0}^N |\s(g,k)|$.

Finally, it is left to prove that $|\s(g)|=|\is(g)$. For that purpose, we define a bijection between the two sets, $h: \is(g) \rightarrow \s(g)$, as follows:
Let $(A'\cup B') \in \is(g)$. Then, $h(A'\cup B')=A'\cup (B\setminus B')$. Note that for every $(a,b)\in E$ we have that if $a\in A'$ then $b\not\in B'$; hence, for every $(a,b)\in E$ it holds that if $a\in A'$ then $b\in (B\setminus B')$. Hence, if $(A'\cup B') \in \is(g)$, then $(A'\cup (B\setminus B')) \in \s(g)$.

\underline{Injectivity:} Let $I_1 = (A'_1\cup B'_1)$ and $I_2 = (A'_2\cup B'_2)$ be two distinct independent sets of $g$ (i.e., $I_1,I_2\in\is(g)$). At least one of the following holds: $A'_1\neq A'_2$, or $B'_1\neq B'_2$. Clearly in both cases we have that $h(I_1)\neq h(I_2)$ as well.

\underline{Surjectivity:} Let $E=(A'\cup B')$ be a subset of vertices in $\s(g)$. Consider the subset $I = (A'\cup (B\setminus B'))$. By the definition of $\s(g)$, for every $(a,b)\in E$ we have that if $a\in A'$ then $b\in B'$. Therefore, for every $(a,b)\in E$ it holds that if $a\in A'$ then $b\not\in (B\setminus B')$. Then, we conclude that $I\in\is(g)$ by definition. It holds that $h(I)=(A'\cup (B\setminus(B \setminus B'))=(A'\cup B')=E$; thus, the function $h$ is surjective.

To conclude, we constructed a reduction from the problem of computing $|\is(g)|$ to that of computing $\shp(D,\cqrsnott, f)$; hence, computing $\shp(D,\cqrsnott, f)$ is $\FP{}$-complete. 
\end{proof}

Finally, we show that for any non-hierarchical self-join-free \CQneg $q$ computing $\shp(D, q, f)$ is $\FP$-complete, using a reduction from the problem of computing $\shp(D, q', f)$ where $q'$ is one of queries $\cqrst$,$\cqnrsnt$,$\cqrnst$, or $\cqrsnt$, depending on the polarity of the atoms in the non-hierarchical triplet in $q$.

\begin{lemma}\label{lemma:sjfbcq_hard}
\lemmasjfbcqhard
\end{lemma}
\begin{proof}
Every non-hierarchical self-join-free \CQneg contains three atoms $\alpha_x, \alpha_y, \alpha_{x,y}$ where $x,y\in\var(q)$, such that $\alpha_x\in A_x\setminus A_y$, $\alpha_y\in A_y\setminus A_x$, $\alpha_{x,y}\in A_x\cap A_y$. We argue that $q$ satisfies another property: if there is a non-hierarchical triplet $\alpha_x, \alpha_y, \alpha_{x,y}$ where $\alpha_{x,y}$ and at least one of $\alpha_x$ or $\alpha_y$ are negative, then there is another non-hierarchical triplet $\alpha'_x, \alpha'_y, \alpha'_{x,y}$ where either $\alpha'_{x,y}$ is positive or both $\alpha_x$ and $\alpha_y$ are positive. Assume, without loss of generality, that $\alpha_x$ is negative.
Since $q$ is safe, there is a positive atom $\alpha'_x$ such that $\alpha'_x\in A_x$. If there exists such an atom $\alpha'_x$ such that $\alpha_x'\in (A_x \cap A_y)$, the triplet $\alpha_x, \alpha'_x, \alpha_y$ satisfies the property. Otherwise, if $\alpha_y$ is positive, the triplet $\alpha_x', \alpha_{x,y}, \alpha_y$ satisfies the property. Finally, if every $\alpha_x'$ is such that $\alpha'_x\in A_x\setminus A_y$ and $\alpha_y$ is negative, since $q$ is safe, we have another positive atom $\alpha'_y\in A_y \setminus A_x$ and $\alpha'_x, \alpha_{x,y}, \alpha'_y$ is a non-hierarchical triplet satisfying the property.

Let $\alpha_x, \alpha_{x,y}, \alpha_y$ be a non-hierarchical triplet of $q$ that satisfies the above property. We construct a reduction from the problem of computing $\shp(D',q',f')$ where $q'$ is one of $\cqrst$, $\cqnrsnt$, $\cqrnst$, or $\cqrsnt$ to computing $\shp(D,q,f)$. We have already established that
computing $\shp(D',q',f')$ for each of these queries is $\FP$-complete; hence, we conclude that $\shp(D,q,f)$ is $\FP$-complete for any non-hierarchical self-join-free \CQneg. We present the four reductions simultaneously, as they all work in a very similar way.

Depending on the polarity of the atoms in the non-hierarchical triplet of $q$ satisfying the property indicated above, we select one of the four reductions (if there are multiple triplets satisfying this property, we choose one randomly):
\begin{enumerate}
    \item If all three atoms are positive, we reduce from computing $\shp(D',\cqrst,f')$.
     \item If $\alpha_{x,y}$ is positive while the other two atoms are negative, we reduce from $\shp(D',\cqnotrsnott,f')$.
    \item If $\alpha_{x,y}$ is negative while the other two atoms are positive, we reduce from  $\shp(D',\cqrnotst,f')$.
    \item If $\alpha_{x,y}$ is negative, and exactly one of $\alpha_x, \alpha_y$ is negative, we reduce from $\shp(D',\cqrsnott,f')$.
    \end{enumerate}
    
The idea is very similar to the corresponding proof in~\cite{shapley-icdt2020}. The main difference is in the construction of the database $D'$, as $q$ may contain negative atoms. We use the atom $\alpha_x$ to represent the atom $R(x)$ (or $\neg R(x)$) in $q'$, the atom $\alpha_y$ to represent the atom $T(y)$ (or $\neg T(y)$) in $q'$, and the atom $\alpha_{x,y}$ to represent the atom $S(x,y)$ (or $\neg S(x,y)$) in $q'$. For every fact $R(\val{a})$ in $D$ (which is the input to the first problem), we insert to the relation $R_{\alpha_x}$ in $D'$ (the input to our problem) a fact obtained by mapping the variable $x$ in $\alpha_x$ to $\val{a}$ and the rest of the variables to a constant $\odot$. Similarly, for every fact $T(\val{b})$ in $D$, we insert to the relation $R_{\alpha_y}$ in $D'$ a fact obtained by mapping the variable $y$ in $\alpha_y$ to $\val{b}$ and the rest of the variables to a constant $\odot$. Each such fact $f'\in D'$ will be endogenous if and only if the fact $f$ it was generated from is endogenous. Finally, for every fact $S(\val{a},\val{b})$ and a positive atom $\alpha$ in $q$ that is not one of $\alpha_x$, $\alpha_y$, or $\alpha_{x,y}$, we insert to the relation $R_\alpha$ in $D'$ every exogenous fact obtained by mapping the variable $x$ in $\alpha$ to $\val{a}$, the variable $y$ to $\val{b}$, and the rest of the variables to $\odot$. We also add to the relation $R_{\alpha_{x,y}}$ in $D'$ every exogenous fact obtained by such a mapping of the variables in $\alpha_{x,y}$.

Note that $|D\endo|=|D'\endo|$; hence, the total number of permutations of the endogenous facts is equal for both databases, and we only need to show that for every fact $f\in D\endo$, the number of permutations of the facts in $D\endo$ where $f$ changes the result of $q'$ is equal to the number of permutations of the facts in $D'\endo$ where $f'$ (which is the fact generated from $f$) changes the result of $q$. We can then conclude that $\shp(D,q',f)=\shp(D',q,f')$.

When considering the fact $S(\val{a},\val{b})$ in the construction of $D'$, we have created a mapping $h$ from the variables of $q$ such that $h(x)=\val{a}$, $h(y)=\val{b}$, and $h(w)=\odot$ for the rest of the variables, and we have added all the resulting facts, associated with positive atoms that are not one of $\alpha_x$, $\alpha_y$, or $\alpha_{x,y}$, to $D'$ (as exogenous facts). Moreover, the relations in $D'$ associated with negative atom of $q$ (except $\alpha_{x,y}$) are empty and do not affect the query result. Finally, it holds that $S(\val{a},\val{b})\in D\exo$ if and only if the fact $f$ obtained from it using the atom $\alpha_{x,y}$ is in $D'\exo$. Hence, a subset $E$ of $D\endo$ is such that there is a homomorphism mapping every positive atom and none of the negative atoms of $q'$ to $E\cup D\exo$ (that is, $(E\cup D\exo)\models q'$) if and only if the subset $E'$ of $D'\endo$ that contains for each fact $f\in E$ the corresponding fact $f'\in D'$ is such that there is a homomorphism mapping every positive atom and none of the negative atoms of $q$ to $E'\cup D\exo'$ (that is, $(E'\cup D\exo')\models q$). Therefore, a permutation $\sigma$ of the facts in $D\endo$ and a fact $f\in D\endo$ satisfies $q'(D\exo\cup\sigma_f)\neq q'(D\exo\cup\sigma_f\cup\{f\})$ if and only if the corresponding permutation $\sigma'$ of the facts in $D'\endo$ and the corresponding fact $f'\in D'\endo$ satisfies $q(D'\exo\cup\sigma_{f'})\neq q(D'\exo\cup\sigma_{f'}\cup\{f'\})$, and that concludes our proof.
\end{proof}

Next, we provide an insight into the complexity of the problem for \CQneg{s} with self joins. Theorem~\ref{thm:sjfBCQ} does not provide us with any information about the complexity of computing the Shapley value for the query $\rel{Unemployed}(x),\rel{Married}(x,y),\rel{Unemployed}(y)$ asking whether there is a married couple where both spouses are unemployed, or for the query $\neg \rel{Citizen}(x),\rel{Married}(x,y),\neg \rel{Citizen}(y)$ asking if there are two married people such that none of them is a citizen. The following result implies that computing the Shapley value for both queries is $\FP$-complete.

\begin{theorem}\label{thm:bcqn_sj}
Let $q$ be a polarity-consistent \CQneg containing a non-hierarchical triplet $(\alpha_x,\alpha_{x,y},\alpha_y)$ such that the relation $R_{\alpha_{x,y}}$ occurs only once in $q$. Then, computing $\shp(D,q,f)$ is $\FP$-complete.
\end{theorem}

To prove the theorem, we construct a reduction from the problem of computing $\shp(D,q',f)$ where $q'$ is one of $\cqrst$, $\cqnrsnt$ or $\cqrsnt$ (depending on the polarity of the atoms $\alpha_x$ and $\alpha_y$) under the following assumptions: \e{(1)} all the facts of $S$ are exogenous, and \e{(2)} for every fact $S(\val{a},\val{1})$ in $D$, both facts $R(\val{a})$ and $T(\val{1})$ are in $D$. The instances constructed in the proofs of hardness for all three queries satisfy these conditions; hence, the problems remain hard under these assumptions. We also assume for simplicity that the set of values used in the facts of $R^D$ and the set of values used in the facts of $T^D$ are disjoint.

The idea is very similar to the construction in the proof of Lemma~\ref{lemma:sjfbcq_hard}, with the main difference being the treatment of the negative atoms in $q$. We again use the atom $\alpha_x$ to represent the atom $(\neg)R(x)$ in $q'$, the atom $\alpha_y$ to represent the atom $(\neg)T(y)$ in $q'$, and the atom $\alpha_{x,y}$ to represent the atom $S(x,y)$ in $q'$. We use the assumption that the relation $R_{\alpha_{x,y}}$ occurs only once in $q$ to ensure that we do not create new connections between values of $x$ and values of $y$. If $\alpha_x$ and $\alpha_y$ are both positive, the reduction is from the problem of computing $\shp(D,\cqrst,f)$, if both atoms are negative, the reduction is from computing $\shp(D,\cqnrsnt,f)$, and if one atom is positive while the other is negative, the reduction is from computing $\shp(D,\cqrsnt,f)$.

Formally, given an input database $D$ to the first problem, we build a database $D'$ in the same way we built it in the proof of Lemma~\ref{lemma:sjfbcq_hard} except for the treatment of the atom $\alpha_{x,y}$. Since the atom $S(x,y)$ is always positive in $q'$, if $\alpha_{x,y}$ is negative, then we insert to the relation $R_{\alpha_{x,y}}$ in $D'$ an exogenous fact $f$ obtained by mapping the variables $x$ and $y$ in $\alpha_{x,y}$ to some values $c_1$ and $c_2$ (from the domain of $D'$), respectively, and the rest of the values to $\odot$ if and only if $S(c_1,c_2)\not\in D$. If $\alpha_{x,y}$ is positive, then we insert $f$ to $D'$ if and only if $S(c_1,c_2)\in D$.

We will now prove that for every endogenous fact $f$ in $D$ and its corresponding fact $f'$ in $D'$ (i.e., the fact that was generated from $f$) it holds that $\shp(D,q',f)=\shp(D',q,f')$ (recall that $q'$ is one of $\cqrst$, $\cqnrsnt$, $\cqrsnt$). We start by proving the following. 

\begin{lemma}\label{lemma:yes_to_yes}
Let $E\subseteq D\endo$ and let $E'$ be the set of corresponding facts in $D'\endo$. If $(D\exo\cup E)\models q'$ then $(D'\exo\cup E')\models q$.
\end{lemma}
\begin{proof}
Since $(D\exo\cup E)\models q'$ there is a mapping $h$ from the variables of $q'$ to the domain of $D$ where $h(x)=\val{a}$ for some value $\val{a}$ from the domain of $R^D$ and $h(y)=\val{1}$ for some value $\val{1}$ from the domain of $T^D$ such that $h$ maps every positive atom and none of the negative atoms of $q'$ to a fact of $D\exo\cup E$. We claim that the mapping $h'$ such that $h'(x)=h(x)$, $h'(y)=h(y)$, and $h'(w)=\odot$ for the rest of the variables, maps every positive atom and none of the negative atoms of $q$ to facts of $D'\exo\cup E'$; hence $(D'\exo\cup E')\models q$.

As in the proof of Lemma~\ref{lemma:sjfbcq_hard}, from the construction of $D'$, we have that every positive atom of $q$ is mapped by $h'$ to a fact of $D'\exo\cup E'$. The relations associated with negative atom of $q$, except for $R_{\alpha_x}$, $R_{\alpha_y}$, and $R_{\alpha_{x,y}}$ are empty and do not affect the result of the query (recall that $q$ is polarity-consistent; hence, a relation that appears as a negative atom cannot appear as a positive atom as well). Moreover, the relation $R_{\alpha_{x,y}}$ contains the fact obtained from $\alpha_{x,y}$ using the mapping $h'$ is in $D'\exo$ if and only if $S(\val{a},\val{q})\in D\exo$.

It is only left to show that there is no negative atom of $q$ that is mapped by $h'$ to a fact in $D'\endo$. Let us assume, by way of contradiction, that a negative atom $\beta$ of $q$ is mapped by $h'$ to a fact $f$ in $D'\endo$, and assume, without loss of generality, that $f$ belongs to the relation $R_{\alpha_x}$ in $D'$. Since $q$ is polarity-consistent, the atom $\alpha_x$ is a negative atom as well. Moreover, in this case, the relation $R$ appears as a negative atom in $q'$. From the construction of $D'$, every endogenous fact $f$ in the relation $R_{\alpha_x}$ in $D'$ is obtained by a homomorphism from the variables of $\alpha_x$ to the constants of $f$ that maps the variable $x$ to a value from the domain of $R^D$, the variable $y$ to a value from the domain of $T^D$, and the rest of the variables to $\odot$. If $h'$ maps $\beta$ to $f$, then there is also a homomorphism from $\beta$ to $\alpha_x$ (and from $\alpha_x$ to $\beta$) where $x$ is mapped to itself. We conclude that $h'$ maps the atom $\alpha_x$ to the fact $f$. From the construction of $D'$, we have that $R(\val{a})\in E$, which is a contradiction to the fact that $R$ appears as a negative atom in $q'$ and $(D\exo\cup E)\models q'$.
\end{proof}

Next, we prove the following.
\begin{lemma}\label{lemma:not_to_not}
Let $E\subseteq D\endo$ and let $E'$ be the set of corresponding facts in $D'\endo$. If $(D\exo\cup E)\not\models q'$ then $(D'\exo\cup E')\not\models q$.
\end{lemma}
\begin{proof}
Let us assume, by way of contradiction, that $(D'\exo\cup E')\models q$. Then, there is a mapping $h$ from the variables of $q$ to the domain of $D'$ that maps every positive atom and none of the negative atoms of $q$ to a fact in $D'\exo\cup E'$. In particular, the atom $\alpha_{x,y}$, is mapped to a fact of $D'\exo\cup E'$ if and only if it is positive. From the construction of $D'$ and the uniqueness of the atom (i.e., the fact that its relation does not appear in another atom of $q$), we have that if $\alpha_{x,y}$ is positive, then $h$ maps $\alpha_{x,y}$ to a fact of $D'\exo$ if and only there exists a fact $S(\val{a},\val{1})$ in $D$ such that $h(x)=\val{a}$ and $h(y)=\val{1}$. If $\alpha_{x,y}$ is negative, then $h$ does not map $\alpha_{x,y}$ to a fact of $D'\exo$ if and only if there exists a fact $S(\val{a},\val{1})$ in $D$ such that $h(x)=\val{a}$ and $h(y)=\val{1}$.

We claim that the mapping $h$ is such that every positive atom and none of the negative atoms of $q'$ is mapped to a fact of $D\exo\cup E$, which is a contradiction to the fact that $(D\exo\cup E)\not\models q$. We have already established, that there exists an exogenous fact $S(\val{a},\val{1})$ in $D$ and it holds that $h(x)=\val{a}$ and $h(y)=\val{1}$. It is only left to show that the fact $R(\val{a})$ belongs to $E$ if and only if $R$ occurs as a positive atom in $q'$, and, similarly, the fact $T(\val{1})$ belongs to $E$ if and only if $T$ occurs as a positive atom in $q'$.

If $\alpha_x$ is a positive atom, then there is a fact $f$ in the relation $R_{\alpha_x}$ in $D'$ obtained from $\alpha_x$ by mapping the variable $x$ to the value $\val{a}$ and the rest of the variables to the value $\odot$, such that $f\in E'$. In this case, the relation $R$ also appears in $q'$ as a positive atom and the fact $f'=R(\val{a})$ corresponding to $f$ appears in $E$. If $\alpha_x$ is a negative atom (in which case, the relation $R$ occurs in $q'$ as a negative atom), then the fact $f$ does not appear in $E'$ (or, otherwise, $D'\exo\cup E'$ will not satisfy $q$), which implies that the fact $f'$ does not appear in $E$. We can similarly show that the fact $T(\val{1})$ appears in $E$ if and only if its corresponding fact appears in $E'$, and that concludes our proof.
\end{proof}

Lemmas~\ref{lemma:yes_to_yes} and~\ref{lemma:not_to_not} imply that the fact $f$ changes the result of $q'$ in a permutation $\sigma$ of $D\endo$ if and only if the fact $f'$ changes the result of $q$ in a permutation $\sigma'$ of $D'\endo$. Since the total number of permutations of the facts in $D\endo$ and $D'\endo$ is equal, we conclude that indeed $\shp(D,q',f)=\shp(D',q,f')$.

\section{Details for Section \ref{sec:exo}}

We start by proving the hardness side of the theorem. Let $\signature_X$ be a schema and let $q$ be a self-join-free \CQneg that contains a non-hierarchical path. Similarly to the proof of Theorem~\ref{thm:bcqn_sj}, we construct a reduction from the problem of computing $\shp(D,q',f)$ where $q'$ is one of $\cqrst,\cqnrsnt$ or $\cqrsnt$ to that of computing $\shp(D,q,f)$. We again assume that in the input to the first problem all the facts of $S$ are exogenous, and for every fact $S(\val{a},\val{1})$ in $D$, both facts $R(\val{a})$ and $T(\val{1})$ are in $D$. We also assume that the set of values used in the facts of $R^D$ and the set of values used in the facts of $T^D$ are disjoint.

Since $q$ has a non-hierarchical path, there exist two atoms $\alpha_x$ and $\alpha_y$ in $q$ and two variables $x,y$, such that $R_{\alpha_x}\not\in X$ and $R_{\alpha_y}\not\in X$, the variable $x$ occurs in $\alpha_x$ but not in $\alpha_y$ and the variable $y$ occurs in $\alpha_y$ but not in $\alpha_x$. Moreover, there exists a path $x-v_1-\dots-v_n-y$ in the graph obtained from the Gaifman graph $\graph(q)$ of $q$ by removing every variable in $(\var(\alpha_x)\cup\var(\alpha_y))\setminus \set{x,y}$. The idea is the following. We use the atoms $\alpha_x$ and $\alpha_y$ to represent the atoms $(\neg)R(x)$ and $(\neg)T(y)$ in $q'$, respectively, and we use the non-hierarchical path to represent the connections between them (i.e., the atom $S(x,y)$). If $\alpha_x$ and $\alpha_y$ are both positive, the reduction is from the problem of computing $\shp(D,\cqrst,f)$, if both atoms are negative, the reduction is from computing $\shp(D,\cqnrsnt,f)$, and if one atom is positive while the other is negative, the reduction is from computing $\shp(D,\cqrsnt,f)$.

Formally, given an input database $D$ to the first problem, we build a database $D'$ in the following way. For every fact $f=R(\val{a})$ we assign the value $\val{a}$ to the variable $x$ in $\alpha_x$ and the value $\odot$ to the rest of the variables, and we add the corresponding fact $f'$ to the relation $R_{\alpha_x}$ in $D'$. The fact $f'$ will be endogenous if and only if $f$ is endogenous. Similarly, for every fact $f=T(\val{1})$ we assign the value $\val{1}$ to the variable $y$ in $\alpha_y$ and the value $\odot$ to the rest of the variables, and we add the corresponding fact $f'$ to the relation $R_{\alpha_y}$ in $D'$. Again, the fact $f'$ will be endogenous if and only if $f$ is endogenous.
Next, for every fact $S(\val{a},\val{1})$ in $D$ and atom $\alpha$ in $q$ that is not one of $\alpha_x$ or $\alpha_y$, we assign the value $\val{a}$ to the variable $x$, the value $\val{1}$ to the variable $y$, the value $\langle\val{a},\val{1}\rangle$ to the variables $v_1,\dots,v_n$ along the non-hierarchical path, and the value $\odot$ to the rest of the variables, and we add the corresponding exogenous fact to the relation $R_\alpha$ in $D'$ if we have not added this fact to $D'$ already. Note that $|D\endo|=|D\endo'|$.

Now, given the database $D'$ we construct a database $D''$ which will be the input to our problem in the following way. We first copy all the endogenous facts from $D'$ to $D''$. Then, for every relation $R$ in $D'$ corresponding to a positive atom of $q$, we copy every exogenous fact from $R^{D'}$ to $R^{D''}$.
For every relation $R$ in $D'$ corresponding to a negative atom of $q$, we add to $R^{D''}$ every exogenous fact over the domain of $D'$ if and only if it does not occur in $R^{D'}$ (i.e., $R^{D''}=\overline{R^{D'}}$). Note that since we did not change the endogenous facts, we have that $|D\endo|=|D'\endo|=|D''\endo|$.

We will now prove that for every endogenous fact $f_1$ in $D$ and its corresponding fact $f_2$ in $D''$ it holds that $\shp(D,q',f_1)=\shp(D'',q,f_2)$ (recall that $q'$ is one of $\cqrst$, $\cqnrsnt$, $\cqrsnt$). We start by proving the following. 

\begin{lemma}\label{lemma:yes_to_yes_ex}
Let $E\subseteq D\endo$ and let $E''$ be the set of corresponding facts in $D''\endo$. If $(D\exo\cup E)\models q'$ then $(D''\exo\cup E'')\models q$.
\end{lemma}
\begin{proof}
Since $(D\exo\cup E)\models q'$ there is a mapping $h$ from the variables of $q'$ to the domain of $D$ where $h(x)=\val{a}$ for some value $\val{a}$ from the domain of $R^D$ and $h(y)=\val{1}$ for some value $\val{1}$ from the domain of $T^D$ such that $h$ maps every positive atom and none of the negative atoms of $q'$ to a fact of $D\exo\cup E$. We claim that the mapping $h'$ such that $h'(x)=\val{a}$, $h'(y)=\val{1}$, $h'(z)=\langle\val{a},\val{1}\rangle$ for every variable along the non-hierarchical path, and $h'(w)=\odot$ for the rest of the variables, maps every positive atom and none of the negative atoms of $q$ to facts of $D''\exo\cup E''$; hence $(D''\exo\cup E'')\models q$.

When considering the fact $S(\val{a},\val{1})$ in the construction of $D'$, we have created a mapping $h'$ from the variables of $q$ such that $h'(x)=\val{a}$, $h'(y)=\val{1}$, $h'(z)=\langle\val{a},\val{1}\rangle$ for every variable along the non-hierarchical path, and $h'(w)=\odot$ for the rest of the variables, and we have added all the resulting facts (associated with atoms that are not one of $\alpha_x$ or $\alpha_y$) to $D'$ (as exogenous facts). When constructing $D''$ we have removed every such fact if it was generated from a negative atom of $q$. Hence, $h'$ is a mapping from the the variables of $q_{\setminus\set{\alpha_x,\alpha_y}}$ (which is the query obtained from $q$ by removing the atoms $\alpha_x$ and $\alpha_y$) to the domain of $D''$ such that every positive atom and none of the negative atoms of $q$ appears as a fact in $D\exo''$. Moreover, it holds that $R(\val{a})\in D\endo$ if and only $h'(\alpha_x)\in D''\endo$ and similarly $T(\val{1})\in D\endo$ if and only if $h'(\alpha_y)\in D''\endo$. There, $(D''\exo\cup E'')$ indeed satisfies $q$ and that concludes our proof.
\end{proof}

Next, we prove the following.
\begin{lemma}\label{lemma:not_to_not_ex}
Let $E\subseteq D\endo$ and let $E''$ be the set of corresponding facts in $D''\endo$. If $(D\exo\cup E)\not\models q'$ then $(D''\exo\cup E'')\not\models q$.
\end{lemma}
\begin{proof}
Assume, by way of contradiction, that $D''\exo\cup E''$ satisfies $q$. Hence, there is a mapping $h$ from the variables of $q$ to the domain of $D''$ such that every positive atom and none of the negative atoms of $q$ is mapped into a fact in $D''\exo\cup E''$. 
We now look at the non-hierarchical path $x-v_1-\dots-v_n-y$ in the Gaifman graph of $q$. From the construction of $D'$, every fact $f\in D'$ in a relation corresponding to an atom $\alpha$ that uses both $x$ and $v_1$ is obtained from $\alpha$ by mapping the variable $x$ to some value $c_1$ and the variable $v_1$ to some value $\langle c_1, c_2 \rangle$ such that $S(c_1,c_2)$ is in $D$. If $\alpha$ is positive, then $D''$ also contains only such facts, and if $\alpha$ is negative, then $D''$ does not contain only such facts. Using an atom $\alpha'$ containing the variables $y$ and $v_n$, we can show, in a similar way, that $h(v_n)=\langle d_1,d_2 \rangle$ for some values $d_1,d_2$ such that $S(d_1,d_2)$ is in $D$. Finally, every two consecutive variables $v_i,v_{i+1}$ in the non-hierarchical path occur together in at least one atom $\alpha$ of $q$, and from the construction of $D''$, it holds that if $\alpha$ is positive, then $R_\alpha$ contains only facts where both $v_i$ and $v_{i+1}$ are mapped to the same value, and if $\alpha$ is negative, then these are the only facts that are not in $R_\alpha$; hence, we have that $h(v_i)=h(v_{i+1})$ and we conclude that $c_1=d_1$ and $c_2=d_2$, and the mapping $h$ assigns some value $\langle\val{a},\val{1}\rangle$ to every variable along the non-hierarchical path, such that $S(\val{a},\val{1})$ is in $D$.

When constructing the database $D'$, we have only assigned the value $\langle\val{a},\val{1}\rangle$ to variables if there exists an exogenous fact $S(\val{a},\val{1})$ in $D$. Hence, we have established that such a fact exists in $D$. Moreover, it holds that $R(\val{a})\in E$ if and only if $h(\alpha_x)\in E''$ and similarly $T(\val{1})\in E$ if and only if $h(\alpha_y)\in E''$. In all cases, the restriction of $h$ to the variables $x$ and $y$ maps every positive atom and none of the negative atoms of $q'$ to $D\exo\cup E$; thus, $(D\exo\cup E)\models q'$, which is a contradiction to our assumption.
\end{proof}

The remainder of the proof is rather straightforward based on these two lemmas. The total number of permutations of the facts in $D\endo$ and $D''\endo$ is equal, and the lemmas prove that the number of permutations where $f$ changes the result of $q'$ in $D$ is equal to the number of permutations where it changes the result of $q$ in $D'$; hence, we conclude that $\shp(D,q',f)=\shp(D',q,f)$.

Next, we provide the missing proofs for the lemmas used in the proof of the positive side of Theorem~\ref{thm:exo}. First, we prove that we can replace every negated atom of $q$ corresponding to an exogenous relation of $D$ by a positive atom and the corresponding relation in $D$ by its complement relation, without affecting the Shapley value.

\def\Rbarlemma{
Let $q$ be a self-join-free \CQneg, and let $\alpha\in(\at\exo(q)\cap\negq(q))$. Then, computing $\shp(D,q,f)$ can be reduced to computing $\shp(D',q',f)$, where $q'$ is obtained from $q$ by substituting $\alpha$ with $\overline{\alpha}$, and $D'$ is obtained from $D$ by substituting $R_\alpha^D$ with $\overline{R_\alpha^D}$.}
\begin{lemma}\label{lemma:rbar}
\Rbarlemma
\end{lemma}

\begin{proof} 
Note that the difference between $D$ and $D'$ is restricted to the exogenous facts; thus, we have that $D\endo=D'\endo$. Moreover, for every $E\subseteq D\endo$, it holds that $(D\exo \cup E) \models q$ if and only if $(D'\exo\cup E)\models q'$. This is rather straightforward from the construction of $D'$. If $(D\exo \cup E) \models q$, then there is a homomorphism $h$ from the variables of $q$ to the constants of $D$ that does not map $\alpha$ to any fact of $R_{\alpha}^D$ (since $\alpha$ is a negated atom); hence, it maps $\overline{\alpha}$ to a fact of $\overline{R_{\alpha}^D}$. Every other atom $\beta$ of $q$ also occurs in $q'$ and we have not changed the relations corresponding to other atoms; thus, $h$ maps $\beta$ to a fact of $(D\exo \cup E)$ if and only if it maps $\beta$ to a fact of $(D'\exo\cup E)$, and we conclude that $h$ maps every positive atom and none of the negative atoms of $q'$ to a fact of $(D'\exo\cup E)$.

The proof of the second direction is very similar. If $(D'\exo \cup E) \models q'$, then there is a homomorphism $h$ from the variables of $q'$ to the constants of $D'$ that maps $\overline{\alpha}$ to a fact of $\overline{R_{\alpha}^D}$; hence, it does not map $
\alpha$ to a fact of $R_{\alpha}^D$. Again, since the rest of the atoms are unchanged in $q'$, we conclude that $h$ maps every positive atom and none of the negative atoms of $q$ to a fact of $(D\exo\cup E)$.
We conclude that the total number of permutations is equal in both databases, and the number of permutations where $f$ changes the query result is equal as well; hence, $\shp(D,q,f)=\shp(D',q',f)$.
\end{proof}

Next, we prove that we can combine all the exogenous atoms of $q$ in a connected component of $\exograph(q)$ into a single exogenous atom, without affecting the Shapley value. Recall that from now on we assume (based on Lemma~\ref{lemma:rbar}) that every atom of $q$ corresponding to an exogenous relation of $D$ is positive.

\begin{replemma} {\ref{lemma:joinexo}}
\joinlemma  
\end{replemma}

\begin{proof}
Let $C$ be a connected component of $\exograph(q)$, and let the set $\set{\alpha_1,...,\alpha_k}$ be the set of (exogenous) atoms in $C$. 
Let $q'$ be the query obtained from $q$ by replacing all the atoms of $C$ with a single atom $\alpha_C$, such that $\var(\alpha_C)=\cup_{i\in\set{1,\dots,k}}\var(\alpha_i)$. Observe that since $C$ is a connected component of $\exograph(x)$, none of the exogenous variables occurring in $\alpha_C$ also occurs in another atom of $q'$. Let $D'$ be the database obtained from $D$ by replacing the exogenous relations $R_{\alpha_1},...,R_{\alpha_k}$ with a single exogenous relation $R_{\alpha_C}$ consisting of the set of answers to the query $q_C(\vec{x})\dl \alpha_1,\dots,\alpha_k$ on the database $D$ (where every variable of $\alpha_1,\dots,\alpha_k$ occurs in $\vec{x}$). That is, the facts in the relation $R_{\alpha_C}^{D'}$ are obtained by an inner join between the relations $R_{\alpha_1}^D,\dots,R_{\alpha_k}^D$, where the relations are joined according to the variables of the corresponding atoms.

Since we have only changed the exogenous relations in $D$ to obtain $D'$, we have that $D'\endo = D\endo$. We now prove that for every $E\subseteq D\endo$ it holds that $(D\exo \cup E) \models q$ if and only if $(D'\exo \cup E)\models q'$, which implies that $\shp(D,q,f)=\shp(D',q',f)$ for every endogenous fact $f$.
Let $E\subseteq D\endo$. If $(D\exo \cup E)\models q$, then there is a homomorphism $h$ from $q$ to $D\exo\cup E$. Note that the only negative atoms of $q$ are atoms corresponding to non-exogenous relations; hence, if $h$ does not map any negative atom of $q$ to a fact of $D\exo \cup E$, it also does not map any negative atom of $q'$ to a fact of $D'\exo \cup E$. As for the positive atoms, the homomorphism $h$ maps every positive atom of $q$, and, in particular, the atoms $\alpha_1,\dots,\alpha_k$ of the connected component $C$, to facts of $D\exo \cup E$. Assume that $h(v)=c_v$ for every variable in $\alpha_1,\dots,\alpha_k$. By the definition of $q_C$, the tuple $(c_{v_1},\dots,c_{v_n})$ (where $v_1,\dots,v_n$ are the variables occurring in $\alpha_1,\dots,\alpha_k$) is an answer to $q_C$ and appears in $R_{\alpha_C}^{D'}$. Hence, the atom $\alpha_C$ in $q'$ is mapped to a fact of $D'\exo$. Every positive atom of $q$ that is not one of $\alpha_1,\dots,\alpha_k$ also occurs in $q'$ and it is mapped to a fact of $D\exo \cup E$ that also appears in $D'\exo \cup E$; hence, we conclude that $h$ is a homomorphism from $q'$ to $D'\cup E$.

Similarly, if we assume that $(D'\exo \cup E)\models q'$, then there is a homomorphism $h$ from $q'$ to $D'\exo \cup E$. Every atom $\alpha\in (\at(q)\setminus\set{\alpha_1,\dots,\alpha_k})$ occurs in both $q$ and $q'$, and the relation $R_{\alpha}$ is the same in $D$ and $D'$; thus, every such $\alpha$ is mapped to a fact of $D\exo\cup E$ if and only if it is mapped to a fact of $D'\exo\cup E$. Since every fact in $R_{\alpha_c}^{D'}$ is an answer to $q_C$ on the database $D$, if the atom $\alpha_C$ of $q'$ is mapped by $h$ to a fact $R_{\alpha_C}(c_{v_1},\dots,c_{v_n})$ in $D'\exo$, then every atom $\alpha_i$ in $q$ is mapped by $h$ to a fact $R_{\alpha_i}(c_{v_{i_1}},\dots c_{v_{i_k}})$ in $D$ where $\set{v_{i_1},\dots,v_{i_k}}$ is the set of variables occurring in $\alpha_i$, as if such a fact did not exists, we would never obtain the tuple $(c_{v_1},\dots,c_{v_n})$ as an answer to $q_C$ on $D$. Hence, we have that $(D\exo\cup E)\models q$, as evidenced by $h$. 

The above argument holds for every connected component of $\exograph(q)$; hence, we can replace every connected component with a single atom in $q$ and change the database $D$ accordingly. This will result in a query $q'$ where every exogenous variable occurs exactly once and that concludes our proof. We finish this proof by showing that $q'$ does not have a non-hierarchical path.

Let us assume, by way of contradiction, that the query $q'$ has a non-hierarchical path induced by the atoms $\alpha_x$ and $\alpha_y$. Hence, in the Gaifman graph of $q'$, there is a path $x-v_1-\dots-v_n-y$ that does not pass through the variables of $\alpha_x$ and $\alpha_y$. We claim that there is also a non-hierarchical path induced by $\alpha_x$ and $\alpha_y$ in $q$, in contradiction to the fact that $q$ does not have a non-hierarchical path. Let $v_i,v_{i+1}$ be two consecutive variables in the path. If $v_i$ and $v_{i+1}$ occur together in a non-exogenous atom of $q'$, then they occur together in the same non-exogenous atom of $q$, and $v_i,v_{i+1}$ are also connected in the Gaifman graph of $q$. Otherwise, $v_i,v_{i+1}$ occur together in an exogenous atom of $q'$. This exogenous atom represents a connected component $\set{\alpha_1,\dots,\alpha_k}$ in $\exograph(q)$. Let $\alpha_j$ be the atom where the variable $v_i$ occurs and let $\alpha_r$ be the variable where the variable $v_{i+1}$ occurs. By the definition of $\exograph(q)$, there is a path $u_1-\dots-u_m$ between $\alpha_j$ and $\alpha_r$ such that $u_1\in \alpha_j$, $u_m\in\alpha_r$ and every $u_t$ is an exogenous variable (hence, it does not occur in $\alpha_x$ or $\alpha_y$). We conclude that the Gaifam graph of $q$ contains the path $v_i-u_1-\dots-u_m-v_{i+1}$ that does not pass through the variables of $\alpha_x$ or $\alpha_y$. Therefore, there is a non-hierarchical path between $x$ and $y$ in $q$, and that concludes our proof.
\end{proof}

In the next lemma we will use the following notation. For an atom $\alpha\in\at(q)$, a variable $v\in\var(\alpha)$ and a fact $f\in R_{\alpha}^D$, we denote by $f[v]$ the value of the fact $f$ in the attribute of $R_{\alpha}^D$ corresponding to the position of $v$ in $\alpha$. For example, for $\alpha=R(x,y,z)$, we denote by $f[y]$ the value of the fact $f$ in the second attribute of $R^D$.

\begin{replemma}{\ref{lemma:removeexo}}
\matchingAtomLemma
\end{replemma}

\begin{proof}
As shown in Lemma~\ref{lemma:joinexo}, we can reduce the problem of computing $\shp(D,q,f)$, given $D$ and $f$, to that of computing $\shp(D',q',f)$ where $q'$ is such that every exogenous variable occurs in a single atom of $q'$. This means that every connected component of $\exograph(q')$ contains a single atom. Moreover, we have that $q'$ does not have a non-hierarchical path. For convenience, from now on, we refer to the query $q'$ simply as $q$ and to the database $D'$ simply as $D$, as we do not rely on the original query and database in our proof. We show that we can further reduce the problem of computing $\shp(D,q,f)$ to that of computing $\shp(D',q',f)$, where for every $\alpha\in\at\exo(q')$ there is $\alpha'\in\at\nexo(q')$ such that $\var(\alpha)=\var(\alpha')$. We will do that by first removing the exogenous variables of $q$ and then adding to each exogenous atom all the variables occurring in the non-exogenous atom that ``contains'' it.

Let $\alpha\in\at\exo(q)$. Lemma~\ref{lemma:ccsubset} implies that there exists $\beta\in\at\nexo(q)$ such that $\var\nexo(\alpha)\subseteq\var(\beta)$. We generate the query $q'$ in two steps. First, we remove from $\alpha$ every exogenous variable, and obtain a new atom $\alpha'=R_{\alpha'}(x_1,\dots,x_n)$, where $x_1,\dots,x_n$ are the non-exogenous variables in $\alpha$. Then, we replace the relation $R_{\alpha}$ in $D$ with the relation $R_{\alpha'}$ consisting of the set of answers to the query $q(x_1,\dots,x_n)\dl \alpha$ on $D$. In the next step, we obtain an atom $\alpha''$ by adding to $\alpha'$ every variable in $\var(\beta)\setminus \var\nexo(\alpha)$. That is, if $\set{v_1,\dots,v_m}$ is the set of variables occurring in $\beta$ but not in $\alpha$, then $\alpha''=R_{\alpha''}(x_1,\dots,x_n,v_1,\dots,v_m)$. Then, we obtain the relation $R_{\alpha''}^D$ from $R_{\alpha'}^D$ in the following way. From every $f=R_{\alpha'}(c_1,\dots,c_n)$ in  $R_{\alpha'}^D$, we generate $|\Dom(D)|^m$ facts of the form $f=R_{\alpha''}(c_1,\dots,c_n,d_1,\dots,d_m)$, where $d_1,\dots,d_m\in\Dom(D)$, and add all of them to $R_{\alpha''}^D$. We denote by $q'$ the query obtained from $q$ by replacing the atom $\alpha$ with the atom $\alpha''$, and by $D'$ the database obtained from $D$ by replacing the relation $R_\alpha$ with the relation $R_{\alpha''}$. Note that $D\endo=D'\endo$.

We now prove that for every $E\subseteq D\endo$ it holds that $(D\exo \cup E) \models q$ if and only if $(D'\exo \cup E) \models q'$. 
Let  $E\subseteq D\endo$ such that $(D\exo \cup E) \models q$. Thus, there is a homomorphism $h$ from $q$ do $D\exo\cup E$. In particular, $h$ maps the atom $\alpha$ to a fact $f\in R_{\alpha}^D$. Assume that $h(v)=c_v$ for every variable $v$ of $q$. Hence, for every non-exogenous variable $x_i$ in $\alpha$ we have that $h(x_i)=c_{x_i}$ and for every variable $v_j$ in $\var(\alpha')\setminus \var\nexo(\alpha)$ we have that $h(v_j)=c_{v_j}$. From the construction of $D'$, if $h$ maps the atom $\alpha$ to a fact $f$ in $R_\alpha^D$, there is a fact $R_{\alpha''}(c_{x_1},\dots,c_{x_n},c_{v_1},\dots,c_{v_m})$ in $D'$ and $h$ maps the atom $\alpha''$ in $q'$ to this fact. Every other atom $\beta$ of $q'$ also appears in $q$ and we have not changed any other relation of $D$; hence, $h$ maps $\beta$ to a fact of $D\exo\cup E$ if and only if it maps $\beta$ to a fact of $D'\exo\cup E$. Therefore, $h$ is a homomorphism from $q'$ to $D'\exo \cup E$, as we conclude that $(D\exo'\cup E)\models q'$.

Next, let $E\subseteq D\endo$ such that $(D'\exo \cup E) \models q'$, as evidenced by a homomorphism $h$. Again, every atom $\beta$ of $q'$ that is not $\alpha''$ also appears in $q$, and $h$ maps $\beta$ to a fact of $D\exo\cup E$ if and only if it maps $\beta$ to a fact of $D'\exo\cup E$.
As for the atom $\alpha''$, the homomorphism $h$ maps it to a fact $f'\in R_{\alpha''}^{D'}$. Assume that $f'=R_{\alpha''}(c_{x_1},\dots,c_{x_n},c_{v_1},\dots,c_{v_m})$. From the construction of $D'$, we have that there exists a fact $f$ in $R_\alpha^D$ such that $f[x_i]=c_{x_i}$ for every $i\in\set{1,\dots,n}$. Assume that the exogenous variables in $\alpha$ are $u_1,\dots,u_r$ and $f[u_j]=d_j$ for every $j\in\set{1,\dots,r}$. Then, if we extend the mapping $h$ to a mapping $h'$ such that $h'(x)=h(x)$ for every non-exogenous variable $x$ in $\alpha$ and $h'(u_j)=d_j$ for every exogenous variable $u_j$ in $\alpha$, then $h'$ will map the atom $\alpha$ in $q$ to the fact $f$. Note that this extension does not affect any other atom of $q$ since the exogenous variables of $\alpha$ do not occur in any other atom of $q$. Hence, the mapping $h'$ is a homomorphism from $q$ to $D\exo \cup E$, and $(D\exo \cup E) \models q$.

We can repeat this process for every exogenous atom of $q$ and obtain a query $q'$ satisfying the property of the lemma, such that $\shp(D,q,f)=\shp(D',q',f)$ for every $f\in D\endo$. Finally, we prove that $q'$ does not have a non-hierarchical path. Let us assume, by way of contradiction, that $q'$ has a non-hierarchical path induced by the atoms $\alpha_x$ and $\alpha_y$. Hence, in the Gaifman graph of $q'$, there is a path $x-v_1-\dots-v_n-y$ that does not pass through the variables of $\alpha_x$ and $\alpha_y$. We claim that the same path exists in the Gaifman graph of $q$, in contradiction to the fact that $q$ does not have a non-hierarchical path. Let $v_i,v_{i+1}$ be two consecutive variables in the path. If $v_i$ and $v_{i+1}$ occur together in a non-exogenous atom of $q'$, then they occur together in the same non-exogenous atom of $q$, and $v_i,v_{i+1}$ are also connected in the Gaifman graph of $q$. Otherwise, $v_i,v_{i+1}$ occur together in an exogenous atom of $q'$. Since there are no exogenous variables in $q'$, both $v_i$ and $v_{i+1}$ occur in non-exogenous atoms of $q$. Moreover, since for every exogenous atom $\alpha$ in $q$ there exists a non-exogenous atom $\alpha'$ of $q$ such that $\var(\alpha)=\var(\alpha')$, we again conclude that $v_i$ and $v_{i+1}$ occur together in the same non-exogenous atom of $q$, and $v_i,v_{i+1}$ are also connected in the Gaifman graph of $q$.
\end{proof}

\section{Details for Section~\ref{sec:approx}}

We now prove that the ``gap property'' does not hold for \CQneg{s}.

\begin{reptheorem}{\ref{thm:negation-no-gap}}
\thmgap
\end{reptheorem}
\begin{proof}
Since $q$ is satisfibale, there exists a minimal database $D$ such that $D\models q$. Now, we start adding facts to the relations corresponding to negated atom of $q$, one by one. Clearly, at some point, we will obtain a database $D'$ that does not satisfy the query. Let $f$ be the last fact added to $D'$. We have that $(D'\setminus\set{f})\models q$ but $D'\not\models q$. Let $D_q$ be a minimal database that satisfies this property. We create $n$ copies $D_1,\dots, D_n$ of $D_q$, such that $(\Dom(D_i)\cap \Dom(D_j))=\emptyset$ for all $i,j\in\set{1,\dots,n}$ (this is possible since $q$ has no constants). We denote by $f_i$ the fact for which $(D_i\setminus\{f_i\})\models q$ for every $i\in\set{1,\dots,n}$. Next, let $D_q'$ be a minimal database such that $D_q'\models q$ and let $f'$ be a fact in $D_q'$. Since $D_q'$ is minimal, we have that $(D_q'\setminus\{f'\})\not\models q$. We create $n+1$ copies $D_0,D_{n+1},\dots,D_{2n}$ of $D_q'$, such that $(\Dom(D_i)\cap \Dom(D_j))=\emptyset$ for all $i,j\in\set{0,\dots,2n}$. We again denote by $f_i$ the fact for which $(D_i\setminus\{f_i\})\not\models q$ for every $i\in\set{0,n+1,\dots,2n}$. Finally, we construct a database $D$ by taking the union of the databases $D_0,\dots,D_{2n}$. Every fact in $D$ except for $\set{f_0,\dots,f_{2n}}$ will be exogenous. We will show that the fact $f_0$ does not satisfy the gap property.

Note that $D\exo\models q$ since there is a homomorphism $h$ from $q$ to $D_1\setminus\{f_1\}$ (and, in fact, to every $D_i\setminus\{f_i\}$ for $i\in\set{1,\dots,n}$), and we claim that the same $h$ is a homomorphism from $q$ to $D$. Since $q$ has safe negation, every variable in every negated atom of $q$ also occurs in a positive atom of $q$. Hence, if $h$ maps a negated atom of $q$ to a fact $f\in D\exo$, every value $v$ in $f$ is such that $v\in\Dom(D_1)$, and we have that $f\in (D_1\setminus \{f_1\})$. This is a contradiction to the fact that $h$ is a homomorphism from $q$ to $D_1\setminus\{f_1\}$.

Next, we prove that $(D\exo\cup\{f_1,\dots,f_n\})\not\models q$. Assume, by way of contradiction, that this is not the case. Then, there is a homomorphism $h$ from $q$ to $D\exo\cup\{f_1,\dots,f_n\}$. Assume that $h$ maps the positive atoms of $q$ to the facts $g_1,\dots,g_m$. Clearly, it cannot be the case that $\set{g_1,\dots,g_m}\subseteq D_i$ for some $i\in\set{1,\dots,n}$. This holds true since $D_i\not\models q$; hence, $h$ maps a negated atom of $q$ to a fact of $D_i$ which is also in $D\exo\cup\{f_1,\dots,f_n\}$. Moreover, it cannot be the case that $\set{g_1,\dots,g_m}\subseteq (D_i\setminus\{f_i\})$ for some $i\in\set{0,n+1,\dots,2n}$, since $(D_i\setminus\{f_i\})\not\models q$. Hence, $h$ again maps a negated atom of $q$ to a fact of $D_i\setminus\{f_i\}$ which is also in $D\exo\cup\{f_1,\dots,f_n\}$.

Therefore, the only possible case is that there exist $g_i, g_j$ such that $g_i\in D_k$ and $g_j\in D_r$ for some $k\neq r$. Let $\alpha$ be the positive atom of $q$ mapped by $h$ to $g_i$ and let $\beta$ be the positive atom of $q$ mapped by $h$ to $g_j$. Since $q$ is positively connected, the atoms $\alpha$ and $\beta$ are connected. Thus, there is a path $x-v_1-\dots-v_t-y$ in the Gaifman graph of $q$ from every variable $x$ in $\alpha$ to every variable $y$ in $\beta$ such that all the atoms along the edges of the path are positive. Let $x,y$ be arbitrary such variables. Since $g_i\in D_k$, it holds that $h(x)=c_x$ such that $c_x\in\Dom(D_k)$. Moreover, Since $g_j\in D_r$, it holds that $h(y)=c_y$ such that $c_y\in\Dom(D_r)$. Since every two consecutive variables $v_i,v_{i+1}$ in the path occur together in some positive atom, from the construction of $D$, we have that $h(v_i)$ and $h(v_{i+1})$ are both in $\Dom(D_p)$ for some $p\in\set{0,\dots,2n}$. Let $v_i$ be the first variable in the path such that $h(v_i)\not\in\Dom(D_k)$. There exists such $v_i$ since $h(y)\not\in\Dom(D_k)$. Then, we have that $h(v_{i-1})\in\Dom(D_k)$ and $h(v_{i})\not\in\Dom(D_k)$ and we get a contradiction.

We have established that $(D\exo\cup\{f_1,\dots,f_n\})\not\models q$. We will now prove that $(D\exo\cup E)\models q$ for every $E\subset D\endo$ such that $\{f_1,\dots,f_n\}\not\subseteq E$. Let $f_i\not\in E$ such that $i\in\set{1,\dots,n}$. Let $h$ be a homomorphism from $q$ to $D_i\setminus\{f_i\}$ (recall that $(D_i\setminus\{f_i\})\models q$). We claim that $h$ is a homomorphism from $q$ to $D\exo\cup E$. If this is not the case, then $h$ maps a negated atom of $q$ to a fact of $D\exo\cup E$. Since $q$ has safe negation, if $h$ maps a negated atom of $q$ to a fact $f\in (D\exo\cup E)$, every value $v$ in $f$ is such that $v\in\Dom(D_i)$, and we have that $f\in (D_i\setminus \{f_i\})$. This is a contradiction to the fact that $h$ is a homomorphism from $D_i\setminus \{f_i\}$ to $q$.

Finally, we show that for each $E\subseteq D\endo$, $\set{f_0, f_{n+1},\dots,f_{2n}}\cap E\neq\emptyset$ we have that $(D\exo\cup E)\models q$. Let $f_i\in E$ such that $i\in\set{0,n+1,\dots,2n}$. Since $D_i\models q$, there is a homomorphism $h$ from $q$ to $D_i$. We claim that $h$ is a homomorphism from $q$ to $D\exo\cup E$. Clearly, every positive atom of $q$ is mapped by $h$ to a fact of $D\exo\cup E$ (since $D_i\subseteq(D\exo\cup E)$). Again, since $q$ has safe negation, we have that if $h$ maps a negated atom of $q$ to a fact $f\in (D\exo\cup E)$, then $f\in D_i$ and we get a contradiction to the fact that $h$ is a homomorphism from $q$ to $D_i$.

We conclude that the fact $f_0$ must be added in a permutation before any of the facts $f_{n+1},\dots,f_{2n}$ and after all the facts $f_1,\dots,f_n$ to affect the query result. Hence, there is exactly one subset $E$ of endogenous facts in $D$, containing $n$ facts, that should appear before $f$ in a permutation where it changes the query result from false to true; hence, the number of such permutations is $\frac{n!\cdot n!}{(2n+1)!}$ (as the total number of endogenous facts is $2n+1$).

Finally, since we consider data complexity, we can assume that the size of each $D_i$ is bounded by some constant $k$.
Hence, the database $D$ contains $\theta(n)$ facts. Therefore, we have that $n=\theta(|D|)$ and the Shapley value of $f$ w.r.t.~$q$ and $D$ is $\frac{n!n!}{(2n+1)!}<\frac{1}{2^n}$. Overall, we have that:
\[0<\shp(D,q,f)\leq 2^{-n} = 2^{-\theta(|D|)}\]
and that concludes our proof.

\end{proof}

Next, we prove the following proposition.

\begin{repproposition}{\ref{prop:nsc1}}
\propnscfirst
\end{repproposition}


The following lemma states the NP-completeness of $(2^+,2^-,4^{+-})$-SAT.  (We found it easier to prove it directly rather than showing that it falls on the negative side of Schaefer's dichotomy theorem~\cite{schaefer1978complexity}.)  As a preface to our proof, we define the $(3^+, 2^-)$-SAT problem: given a monotone 3CNF formula $\varphi$ where every literal is positive, and a monotone 2CNF formula $\varphi'$ where every literal is negative, defined over the same variables as $\varphi$, is $\varphi \wedge \varphi'$ satisfiable? We first prove that this problem is NP-complete using a reduction from the 3-colorability problem.\footnote{This proof is inspired by a proof given in~\url{https://cs.stackexchange.com/questions/16634/complexity-of-monotone-2-sat-problem}.}  Next, we define the $(2^+,2^-,4^{+-})$-SAT problem, where the input is a conjunction of clauses of the following forms: \e{(1)} $(x_i\vee x_j)$, \e{(2)} $(\neg x_i\vee \neg x_j)$, or \e{(3)} $(x_i\vee x_j\vee \neg x_k\vee \neg x_l)$. We prove that this problem is NP-complete using a reduction from the $(3^+, 2^-)$-SAT problem. Then, we will construct a reduction from the $(2^+,2^-,4^{+-})$-SAT problem to that of deciding whether $f$ is relevant.

\begin{lemma}
The $(2^+,2^-,4^{+-})$-SAT problem is $NP$-complete.
\end{lemma}
\begin{proof}
 
Given an undirected graph $G=(V,E)$ and a set of three colours $C=\{c_1,c_2,c_3\}$, we will build a $(3^+, 2^-)$-CNF formula denoted as $\varphi$. For every $v\in V$ and every $c_i\in C$, we introduce a variable $x_v^{c_i}$. For every vertex $v\in V$, we introduce a clause $x_v^{c_1}\vee x_v^{c_2} \vee x_v^{c_3}$. For every edge $(u,v)\in E$ and for every $c_i\in C$, we introduce a clause $\neg x_u^{c_i} \vee \neg x_v^{c_i}$. Finally, for every two colours $c_i\neq c_j$ and every $v\in V$, we introduce a clause $\neg x_v^{c_i} \vee \neg x_v^{c_j}$. We say that $G$ has a valid $3$-colouring $h$, if $h$ is a mapping $h:V\rightarrow C$, such that every vertex in $V$ is mapped to a single color in $C$, and every edge $(u,v)$ in $G$ satisfies that $h(u)\neq h(v)$.

The correspondence of the reduction is rather clear. If $G$ has a valid $3$-colouring $h$, the assignment $z$ assigning the value $1$ to every variable $x_v^{c_i}$ such that $h(v)=c_i$ and the value $0$ to the rest of the variables satisfies every clause in $\varphi$. Every clause of the form $(x_v^{c_1}\vee x_v^{c_2} \vee x_v^{c_3})$ is satisfied since $h$ assigns a color to each vertex. Every clause of the form $(\neg x_u^{c_i} \vee \neg x_v^{c_i})$ is satisfied since in a valid colouring, two adjacent vertices cannot be mapped to the same color. In addition, $h$ maps every vertex to a single color in $C$; therefore, each clause of the form $(\neg x_v^{c_i} \vee \neg x_v^{c_j})$ is satisfied as well. Overall, we have that $\varphi$ is satisfiable.

Next, assume that $\varphi$ is satisfiable, and let $z$ be a satisfying assignment. We claim that the coloring $h$ defined by $h(v)=c_i$ if $z(x_c^{v_i})=1$ is a valid $3$-coloring of $G$. Since all the clauses of the form $(x_v^{c_1}\vee x_v^{c_2} \vee x_v^{c_3})$ are satisfied, for every vertex $v$, the assignment $z$ assigns the value $1$ to at least one variable $x_v^{c_i}$. Moreover, since the clauses of the form $(\neg x_v^{c_i} \vee \neg x_v^{c_j})$ are satisfied, for every $v$ the assignment $z$ assigns the value $1$ to at most one variable $x_v^{c_i}$. Hence, we conclude that $z$ assigns the value $1$ to exactly one variable of the form $x_v^{c_i}$ for every $v\in V$, and $h$ is indeed a coloring. Finally, since the clauses of the form $(\neg x_u^{c_i} \vee \neg x_v^{c_i})$ are satisfied, it cannot be the case that $z(x_u^{c_i})=z(x_v^{c_i})=1$; hence, $h$ does not map two vertices connected by an edge in $G$ to the same color.

Next, we reduce the $(3^+, 2^-)$-SAT problem to the $(2^+,2^-,4^{+-})$-SAT problem. Given an input $\varphi$ to the first problem, we build an input $\varphi'$ to the second problem in the following way. Every clause in $\varphi$ of the form $(\neg x_i\vee \neg x_j)$ remains the same in $\varphi'$. Every clause of the form $(x_i\vee x_j \vee x_k)$ in $\varphi$ is replaced by three clauses in $\varphi'$: \e{(1)} $(x_i \vee x_j \vee \neg y \vee \neg y)$, \e{(2)} $(x_k \vee y)$, and \e{(3)} $(\neg x_k \vee \neg y)$, where $y$ is a new unique variable introduced for every clause in $\varphi$. We claim that the clause $(x_i\vee x_j \vee x_k)$ is satisfiable if and only if the formula $(x_i \vee x_j \vee \neg y \vee \neg y)\wedge(x_k \vee y)\wedge (\neg x_k \vee \neg y)$ is satisfiable. This holds true since 
in every satisfying assignment $z$ to the original clause, we either have that $z(x_k)=1$, in which case we satisfy the new formula by defining $z(y)=0$, or we have that $z(x_k)=0$, in which case $z(x_i)=1$ or $z(x_j)=1$ and by defining $z(y)=1$ we again satisfy the new formula. Now, given a satisfying assignment to the new formula, we either have that $z(x_i)=1$ or $z(x_j)=1$ in which case the original formula is clearly satisfied, or we have that $z(y)=0$ in which case $z(x_k)=1$ (otherwise, the clause $(x_k\vee y)$ is not satisfied) and again, the original formula is satisfied. Since we use different variables $y$ for different clauses, we can assign the required value to each one of these variables, and that concludes our proof.
\end{proof}

Next, we give a reduction from $(2^+,2^-,4^{+-})-$-SAT to the problem of deciding whether a fact $f\in T^D$ is relevant to $\cqtrsnr$. Given a formula $\varphi\in(2^+,2^-,4^{+-})$, we build the input database $D$ to our problem as following: for every variable $x_i$ in $\varphi$ we add an endogenous fact $R(i)$, and an exogenous fact $T(i)$ to $D$. For every clause $(x_i\vee x_j)$ in $\varphi$, we add an exogenous fact $S(i,j,a,a)$ where $a$ is a fresh constant. For every clause $(\neg x_i \vee \neg x_j)$ we add an exogenous fact $S(b,b,i,j)$ where $b$ is a new constant. For every $(x_i \vee x_j \vee \neg x_k \neg x_l)$ in $\varphi$ we add an exogenous fact $S(i,j,k,l)$. In addition, we add the exogenous facts $S(d,d,c,c), R(a), R(c), T(a)$ where $c$ and $d$ are fresh constants, and an endogenous fact $T(c)$ which we denote as $f$.

We now show that $\shp(D,q,f)\neq 0$ if and only if $\varphi$ is satisfiable. In fact, we show that $\shp(D,q,f)> 0$ if and only if $\varphi$ is satisfiable, since $T$ appears only as a positive atom in $q$ and $f$ can only be positively relevant to $q$. Observe that $D\exo\models q$ since every (exogenous) fact $S(i,j,a,a)$, along with the exogenous facts $R(a), T(a)$ satisfies $q$. We assume here that every $\varphi$ contains at least one clause of the form $(x_i\vee x_j)$ (otherwise, there is no fact $S(i,j,a,a)$ in $D$). We can assume that since the satisfiability problem is trivial for $(2^+,2^-,4^{+-})$ formulas that do not contain at least one clause of the form $(x_i\vee x_j)$ (as all such formulas are satisfied by the assignment $z$ where $z(x)=0$ for every variable $x$). Hence, the $(2^+,2^-,4^{+-})$-SAT problem remains hard under this assumption.

Assume that $\varphi$ is satisfiable by an assignment $z$, and consider the set $E=\{R(i) \mid z(x_i)=1\}$. We claim that $(D\exo\cup E)\not\models q$. For every exogenous fact $S(i,j,a,a)$, at least one of facts $R(i)$ or $R(j)$ is in $E$, since the clause $(x_i\vee x_j)$ is satisfied; hence, $S(i,j,a,a)$, $T(a)$ and $R(a)$ cannot jointly satisfy $q$. Moreover, for every $S(b,b,i,j)$, at most one of the facts $R(i)$ and $R(j)$ are in $E$, since the clause $(\neg x_i \vee \neg x_j)$ is satisfied as well. Finally, for every $S(i,j,k,l)$, it holds that if both $R(k)$ and $R(l)$ are in $E$, then at least one of $R(i)$ and $R(j)$ is in $E$ as well, since the clause $(x_i \vee x_j \vee \neg x_k \vee \neg x_l)$ is satisfied.  On the other hand, it holds that $(D\exo\cup E \cup f)\models q$, since the facts $S(d,d,c,c), R(c)$ and $T(c)$ are in $(D\exo\cup E \cup f)$ while the fact $R(d)$ is not. Therefore, we conclude that $f$ is relevant to $\cqtrsnr$.

Now, assume that $\varphi$ is not satisfiable. Let $E\subseteq (D\endo\setminus\{f\})$. Recall that the only endogenous facts in $D\endo\setminus\{f\}$ are the facts $R(i)$ for $i\in\set{1,\dots,n}$. We now define the assignment $z$ such that $z(x_i)=1$ if and only if $R(i)\in E$. Since $z$ is not a satisfying assignment, at least one clause $c$ in $\varphi$ is not satisfied. If $c$ is of the form $(x_i\vee x_j)$, then none of $R(i), R(j)$ is in $E$, in which case the exogenous facts $S(i,j,a,a),R(a)$ and $T(a)$ satisfy $q$. If $c$ is of the form $(\neg x_i \vee \neg x_j)$, then both $R(i)$ and $R(j)$ are in $E$, and they satisfy $q$ jointly with the facts $S(b,b,i,j)$ and $T(i)$ (as the fact $R(b)$ is not in $D$). Otherwise, $c$ is of the form $(x_i\vee x_j\vee \neg x_k \vee \neg x_l)$, in which case none of $R(i)$ or $R(j)$ is in $E$, while both $R(k), R(l)$ are in $E$; hence, the facts $S(i,j,k,l),T(k),R(k)$ and $R(l)$ jointly satisfy $q$. In all of these cases, we conclude that $(D\exo \cup E)\models q$; thus, adding $f$ in a permutation after the facts of $E$ would not affect the query result, and $f$ is not relevant to $\cqtrsnr$. This concludes our proof of Proposition~\ref{prop:nsc1}.

We now prove the correctness of $\algname{IsPosRelevant}$ and $\algname{IsNegRelevant}$ for deciding whether a fact is positively or negatively relevant to $q$. We start with $\algname{IsPosRelevant}$ and prove the following.

\begin{lemma}
Let $q$ be a polarity-consistent \CQneg. Then, the algorithm $\algname{IsPosRelevant}(D,q,f)$ returns true, given $D$ and $f$, if and only if $f$ is positively relevant to $q$.
\end{lemma}
\begin{proof}
Assume that $f$ is positively relevant to $q$. Thus, there exists $E\subseteq D\endo$ such that $(D\exo\cup E)\not\models q$ while $(D\exo\cup E\cup\{f\})\models q$. 
Hence, there is a homomorphism $h$ from the variables of $q$ to the constants of $D$ such that every positive atom and none of the negative atoms of $q$ is mapped to a fact of $D\exo\cup E\cup \{f\}$. We claim that the algorithm will return true in the iteration of the for loop when $h$ is selected. By the definition of $h$ we have that $P\subseteq (E\cup\{f\})$, while for every $f'\in N$ it holds that $f'\not\in (E\cup\{f\})$. Moreover, since $(D\exo\cup E)\not\models q$, the homomorphism $h$ maps a positive atom of $q$ to $f$; hence $f\in P$. Since $q$ is polarity consistent, by adding a set of facts corresponding to negative atoms of $q$ we cannot change the query result from false to true. Therefore, the fact that $(D\exo\cup E)\not\models q$ implies that $(D\exo\cup E\cup (\negq_q(D\endo)\setminus N))\not\models q$. Since no fact of $N$ appears in $E$, the set $D\exo\cup(P\setminus\{f\})\cup (\negq_q(D\endo)\setminus N)$ can be obtained from $D\exo\cup E\cup (\negq_q(D\endo)\setminus N)$ by removing a set of facts corresponding to positive atoms of $q$, and, again, since $q$ is polarity consistent, we conclude that $(D\exo\cup(P\setminus\{f\})\cup (\negq_q(D\endo)\setminus N))\not\models q$.

Next, assume that the algorithm returns true. Thus, there exists a mapping $h$ from the variables of $q$ to the constants of $D$ such that $(D\exo\cup(P\setminus\{f\})\cup (\negq_q(D\endo)\setminus N))\not\models q$. Let $E=((P\setminus\{f\})\cup (\negq_q(D\endo)\setminus N))$. We will now show that $(D\exo\cup E\cup\{f\})\models q$ and since $(D\exo\cup E)\not\models q$, this will conclude our proof. By the definition of $N$ and since $h$ does not map any negative atom of $q$ to a fact in $D\exo$, we have that $h$ does not map any negative atom of $q$ to a fact of $D\exo\cup(\negq_q(D\endo)\setminus N)$. Moreover, since $h$ maps every positive atom of $q$ to a fact in $D$, we have that every positive atom of $q$ is mapped by $h$ to a fact in $D\exo\cup P\cup\{f\}$. Therefore, we conclude that $h$ is a homomorphism mapping every positive atom and none of the negative atoms of $q$ to facts of $D\exo\cup(P\setminus \{f\})\cup (\negq_q(D\endo)\setminus N)\cup \{f\}$ and we have that $(D\exo\cup E\cup\{f\})\models q$.
\end{proof}

\begin{algorithm}[t]
\SetAlgoLined
\For{$h:\var(q)\rightarrow\Dom(D)$}{
\If{$h$ maps an atom $\alpha\in\negq(q)$ to some $f'\in D\exo$} {
\textbf{continue}
}
\If{$h$ maps an atom $\alpha\in\posq(q)$ to some $f'\not\in D$} {
\textbf{continue}
}

$P=\{f'\in D\endo\mid h\mbox{ maps an atom }\alpha\in\posq(q)\mbox{ to }f'\}$

$N=\{f'\in D\endo\mid h\mbox{ maps an atom }\alpha\in\negq(q)\mbox{ to }f'\}$

\If{$f\in P$}{
\textbf{continue}
}

\If{$(D\exo\cup P\cup (\negq_q(D\endo)\setminus N)\cup\{f\})\not\models q$} {
\Return true
}

}



\Return false
\caption{\algname{IsNegRelevant}$(D,q,f)$\label{alg:isnegrel}}
\end{algorithm}

We now move on the the algorithm $\algname{IsNegRelevant}$ and prove its correctness. The algorithm is very similar to $\algname{IsPosRelevant}$, except for the fact that we now look for a homomorphism that does not maps any positive atom of $q$ to $f$, and in the last test, we check whether the query is not satisfied by $D\exo\cup P\cup (\negq_q(D)\setminus N)\cup \{f\}$.

\begin{lemma}
Let $q$ be a polarity-consistent \CQneg. Then, the algorithm $\algname{IsNegRelevant}(D,q,f)$ returns true, given $D$ and $f$, if and only if $f$ is negatively relevant to $q$.
\end{lemma}
\begin{proof}
Assume that $f$ is negatively relevant to $q$. Thus, there exists $E\subseteq D\endo$ such that $(D\exo\cup E)\models q$ while $(D\exo\cup E\cup\{f\})\not\models q$. 
Hence, there is a homomorphism $h$ from the variables of $q$ to the constants of $D$ such that every positive atom and none of the negative atoms of $q$ is mapped to a fact of $D\exo\cup E$, while $h$ maps $f$ to a negative atom of $q$. We claim that the algorithm will return true in the iteration of the for loop when $h$ is selected. By the definition of $h$ we have that $P\subseteq E$, while for every $f'\in N$ it holds that $f'\not\in E$. Moreover, since $(D\exo\cup E\cup\{f\})\not\models q$, the homomorphism $h$ does not map any positive atom of $q$ to $f$; hence $f\not\in P$. Since $q$ is polarity consistent, by adding a set of facts corresponding to negative atoms of $q$ we cannot change the query result from false to true. Therefore, the fact that $(D\exo\cup E\cup\{f\})\not\models q$ implies that $(D\exo\cup E\cup \{f\}\cup(\negq_q(D\endo)\setminus N))\not\models q$. Since no fact of $N$ appears in $E$, the set $D\exo\cup P\cup \{f\}\cup(\negq_q(D\endo)\setminus N)$ can be obtained from $D\exo\cup E\cup \{f\}\cup(\negq_q(D\endo)\setminus N)$ by removing a set of facts corresponding to positive atoms of $q$, and, again, since $q$ is polarity consistent, we conclude that $(D\exo\cup P\cup \{f\}\cup(\negq_q(D\endo)\setminus N))\not\models q$.

Next, assume that the algorithm returns true. Thus, there exists a mapping $h$ from the variables of $q$ to the constants of $D$ such that $(D\exo\cup P\cup (\negq_q(D\endo)\setminus N)\cup\{f\})\not\models q$. Let $E=(P\cup (\negq_q(D\endo)\setminus N))$. We will now show that $(D\exo\cup E)\models q$ and since $(D\exo\cup E\cup\{f\})\not\models q$, this will conclude our proof. By the definition of $N$ and since $h$ does not map any negative atom of $q$ to a fact in $D\exo$, we have that $h$ does not map any negative atom of $q$ to a fact of $D\exo\cup(\negq_q(D\endo)\setminus N)$. Moreover, since $h$ maps every positive atom of $q$ to a fact in $D$, we have that every positive atom of $q$ is mapped by $h$ to a fact in $D\exo\cup P$. Therefore, we conclude that $h$ is a homomorphism mapping every positive atom and none of the negative atoms of $q$ to facts of $D\exo\cup P\cup (\negq_q(D\endo)\setminus N)$ and we have that $(D\exo\cup E)\models q$.
\end{proof}

Finally, we prove that the relevance problem is hard for the UCQ $\cqsat$. Recall that $\cqsat()\dl q_1 \vee q_2 \vee q_3 \vee q_4$ where:
\begin{align*}
    q_1()&\dl C(x_1,x_2,x_3,v_1,v_2,v_3), T(x_1,v_1), T(x_2,v_2), T(x_3,v_3)\\
    q_2()&\dl V(x), \neg T(x,\val{1}), \neg T(x,\val{0})\\
    q_3()&\dl T(x,\val{1}), T(x,\val{0})\\
    q_4()&\dl R(\val{0})
\end{align*}

\begin{repproposition}{\ref{prop:ucq_rel}}
\propucqrelhard
\end{repproposition}
\begin{proof}
We construct a reduction from the satisfiability problem for 3CNF formulas.
The input to the satisfiability problem is a formula $\varphi=(c_1\wedge\dots\wedge c_m)$ over the variables $x_1,\dots x_n$, where each $c_i$ is a clause of the form $(l_1\vee l_2\vee l_3)$, and each $l_j$ is either a positive literal $x_k$ or a negative literal $\neg x_k$ for some $k\in\set{1,\dots,n}$. Given such an input, we build an input database $D$ to our problem as follows. For every variable $x_i$ we add an exogenous fact $V(i)$, and two endogenous facts $T(i,\val{1})$ and $T(i,\val{0})$. In addition, for every clause $(l_i \vee l_j \vee l_k)$ where $l_t=x_t$ or $l_t=\neg x_t$ for each $t\in\set{i,j,k}$, we add an exogenous fact $C(i,j,k,v_i,v_j,v_k)$, such that $v_t=\val{1}$ if $l_t=\neg x_t$ and $v_t=\val{0}$ if $l_t=x_t$.
Finally, we add the endogenous fact $R(\val{0})$ which we denote as $f$.
We claim that $f$ is relevant to $q$ if and only if $\varphi$ is satisfiable.

Observe that $E\models q$ for every $E\subseteq D$ such that $f\in E$, since $f$ satisfies the query $q_4$ by itself. Hence, $f$ is relevant (and, more precisely, positively relevant) if and only if there exist $E\subseteq D\endo$ such that $(D\exo\cup E)\not\models q$. Now, assume that $\varphi$ is satisfiable by the assignment $z$. We will show that $f$ is relevant to $q$. Consider the subset $E\subseteq D\endo$ that contains every fact $T(i,v_i)$ such that $z(x_i)=v_i$. Since $z$ is a truth assignment, it assigns a single value to each variable; hence, it is straightforward that $(D\exo\cup E)\not\models q_2$ and $(D\exo\cup E)\not\models q_3$. Regarding the query $q_1$, since $z$ is a satisfying assignment, for every clause in $\varphi$ there is at least one literal $l_i$ such that $z(x_i)=0$ if $l_i=\neg x_i$ and $z(x_i)=1$ if $l_i=x_i$. Therefore, the fact $T(x_i,z(x_i)$ does not appear in $E$, and we have that $(D\exo\cup E)\not\models q$. We conclude that $(D\exo\cup E)\not\models q$ while $(D\exo\cup E\cup\{f\})\models q$ and that concludes our proof of the first direction.

 As for the other direction, given a subset $E\subseteq D\endo$ such that $(D\exo\cup E)\not\models q$ while $(D\exo\cup E\cup\{f\})\models q$, we define an assignment $z$ such that $z(x_i)=1$ if $T(i,\val{1})\in E$ and $z(x_i)=0$ if $T(i,\val{0})\in E$. Since $(D\exo\cup E)\not\models q$, it cannot be the case that $E$ contains two facts $T(i,\val{1}$ and $T(i,\val{0})$ (or, otherwise, $(D\exo\cup E)\models q_3$) and it cannot be the case that none of $T(i,\val{1}$ and $T(i,\val{0})$ belongs to $E$ for some $x_i$ (as otherwise, $(D\exo\cup E)\models q_2$). Hence, $z$ is a truth assignment. It is only left to show that $z$ is a satisfying assignment. Assume, by way of contradiction, the a clause $(l_i,l_j,l_k)$ is not satisfied. In this case, $z(x_t)=0$ if $l_t=x_t$ and $z(x_t)=1$ if $l_t=\neg x_t$ for each $t\in\set{i,j,k}$. Since $E$ contains a fact $T(t,z(x_t))$ for every variable $x_t$, this will imply that the facts $T(i,z(x_i))$, $T(j,z(x_j))$ and $T(k,z(x_k))$ satisfy $q_1$ jointly with the exogenous fact $(i,j,k,z(x_i),z(x_j),z(x_k))$, which is a contradiction to the fact that $(D\exo\cup E)\not\models q$.
 

Since the satisfiablity problem is NP-complete for 3CNF formulas, we conclude that the relevance problem for the given UCQ is NP-complete as well.
\end{proof}